\documentclass[journal]{IEEEtran}

\usepackage[utf8]{inputenc}
\usepackage[english]{babel}
\usepackage{epsfig}
\usepackage{amsfonts}
\usepackage{ifpdf}
\usepackage{multirow}
\usepackage{amssymb}
\usepackage{amsthm}
\usepackage{mathtools}
\usepackage{mathrsfs}
\usepackage{textgreek}
\usepackage{grffile}
\usepackage{xcolor}
\usepackage{bm}
\usepackage{setspace}

\usepackage{cancel}
\usepackage{nccmath}
\usepackage{cuted}

\usepackage[usestackEOL]{stackengine}

  \stackMath

\usepackage{cite}

\ifCLASSINFOpdf
\else
\fi
\usepackage{amsmath}
\usepackage{algorithmic}
\usepackage{algorithm}
\usepackage{graphicx}

\newtheorem{definition}{Definition}
\newtheorem{assumption}{Assumption}
\newtheorem{remark}{Remark}
\newtheorem{lemma}{Lemma}
\newtheorem{corollary}{Corollary}
\newtheorem{theorem}{Theorem}

\ifCLASSOPTIONcompsoc
  \usepackage[caption=false,font=normalsize,labelfont=sf,textfont=sf]{subfig}
\else
  \usepackage[caption=false,font=footnotesize]{subfig}
\fi

\usepackage{stfloats}
\usepackage{etoolbox}

\let\theparentequation\theequation
\patchcmd{\theparentequation}{equation}{parentequation}{}{}

\renewenvironment{subequations}[1][]{
  \refstepcounter{equation}%
  \setcounter{parentequation}{\value{equation}}
  \setcounter{equation}{0}
  \def\theequation{\theparentequation\alph{equation}}%
  \let\parentlabel\label
  \ifx\\#1\\\relax\else\label{#1}\fi
  \ignorespaces
}{%
  \setcounter{equation}{\value{parentequation}}
  \ignorespacesafterend
}

\newcommand*{\nextParentEquation}[1][]{
  \refstepcounter{parentequation}
  \setcounter{equation}{0}
  \ifx\\#1\\\relax\else\parentlabel{#1}\fi
}

\allowdisplaybreaks 

\begin{document}
%


\title{Fundamentals of RIS-Aided Localization in the Far-Field}
\author{Don-Roberts~Emenonye,
        Harpreet~S.~Dhillon,
        and~R.~Michael~Buehrer
\thanks{D.-R. Emenonye, H. S. Dhillon and R. M.  Buehrer are with Wireless@VT,  Bradley Department of Electrical and Computer Engineering, Virginia Tech,  Blacksburg,
VA, 24061, USA. Email: \{donroberts, hdhillon, rbuehrer\}@vt.edu. The support of the US National Science Foundation (Grants ECCS-2030215 and CNS-2107276) is gratefully acknowledged. This work was presented in IEEE ICC Workshop 2023, Rome, Italy \cite{emenonye2023_ICC_conf_workshop}.
}
}
\maketitle
\IEEEpeerreviewmaketitle
\begin{abstract}

This paper develops fundamental bounds for localization in orthogonal frequency division multiplexing (OFDM) systems aided by reconfigurable intelligent surfaces (RISs). Specifically, we start from the assumption that the position and orientation of a RIS can be viewed as prior information for RIS-aided localization in wireless systems and derive Bayesian bounds for the localization of a user equipment (UE). To do this, we first derive the Bayesian Fisher information matrix (FIM) for channel parameters to derive the Bayesian localization bounds. Then, to focus on the geometric channel parameters, we derive the equivalent Fisher information matrix (EFIM) and show that it has a definite structure. Subsequently, we show through the information loss associated with the EFIM that when the RIS reflection coefficients remain constant across all OFDM symbols, and there is no prior information about the nuisance parameters, the corresponding submatrix in the EFIM related to the RIS angle parameters is a zero matrix. As a result of the EFIM being a zero matrix, estimating the RIS-related angle channel parameters is not possible when the RIS reflection coefficients remain constant across all OFDM symbols. This observation is crucial for the estimation of the RIS-related angle parameters. It dictates that to estimate the RIS-related angle parameters, there must be more than one OFDM transmission with differing RIS reflection coefficients. Furthermore, due to this observation, we note that localization of a single antenna UE through the signals received from reflections from a single RIS to the UE is not feasible in the far-field when the RIS reflection coefficients remain constant across all OFDM symbols. We also show that the FIM for the RIS-related channel parameters can be decomposed into 
i) information provided by the receiver, ii) information provided by the transmitter, and iii) information provided by the RIS components. We then transform the Bayesian EFIM for geometric channel parameters to the Bayesian FIM for the UE position and orientation parameters and examine its specific structure under a particular class of RIS reflection coefficients.

\end{abstract}
\section{Introduction}

\IEEEPARstart{W}{ireless} communications systems are conceptualized, designed, and optimized under the assumption that the propagation channels are random and uncontrollable. However, the emerging idea of a reconfigurable intelligent surface (RIS), also known as an intelligent reflecting surface, has challenged this fundamental assumption. An RIS is a novel concept in wireless communications where existing artificial structures such as walls and ceilings of buildings will be equipped with many tightly packed subwavelength-sized  reflecting meta-surfaces. The overall RIS is planar, while each of the metasurfaces in the RIS is software-controlled and are designed to perform a desired transformation on the incoming signal, thereby providing some control over the propagation environment. This control over the wireless propagation environment has been proposed as a means to aid wireless communication systems. It has been shown that with perfect channel state information, the presence of RISs can provide significant gains in energy and bandwidth efficiency when the transmit beamformer and the reflection coefficients at the RIS are jointly optimized \cite{9366805,9779586,8741198,8982186,8811733}. 

Although the initial applications of RIS to wireless systems were limited to RIS-aided communication system designs, it has recently gained attention as a means to improve localization accuracy. The basic idea is to treat RISs as virtual anchors if their positions are known, which is a reasonable assumption for stationary RISs. This assumption is conceptually similar to assuming perfect knowledge of the locations of actual anchors, such as the macro base stations (BSs). Hence, there is the potential of measuring the times of arrival and other geometric channel parameters valuable for positioning with respect to different RISs.  However, in order to fully harness the power of RIS-aided localization, we must first understand its  fundamental limits. These fundamental limits are unknown for the general case in the far-field, which is the focus of this paper.  Thus, in the paper, we derive the Bayesian CRLB and examine its structural properties, which leads to critical insights into RIS-aided localization in the far-field case, including the impact of nuisance parameters, the effects of RIS reflection coefficients,   and the effects of RIS location uncertainty on localization performance. 
\subsection{Related Works}
The following three research directions are of interest to this paper: i)  localization using RISs, ii) localization with large antenna arrays, and iii) Bayesian limits of localization networks and effects of anchor uncertainty. The relevant prior work from these three directions is discussed next. 
\subsubsection{Localization using RISs} It has been shown that due to a large number of geometric channel parameters in RIS-aided wireless systems, the position error bound (PEB) and the orientation error bound (OEB) of a user equipment (UE) can be significantly reduced \cite{keykhosravi2021multi,9528041,fascista2021ris,wymeersch2020beyond}. A challenge with exploiting multiple RISs in localization is identifying which received signals are associated with specific RISs. In \cite{keykhosravi2021multi}, Hadamard matrices have been proposed as a possible solution to this problem as it allows the receiver to extract channel parameters associated with each RIS. In \cite{9528041}, the focus is on the setup where an RIS is attached to each intended receiver and a joint phase design scheme is used to separate the channel parameters associated with each user resulting in submeter UE localization. In \cite{9500281,fascista2021ris,wymeersch2020beyond}, the focus is on using an RIS to enhance the LOS link through coherent combining, thus improving both 3D localization and synchronization. In \cite{9513781}, an RIS is employed to multi-input-multi-output (MIMO) radar to provide the location of multiple targets.  RIS-aided localization without a controlling BS has also been considered \cite{9726785,keykhosravi2022ris}. In \cite{9726785}, two-step positioning is achieved  without a controlling BS by employing a single receive RF chain at each RIS. In \cite{keykhosravi2022ris}, backscatter modulation is used to enable localization in the absence of a controlling BS. A key limitation of these works is their application-specific CRLB analyzes in which the effect of each scenario is abstracted such that it can be approximated through suitable modifications of the existing CRLB results. As a result, the general structure of the Fisher information matrix (FIM) resulting from deploying multiple RISs has not been rigorously investigated. Hence, the primary objective of this paper is to provide a rigorous Bayesian analysis of the FIM.  Moreover, in those prior works \cite{keykhosravi2021multi,9528041,9500281,fascista2021ris,wymeersch2020beyond,9513781,9726785,keykhosravi2022ris}, the RIS reflection coefficients are usually changed across different Orthogonal frequency-division multiplexing
 (OFDM) symbols {in order} to have enough information to estimate the RIS reflection angles. Through the analysis of the Bayesian FIM, we show that changing RIS reflection coefficients across OFDM symbols is in fact necessary since the RIS reflection angles cannot be estimated when the RIS coefficients remain constant across all OFDM symbols.  

For completeness, we also discuss the relevant literature on  RIS-aided localization in the near-field. When the RISs are large enough such that the receivers are considered to be in their near-field, additional information provided by spherical wavefront can be used for localization as shown in\cite{rahal2021ris,9500663,9625826,9508872}. More specifically, a single RIS is attached to a receiver and used to localize a transmitter in \cite{rahal2021ris}, while multiple RISs are used to provide continuous positioning capability by improving the near-field NLOS accuracy in \cite{9500663}. In \cite{9625826}, backscatter modulation is used to empower each RIS element in a single RIS to act as virtual anchors for the time of arrival-based localization. In \cite{9508872}, the RIS-enhanced bounds for $3$D localization in the near-field are provided for the uplink of a system operating both synchronously and asynchronously.

\subsubsection{Localization with Large Antenna Arrays}
The reflection coefficients of the RISs can be designed such that the channel parameters associated with distinct paths can be separated and used as information for localization (see \cite{keykhosravi2021multi,9528041,9500281}), which is conceptually similar to localization enabled by the spatial and temporal resolution offered by large antenna systems operating with large bandwidths \cite{garcia2017direct,8240645,8515231,8356190,guerra2018single,fascista2021downlink,8755880,li2019massive,9082200}. 
Hence, as in massive MIMO-aided localization where the resulting FIM can be diagonalized due to the presence of a large number of antennas (e.g., see \cite{8356190,8515231}), the FIM for RIS-localization can also be diagonalized, {\em albeit} for a different reason (unitary correlation matrix, as will be discussed in detail shortly). 
As a consequence, the existing results on localization with large arrays have the potential of providing useful insights into RIS-aided localization. With this in mind, we now summarize the most relevant literature from this research direction. In \cite{garcia2017direct}, source localization is considered by collecting and processing time of arrival (ToA) and  angle of arrival (AoA) estimates at various distributed BSs. While the ToA estimates are used to restrict the set of possible source positions to a convex set, the AoA estimates are exploited to provide an estimate of the source position. In \cite{8240645}, $2$D position and orientation bounds are derived along with expectation-maximization-based estimation algorithms that achieve these bounds. In \cite{8515231}, these bounds are used to show that under certain conditions, NLOS components provide more information about position and orientation. These bounds are generalized to the $3$D scenario, and various scaling laws are provided for  both the uplink and downlink in \cite{8356190}. More limits are provided for localization using measurements on the uplink of a massive MIMO system in \cite{guerra2018single}. These additional limits indicate that the CRLB is unique in the limit of the number of receive antennas because each possible transmit position leads to distinct observations at the BS. The case of downlink UE positioning with a single antenna receiver is considered in \cite{fascista2021downlink,8755880}. In \cite{li2019massive}, a single anchor is used to estimate UE  trajectory, and the effect of I/Q imbalance on single anchor positioning is considered in \cite{9082200}.  Note that the precoding matrix in large antenna BSs is analogous to the RIS reflection coefficients. In the literature, it is known that when the precoding matrix remains constant across all OFDM symbols, the angle of departure at the BS can still be estimated \cite{8356190}. This is because the large antenna BSs can create multiple spatial streams through the precoding matrix, and these spatial streams can be detected when the UE has more than one receive antenna. However, if the RIS reflection coefficients remain constant across all OFDM symbols, the RIS angles of incidence and reflection can not be estimated (irrespective of the number of UE receive antennas). This is because the passive RIS performs no processing. Hence, it can not generate multiple spatial streams.
\subsubsection{Bayesian Limits of Localization Networks  and Effects of Anchor Uncertainty} In our work we are concerned with the fundamental limits of RIS-enabled localization. A major factor in this analysis is the uncertainty in the exact position/orientation of the RIS. Bayesian approaches to localization can be used to include this uncertainty. Thus, we briefly review relevant work in this area. The Bayesian philosophy to  estimation presented in \cite{kay1993fundamentals,DBLP:books/wi/Trees02} has been applied to the parameter estimation problem in localization, (see \cite{5571900,5571889,5762798,8421288,6655520,8264804,8827486,zekavat2011handbook} for a small subset). In \cite{5571900}, the received waveforms from various anchors are processed, combined with prior information about UE position, and subsequently used to provide localization bounds. This setting is extended to the case of multiple UEs that communicate with each other in \cite{5571889,5762798}, resulting in  cooperative Bayesian localization bounds. In \cite{8421288}, Bayesian limits are presented for network localization which utilizes both position-related parameters and inertial measurement units. In \cite{6655520}, the possibility of improving localization and tracking systems by exploiting prior UE information is investigated. Authors in \cite{8264804,8827486} develop the concept of soft information (SI) for localization. Instead of providing hard decisions on NLOS/LOS and the position of a UE, SI quantifies these decisions through probability distributions. 
While the Bayesian philosophy has been used to generate a posterior distribution that provides localization estimates in the presence of anchor uncertainty \cite{9044729}, a deterministic approach to parameter estimation  in  the presence of anchor uncertainty has also been considered \cite{6862045,7247270,6684554,9186070}.
\subsection{Contributions}
In this paper, we examine the fundamental limits of RIS-aided localization of a UE with a single BS, assuming that one or more RISs provide reflected signals to the UE in addition to an LOS path. We further assume that the UE and the BS are both in the far-field of the RISs, and there are multiple downlink transmissions of OFDM symbols. For this scenario, the critical contributions of this paper are summarized next:
\subsubsection{Bayesian FIM for Channel Parameters} We derive the FIM for the RIS-related channel parameters and show that the FIM can be decomposed into a sum of the FIMs provided by each OFDM symbol. The FIM provided by each OFDM symbol can be further decomposed into: i) information provided by the receiver, ii) information provided by the transmitter, and iii) information provided by the RIS components.  We show that the information provided by the RISs can be further decomposed into a correlation matrix and an information matrix representing the gains due to the RISs. First, through this decomposition, we show that the structure of the FIM can be significantly controlled through the RIS correlation matrix, and this control allows us to investigate the impact of both nuisance parameters and anchor uncertainty on localization performance. Second, through this decomposition, we observe that with identical RIS reflection coefficients across certain OFDM symbols, the information matrix representing the gains due to the RIS coefficients remain constant across those OFDM symbols. In this case of parallel RIS coefficients, the information matrix produced by the RIS coefficients during the transmission of the additional OFDM symbols does not provide any additional information. However, the different OFDM symbol transmissions increases the signal-to-noise-ratio (SNR). 

\subsubsection{Bayesian Equivalent Fisher Information Matrix (EFIM) for Geometric Channel Parameters Provided by the RIS} 
We quantify the information loss associated with the geometric channel parameters along each RIS path due to the unknown nuisance parameters (channel path gains). This quantification is achieved by deriving the Bayesian EFIM for the geometric channel parameters. With this computed information loss, we show that when the RIS coefficients are parallel in time, and there is no prior information about the nuisance parameters, the corresponding submatrix in the EFIM related to the RIS angle parameters is a zero matrix. As a result of the EFIM being a zero matrix, estimating the RIS-related angle parameters is not possible with parallel RIS coefficients. This result is in contrast to angle of departure estimation at a large antenna BS. In large antenna BSs, when the precoding matrix is parallel in time, all the information presented by the EFIM is not necessarily lost due to a lack of knowledge about the nuisance parameters. This is because the large antenna BSs can create multiple spatial streams through the precoding matrix, and these spatial streams can be detected when the UE has more than one receive antenna. This contribution explicitly entails that localization of a single receive antenna UE with reflections from a single RIS is impossible when there is no LOS, and the RIS coefficients are parallel in time.
\subsubsection{Bayesian FIM/EFIM for UE Position and Orientation}
Through a bijective transformation, the FIM of UE position and orientation is obtained from the EFIM of the geometric channel parameters. When paths are separable, the FIM of the UE position and orientation is a sum of the FIM from each of the RIS paths. While any prior information about the UE appears as an additive term in the EFIM, the prior information about the RIS appears in a less simplistic manner. Finally, through Monte-Carlo simulations, we study the effect of the number of receive antennas, the number of RISs, and the number of RIS elements on the localization performance.
\section{System Model}
We consider the downlink of an RIS-assisted single-cell MIMO system consisting of a single BS with $N_T$ antennas, a UE of interest with $N_R$ antennas, and ${M}_1$ distinct RISs. The ${m}^{\text{th}}$ RIS is assumed to contain $N_L^{[{m}]}$ reflecting elements where ${m} \in \mathcal{M}_1 = \{1, 2, \ldots, {M}_1\}$. We further assume OFDM for this transmission. 
 The BS has an arbitrary but known array geometry with its centroid located at $\mathbf{p}_{BS} \in \mathbb{R}^3$. 
The UE is defined by its position $\mathbf{p}_{} \in \mathbb{R}^3$, orientation $(\theta_0, \phi_0)$, and an arbitrary but known array geometry. We use notations $\theta$ and $\phi$, respectively, with appropriate subscripts and superscripts for all the elevation and azimuth angles. The set of RISs is also defined by their positions $\mathbf{p}^{[{m}]} \in \mathbb{R}^3$ and orientation angles $(\theta_0^{[{m}]}, \phi_0^{[{m}]})$, for ${m} \in \mathcal{M}_1$.
\begin{figure}[htb!]
\centering
{\includegraphics[scale=0.2]{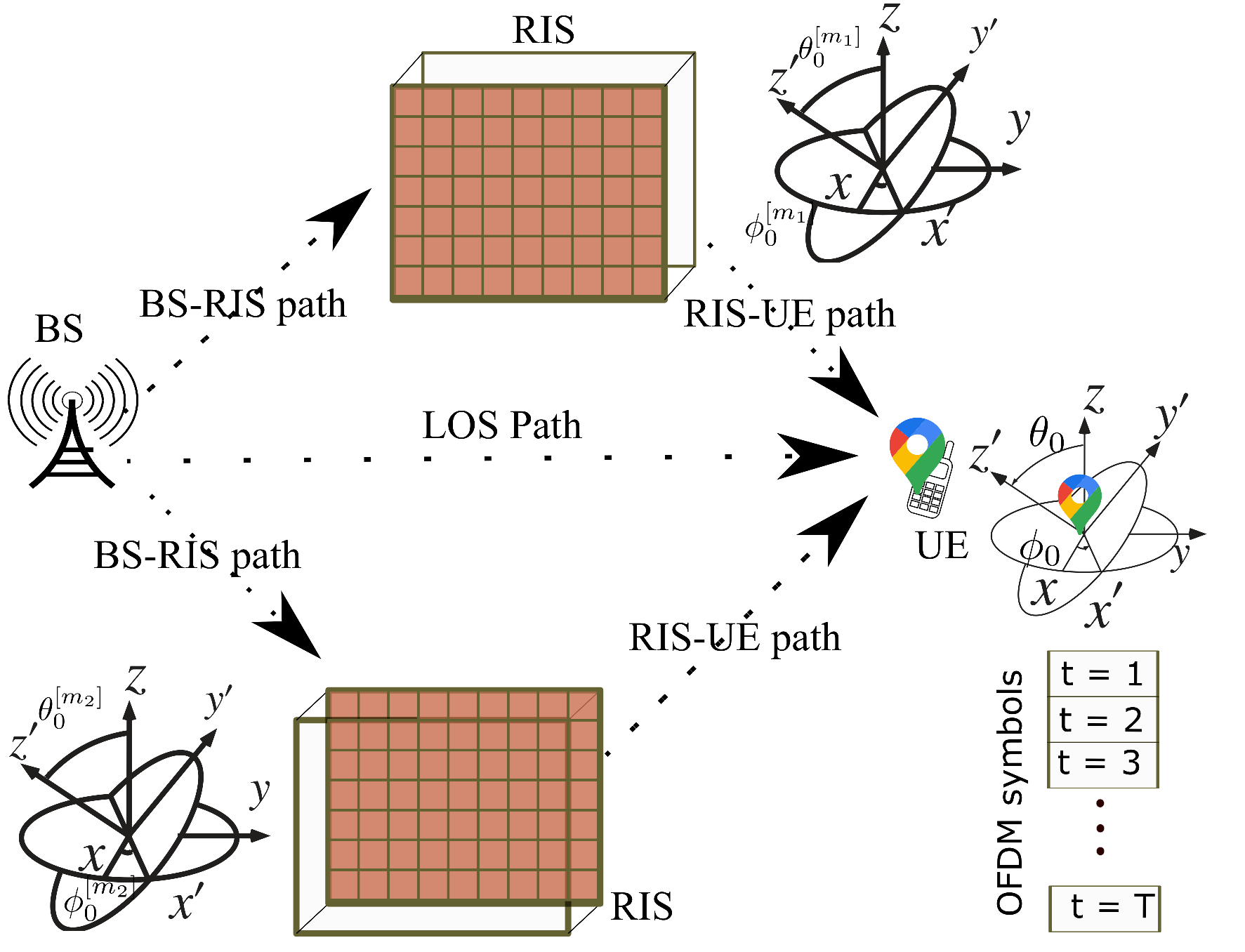}}
\caption{An illustration of the system model.} 
\label{fig:Results/system_model_final}
\end{figure}
The geometry of each RIS is  known but arbitrary.  Without loss of generality, we assume that the arrays can only be rotated in the $x$ and $z$ directions. More specifically, $\phi_{0}$ defines the counter-clockwise rotation of the UE's array around its $z$-axis, and $\theta_{0}$ depicts the clockwise rotation of the UE's array around its $x$-axis. Also, $\phi_{0}^{[m]}$ defines the counter-clockwise rotation of the $m^{\text{th}}$ RIS's array around its $z$-axis, and $\theta_{0}^{[m]}$ depicts the clockwise rotation of the  $m^{\text{th}}$ RIS's array around its $x$-axis.  There are ${M} \geq 3$ paths between the BS and the UE, where the first path $ {m} = 0$ is the LOS path,  a set of paths  $ {m} \in \mathcal{M}_1$ are the virtual LOS paths created by  ${M}_1$ distinct RISs, and there is a set of non-RIS NLOS paths $ \mathcal{M}_2 =  \{{M}_1+1, {M}_1+2, \cdots, {M}_1 + {M}_2 \}$, where ${M} = {M}_1 + {M}_2+1$. The non-RIS NLOS paths (created either by scatterers or reflectors) are usually much weaker compared to the non-RIS LOS paths as well as the virtual LOS paths created by the RISs. Therefore, the non-RIS NLOS paths will not be included in the analysis for notational simplicity.  Further, we partition the set of RISs $\mathcal{M}_1$ into a set with perfectly known position and orientation $\mathcal{M}_1^{a}$ and a set with perturbed position and orientation  $\mathcal{M}_1^{b}$. These sets are mutually exclusive, i.e. $|\mathcal{M}_1| = |\mathcal{M}_1^{a}| + |\mathcal{M}_1^{b}| $.  The number of parameters that needs to be estimated depends on these mutually exclusive sets. 
\subsection{Far-Field Channel Model}
All paths are described in part by their angle of departure (AoD), angle of arrival (AoA), and time of arrival (ToA) as specified by $ (\theta_{\mathrm{t}_{\mathrm{u}}}^{[{m}]}, \phi_{\mathrm{t}_{\mathrm{u}}}^{[{m}]})$,  $(\theta_{\mathrm{r}_{\mathrm{u}}}^{[{m}]}, \phi_{\mathrm{r}_{\mathrm{u}}}^{[{m}]})$, and $\tau^{[{m}]}$, respectively. The array  vector at the transmitter and receiver is specified by 
\begin{equation}
\label{equ:channel_model_array_response}
\begin{aligned}
\mathbf{a}_{\mathrm{t}_{\mathrm{u}}}^{[{m}]}\left(\theta_{\mathrm{t}_{\mathrm{u}}}^{[{m}]}, \phi_{\mathrm{t}_{\mathrm{u}}}^{[{m}]}\right) &\triangleq  e^{-j \Delta_{\mathrm{\mathrm{t}_{\mathrm{u}}}}^{\mathrm{T}} \mathbf{k}\left(\theta_{\mathrm{t}_{\mathrm{u}}}^{[{m}]}, \phi_{\mathrm{t}_{\mathrm{u}}}^{[{m}]}\right)} \in \mathbb{C}^{N_{T}}, 
 \\ 
\mathbf{a}_{\mathrm{r}_{\mathrm{u}}}^{[{m}]}\left(\theta_{\mathrm{r}_{\mathrm{u}}}^{[{m}]}, \phi_{\mathrm{r}_{\mathrm{u}}}^{[{m}]}\right) &\triangleq  e^{-j \Delta_{\mathrm{r}_{\mathrm{u}}}^{\mathrm{T}} \mathbf{k}\left(\theta_{\mathrm{r}_{\mathrm{u}}}^{[{m}]}, \phi_{\mathrm{r}_{\mathrm{u}}}^{[{m}]}\right)} \in \mathbb{C}^{N_{R}},
    \end{aligned}
\end{equation}
respectively, where $\mathbf{k}(\theta, \phi)=\frac{2 \pi}{\lambda}[\cos \phi \sin \theta, \sin \phi \sin \theta, \cos \theta]^{\mathrm{T}}$ is the wavenumber vector, $\lambda$ is the wavelength, $\mathbf{\Delta}_{\mathrm{r}_{\mathrm{u}}} \triangleq\left[\mathbf{u}_{\mathrm{r}_{\mathrm{u}}, 1}, \mathbf{u}_{\mathrm{r}_{\mathrm{u}}, 2}, \ldots, \mathbf{u}_{\mathrm{r}_{\mathrm{u}}, N_{R}}\right], \mathbf{u}_{\mathrm{r}_{\mathrm{u}}, n} \triangleq\left[x_{\mathrm{r}_{\mathrm{u}}, n}, y_{\mathrm{r}_{\mathrm{u}}, n}, z_{\mathrm{r}_{\mathrm{u}}, n}\right]^{\mathrm{T}}$ is a vector of Cartesian coordinates of the $n^{\text {th }}$ receiver element, and $N_{R}$ is the number of receiving antennas. Similarly, parameters $N_{T}, \mathbf{\Delta}_{\mathrm{r}_{\mathrm{u}}}$ and $\mathbf{u}_{\mathrm{\mathrm{t}_{\mathrm{u}}}, n}$ are defined for the transmit vector. 
The array response due to the AoR (angle of reflection) and AoI (angle of incidence) at the ${m}^\text{th}$ RIS can be written as 
\begin{equation}
\label{equ:channel_model_array_response_at_ris}
\begin{aligned}
\begin{aligned}
\mathbf{a}_{\mathrm{t}_{\mathrm{l}}}^{[{m}]}\left(\theta_{\mathrm{t}_{\mathrm{l}}}^{[{m}]}, \phi_{\mathrm{t}_{\mathrm{l}}}^{[{m}]}\right) &\triangleq  e^{-j \Delta_{\mathrm{l,m}}^{\mathrm{T}} \mathbf{k}\left(\theta_{\mathrm{t}_{\mathrm{l}}}^{[{m}]}, \phi_{\mathrm{t}_{\mathrm{l}}}^{[{m}]}\right)} \in \mathbb{C}^{{{N_{L}^{[{m}]}}}}, \\
\mathbf{a}_{\mathrm{r}_{\mathrm{l}}}^{[{m}]}\left(\theta_{\mathrm{r}_{\mathrm{l}}}^{[{m}]}, \phi_{\mathrm{r}_{\mathrm{l}}}^{[{m}]}\right) &\triangleq  e^{-j \Delta_{\mathrm{l,m}}^{\mathrm{T}} \mathbf{k}\left(\theta_{\mathrm{r}_{\mathrm{l}}}^{[{m}]}, \phi_{\mathrm{r}_{\mathrm{l}}}^{[{m}]}\right)} \in \mathbb{C}^{{{N_{L}^{[{m}]}}}},
\end{aligned}
    \end{aligned}
\end{equation}
with $\mathbf{\Delta}_{\mathrm{l,m}} \triangleq\left[\mathbf{u}_{\mathrm{l,m}, 1}, \mathbf{u}_{\mathrm{l,m}, 2}, \ldots, \mathbf{u}_{\mathrm{l,
m}, N_{\mathrm{L}}^{[{m}]}}\right], \text{ where } \mathbf{u}_{\mathrm{l,m}, n} \triangleq\left[x_{\mathrm{l,m}, n}, y_{\mathrm{l,m}, n}, z_{\mathrm{l,m}, n}\right]^{\mathrm{T}}$ is a vector of Cartesian coordinates of the $n^{\text {th }}$ RIS element. 
We also define $\mathbf{X}_{\mathrm{l,m}} = \operatorname{diag}([x_{\mathrm{l,m,1}}, \cdots, x_{\mathrm{l,m,N_{\mathrm{L}}^{[{m}]}}}]^{\mathrm{T}})$. The definitions of $\mathbf{Y}_{\mathrm{l,m}}$ and $\mathbf{Z}_{\mathrm{l,m}}$ are similar.  Assuming synchronization\footnote{To acquire a fundamental understanding of the localization problem, this assumption is common in the literature \cite{8240645,8515231,8356190,fascista2021downlink,8755880,9082200}. Although it appears restrictive, it can be attained using a preliminary two-way synchronization approach or a joint localization       and synchronization approach \cite{9064586,1097833}. These approaches are beyond the scope of this paper.},  the channel at the $n^{\text{th}}$ subcarrier during the $t^{\text{th}}$ OFDM symbol is written as 
$\mathbf{H}_{t}[n]=\sum_{m=0}^{M} \frac{\beta^{[{m}]}}{\sqrt{\rho^{[{m}]}}} \mathbf{H}^{[{m}]}_{t}[n] e^{\frac{-j 2 \pi n \tau^{[{m}]}}{N T_{S}}}  \in \mathbb{C}^{N_{R} \times N_{T}},$
where  $\beta^{[{m}]}$ is the complex channel gain, $\sqrt{\rho^{[{m}]}}$ is the pathloss of the ${m}^{\text{th}}$ path. $\mathbf{H}^{[{m}]}_{t}[n]$  is expressed in (\ref{equ:channel_model_cases_channel}).
\begin{figure*}
\begin{align}
\label{equ:channel_model_cases_channel}
\begin{split}
&\mathbf{H}^{[{m}]}_{t}[n] =     \begin{cases}
      \mathbf{{a}}_{\mathrm{r}_{\mathrm{u}}}^{[{m}]}\left(\theta_{\mathrm{r}_{\mathrm{u}}}^{[{m}]}, \phi_{\mathrm{r}_{\mathrm{u}}}^{[{m}]}\right) {\mathbf{{a}}^{\mathrm{H}}}_{\mathrm{t}_{\mathrm{u}}}^{[{m}]}\left(\theta_{\mathrm{t}_{\mathrm{u}}}^{[{m}]}, \phi_{\mathrm{t}_{\mathrm{u}}}^{[{m}]}\right),  \; \; \;  {m} \in \{ \{0\} \cup  \mathcal{M}_2\},\\
     \mathbf{a}_{\mathrm{r}_{\mathrm{u}}}^{[{m}]}\left(\theta_{\mathrm{r}_{\mathrm{u}}}^{[{m}]}, \phi_{\mathrm{r}_{\mathrm{u}}}^{[{m}]}\right) {\mathbf{{a}}^{\mathrm{H}}}_{\mathrm{t}_{\mathrm{l}}}^{[{m}]}\left(\theta_{\mathrm{t}_{\mathrm{l}}}^{[{m}]}, \phi_{\mathrm{t}_{\mathrm{l}}}^{[{m}]}\right)  \mathbf{\Omega}_{t}^{[{m}]}[n] \mathbf{a}_{\mathrm{r}_{\mathrm{l}}}^{[{m}]}\left(\theta_{\mathrm{r}_{\mathrm{l}}}^{[{m}]}, \phi_{\mathrm{r}_{\mathrm{l}}}^{[{m}]}\right){\mathbf{{a}}^{\mathrm{H}}}_{\mathrm{t}_{\mathrm{u}}}^{[{m}]}\left(\theta_{\mathrm{t}_{\mathrm{u}}}^{[{m}]}, \phi_{\mathrm{t}_{\mathrm{u}}}^{[{m}]}\right), \; \;   {m} \in\mathcal{M}_1.\\
    \end{cases} \\
    \end{split}
    \end{align}
\end{figure*}
\subsection{Transmit Processing}
We consider the transmission of $T$ OFDM symbols each containing  $N$ OFDM subcarriers. The BS precodes a vector of communication symbols $\mathbf{x}^{}[n]=\left[x_{1}[n], \ldots, x_{N_{\mathrm{B}}}[n]\right]^{\mathrm{T}} \in \mathbb{C}^{N_{\mathrm{B}}}$ at the subcarrier level with a directional precoding matrix $\mathbf{F} \in \mathbb{C}^{N_{T} \times N_{\mathrm{B}}} $. After precoding, the symbols are modulated with an $N-$point inverse fast Fourier transform (IFFT). A cyclic prefix of sufficient length $N_{cp}$ is added to the transformed  symbol. In the time domain, this cyclic prefix has length $T_{cp} = N_{cp}T_s$ where $T_s = 1/B$ represents the sampling period. The directional precoding matrix is defined as  $\mathbf{F} \triangleq\left[\mathbf{f}_{1}, \mathbf{f}_{2}, \ldots \mathbf{f}_{N_{\mathrm{B}}}\right]$  where $\mathbf{f}_{\ell}=\frac{1}{\sqrt{N_{\mathrm{B}}}} \mathbf{a}_{\mathrm{t}, b}\left(\theta_{\mathrm{t}, b}^{[l]}, \phi_{\mathrm{t}, b}^{[l]}\right), \quad 1 \leq \ell \leq N_{\mathrm{B}},$ is the beam pointing in the direction $\left(\theta_{\ell}, \phi_{\ell}\right)$ and has the same representation as (\ref{equ:channel_model_array_response}). In order to ensure a power constraint,  we set $\operatorname{Tr}\left(\mathbf{F}^{\mathrm{H}} \mathbf{F}\right)=$
1, and $\mathbb{E}\left\{\mathbf{x}[n] \mathbf{x}^{\mathrm{H}}[n]\right\}=\mathbf{I}_{N_{\mathrm{B}}}$, where $\operatorname{Tr}(\cdot)$ denotes the matrix trace and $\mathbf{I}_{N_{\mathrm{B}}}$ is the $N_{\mathrm{B}}$-dimensional identity matrix.
\subsection{Far-Field RIS Reflection Control}
The reflection coefficients of the ${m}^{\text{th}}$ RIS during the $t^{\text{th}}$ OFDM symbol can be decomposed into

\begin{equation}
\label{equ:far_field_rf_control}
\begin{aligned}
\mathbf{\Omega}_{t = t^{'} + T^{'}q}^{[{m}]} = \gamma_{t^{'}}^{[{m}]} \mathbf{\Gamma}_{q}^{[{m}]}, \; \;  t^{'} \in\{1, \ldots, T^{'}\}, q\in\{1, \ldots N_{Q}\},
    \end{aligned}
\end{equation}
where $T = T^{'}N_{Q}$,  $\gamma_{t^{'}}^{[{m}]}$ and $\mathbf{\Gamma}_q^{[{m}]}$ are termed fast and slow varying RIS coefficients respectively. This is because $\gamma_{t^{'}}^{[{m}]}$ varies across $T^{'}$ OFDM symbols, while $\mathbf{\Gamma}_{q}^{[{m}]}$ is constant across the $q^{\text{th}}$ block of $T^{'}$ OFDM symbols, and there are $N_{Q}$ blocks of OFDM symbols. The fast varying RIS coefficients can be used to orthogonalize the RISs' paths. When $N_{Q} > 1$ and $\mathbf{\Gamma}_q^{[{m}]} \neq \alpha \mathbf{\Gamma}_{q^{'}}^{[{m}]}$, where $\alpha$ is a scalar, the RIS coefficients during the $q^{\text{th}}$ and $(q^{'})^{\text{th}}$ OFDM symbol blocks are non-parallel.   
After the removal of the cyclic prefix and the application of an $N$-point fast Fourier transform (FFT), the received signal at the $n^{\text{th}}$ subcarrier during the $t^{\text{th}}$ OFDM symbol is 
\begin{equation}
\label{equ:receive_processing_Fundamentals}
\begin{aligned}
  \mathbf{r}_{t}[n] = \mathbf{H}_{t}[n] \mathbf{F} \mathbf{x}[n] +  \mathbf{n}_{t}[n],
\end{aligned}
\end{equation}
where $\mathbf{n}_{t}[n] \sim \mathcal{C}\mathcal{N}(0,N_0)$ is the Fourier transform of the thermal noise local to the UE's antenna array at the $n^{\text{th}}$ subcarrier during the $t^{\text{th}}$ OFDM symbol and $\mathbf{x}[n]$ are pilots transmitted.
\subsection{Far-Field Receive Processing}
As stated above already, we assume that the only non-RIS path is a single LOS path,  hence the signals received at the $t^{\text{th}}$ OFDM symbol can be written as 
\begin{equation}
\label{equ:far_field_receive_processing}
\begin{aligned}
  \mathbf{r}_{t}[n] &=   \frac{\beta^{[\mathrm{0}]}}{\sqrt{\rho^{[\mathrm{0}]}}} \mathbf{H}^{[\mathrm{0}]}_{t}[n] e^{\frac{-j 2 \pi n \tau^{[\mathrm{0}]}}{N T_{S}}} \mathbf{F} \mathbf{x}[n] \\ &+   \sum_{m = 1}^{{M}_1} \frac{\beta^{[{m}]}}{\sqrt{\rho^{[{m}]}}} \mathbf{H}^{[{m}]}_{t}[n] e^{\frac{-j 2 \pi n \tau^{[{m}]}}{N T_{S}}} \mathbf{F} \mathbf{x}[n] +  \mathbf{n}_{t}[n].
\end{aligned}
\end{equation}
The first term corresponds to the LOS path and the second term corresponds to the RIS paths. To facilitate any subsequent derivations, we also write the received signal as
\begin{equation}
\label{equ:far_field_receive_processing_2}
\begin{aligned}
  \mathbf{r}_{t}[n] =   \bm{\mu}^{}_{t}[n] + \mathbf{n}_{t}[n] , 
 \quad t=1,2, \ldots, T, \quad n=1,2, \ldots, N ,
\end{aligned}
\end{equation}
where $\bm{\mu}^{}_{t}[n]$ is the noise-free part of $\mathbf{r}_{t}[n]$. Based on the received signal in (\ref{equ:far_field_receive_processing_2}), the vectors of the unknown channel parameters related to the RIS paths are defined as 
$$
\begin{aligned}
\begin{array}{llll}
\bm{\theta}_{\mathrm{r}_{\mathrm{u}}} \triangleq\left[\theta_{\mathrm{r}_{\mathrm{u}}}^{[1]},  \ldots, \theta_{\mathrm{r}_{\mathrm{u}}}^{[{M_1}]}\right]^{\mathrm{T}}, &
\bm{\phi}_{\mathrm{r}_{\mathrm{u}}}  \triangleq\left[\phi_{\mathrm{r}_{\mathrm{u}}}^{[1]}, \ldots, \phi_{\mathrm{r}_{\mathrm{u}}}^{[{M_1}]}\right]^{\mathrm{T}}, \\
\bm{\theta}_{\mathrm{t}_{\mathrm{l}}} \triangleq\left[\theta_{\mathrm{t}_{\mathrm{l}}}^{[1]},  \ldots, \theta_{\mathrm{t}_{\mathrm{l}}}^{[{M_1}]}\right]^{\mathrm{T}}, 
&\bm{\phi}_{\mathrm{t}_{\mathrm{l}}}  \triangleq\left[\phi_{\mathrm{t}_{\mathrm{l}}}^{[1]},  \ldots, \phi_{\mathrm{t}_{\mathrm{l}}}^{[{M_1}]}\right]^{\mathrm{T}}, \\
\bm{\theta}_{\mathrm{r}_{\mathrm{l}}}  \triangleq\left[\theta_{\mathrm{r}_{\mathrm{l}}}^{[1]}, \ldots, \theta_{\mathrm{r}_{\mathrm{l}}}^{[{M_1}]}\right]^{\mathrm{r}}, &
\bm{\phi}_{\mathrm{r}_{\mathrm{l}}}  \triangleq\left[\phi_{\mathrm{r}_{\mathrm{l}}}^{[1]}, \ldots, \phi_{\mathrm{r}_{\mathrm{l}}}^{[{M_1}]}\right]^{\mathrm{r}}, \\
\bm{\theta}_{\mathrm{t}_{\mathrm{u}}}  \triangleq\left[\theta_{\mathrm{t}_{\mathrm{u}}}^{[1]},\ldots, \theta_{\mathrm{t}_{\mathrm{u}}}^{[{M_1}]}\right]^{\mathrm{T}}, 
&
\bm{\phi}_{\mathrm{t}_{\mathrm{u}}}  \triangleq\left[\phi_{\mathrm{t}_{\mathrm{u}}}^{[1]},  \ldots, \phi_{\mathrm{t}_{\mathrm{u}}}^{[{M_1}]}\right]^{\mathrm{T}}, \\
\bm{\beta}  \triangleq\left[\beta^{[1]},  \ldots, \beta^{[{M_1}]}\right]^{\mathrm{T}}, &
\bm{\tau}_{}  \triangleq\left[\tau^{[1]}_{},  \ldots, \tau^{[{M_1}]}_{}\right]^{\mathrm{T}}.
\end{array}
\end{aligned}
$$ Hence, the unknown channel parameters can be represented by the vector 
\begin{equation}
\label{equ:FIM_parameter_vector}
\begin{aligned}
 \bm{\eta} \triangleq\left[\bm{\theta}_{\mathrm{r}_{\mathrm{u}}}^{\mathrm{T}}, \bm{\phi}_{\mathrm{r}_{\mathrm{u}}} ^{\mathrm{T}}, \bm{\theta}_{\mathrm{t}_{\mathrm{l}}}^{\mathrm{T}}, \bm{\phi}_{\mathrm{t}_{\mathrm{l}}}^{\mathrm{T}},
 \bm{\theta}_{\mathrm{r}_{\mathrm{l}}}^{\mathrm{T}},
 \bm{\phi}_{\mathrm{r}_{\mathrm{l}}}^{\mathrm{T}},
 \bm{\theta}_{\mathrm{t}_{\mathrm{u}}}^{\mathrm{T}},
 \bm{\phi}_{\mathrm{t}_{\mathrm{u}}}^{\mathrm{T}},
 \bm{\tau}^{\mathrm{T}}_{},
  \bm{\beta}_{\mathrm{R}}^{\mathrm{T}}, \bm{\beta}_{\mathrm{I}}^{\mathrm{T}}\right]^{\mathrm{T}},
\end{aligned}
\end{equation}
where $\bm{\beta}_{\mathrm{R}} \triangleq \Re\{\bm{\beta}\}$, and $\bm{\beta}_{\mathrm{I}} \triangleq \Im\{\bm{\beta}\}$ are the real and imaginary parts of $\bm{\beta}$, respectively.
Also, the unknown channel parameters related to the LOS path are written as
$
\begin{array}{llll}
\bm{\psi}_{} \triangleq\left[\theta_{\mathrm{r}_{\mathrm{u}}}^{[0]}, \phi_{\mathrm{r}_{\mathrm{u}}}^{[0]}, \theta_{\mathrm{t}_{\mathrm{u}}}^{[0]}, \phi_{\mathrm{t}_{\mathrm{u}}}^{[0]}, \tau^{[\mathrm{0}]}_{},   {\beta}_{\mathrm{R}}^{\mathrm{[0]}}, {\beta}_{\mathrm{I}}^{\mathrm{[0]}}  \right]^{\mathrm{T}}.
\end{array}
$ Finally, the unknown channel parameters related to the LOS plus RIS paths can be written as $\bm{\zeta}_{} \triangleq\left[ \bm{\psi}_{}^{\mathrm{T}},  \bm{\eta}_{}^{\mathrm{T}} \right]^{\mathrm{T}}.$

\begin{definition}
\label{definition_FIM_1_fundamentals}

Based on a set of observations $\mathbf{r}$, the Bayesian Fisher information of a parameter vector, $\bm{\eta}_{}$,  is written as
\begin{equation}
\label{equ:definition_FIM_1}
\begin{aligned}
\mathbf{J}_{ \bm{\eta}_{}} &\triangleq 
\mathbb{E}_{\mathbf{r}, \bm{\eta}_{}}\left[-\frac{\partial^{2} \ln \chi(\mathbf{r}_{}[n];  \bm{\eta}_{} )}{\partial \bm{\eta}_{} \partial \bm{\eta}_{}^{\mathrm{T}}}\right] \\
&= -\mathbb{E}_{\bm{\eta}}\left[\mathbb{E}_{\mathbf{r} \mid \mathbf{\eta}}\left[\frac{\partial^{2} \ln \chi(\mathbf{r}_{}[n]|  \bm{\eta}_{} )}{\partial \bm{\eta}_{} \partial \bm{\eta}_{}^{\mathrm{T}}}\right]\right] -\mathbb{E}_{ \bm{\eta}_{}}\left[\frac{\partial^{2} \ln \chi(  \bm{\eta}_{} )}{\partial \bm{\eta}_{} \partial \bm{\eta}_{}^{\mathrm{T}}}\right] \\ &= \mathbf{J}_{ \bm{\eta}_{}}^{\mathrm{D} } + \mathbf{J}_{ \bm{\eta}_{}}^{\mathrm{P} }
\end{aligned}
\end{equation}
where $\mathbb{E}_{\nu}$ is expectation taken over the random variable $\nu$, $\chi$ is the probability density function (PDF), $ \mathbf{J}_{ \bm{\eta}_{}}^{\mathrm{D} }$, and $ \mathbf{J}_{ \bm{\eta}_{}}^{\mathrm{P} }$ are the FIMs related to the likelihood and the prior, respectively.

\end{definition}

\begin{definition}
\label{definition_EFIM_2}
If the FIM of a parameter $\bm{\eta} = [\bm{\eta}_1^{\mathrm{T} }  \; \; \bm{\eta}_2^{\mathrm{T} } ]^{\mathrm{T} }$ is specified by 

\begin{equation}
\label{equ:definition_EFIM_2}
\begin{aligned}
\mathbf{J}_{\bm{\eta}}=\left[\begin{array}{cc}
\mathbf{J}_{\bm{\eta}_1\bm{\eta}_1} & \mathbf{J}_{\bm{\eta}_1\bm{\eta}_2} \\
\mathbf{J}_{\bm{\eta}_1\bm{\eta}_2}^{\mathrm{T}} & \mathbf{J}_{\bm{\eta}_2\bm{\eta}_2}
\end{array}\right]
    \end{aligned}
\end{equation}

where $\bm{\eta} \in \mathbb{R}^{N}, \bm{\eta}_{1} \in \mathbb{R}^{n}, \mathbf{J}_{\bm{\eta}_1\bm{\eta}_1}\in \mathbb{R}^{n \times n}, \mathbf{J}_{\bm{\eta}_1\bm{\eta}_2}\in \mathbb{R}^{n \times(N-n)}$, and $\mathbf{J}_{\bm{\eta}_2\bm{\eta}_2}\in$ $\mathbb{R}^{(N-n) \times(N-n)}$ with $n<N$, 
then the EFIM \cite{5571900} of  parameter ${\bm{\eta}_{1}}$ is given by 
\begin{equation}
\label{equ:definition_EFIM_1}
\begin{aligned}
\mathbf{J}_{\bm{\eta}_{1}}^{\mathrm{e}}&=\mathbf{J}_{{\bm{\eta}_{1}} {\bm{\eta}_{1}}}- \mathbf{J}_{{\bm{\eta}_{1}} {\bm{\eta}_{1}}}^{nu} =\mathbf{J}_{{\bm{\eta}_{1}} {\bm{\eta}_{1}}}-
\mathbf{J}_{{\bm{\eta}_{1}} {\bm{\eta}_{2}}} \mathbf{J}_{{\bm{\eta}_{2}}{\bm{\eta}_{2}} }^{-1} \mathbf{J}_{{\bm{\eta}_{1}} {\bm{\eta}_{2}}}^{\mathrm{T}}
    \end{aligned}
\end{equation}
Note that the term $\mathbf{J}_{{\bm{\eta}_{1}} {\bm{\eta}_{1}}}^{nu} = \mathbf{J}_{{\bm{\eta}_{1}} {\bm{\eta}_{2}}} \mathbf{J}_{{\bm{\eta}_{2}}{\bm{\eta}_{2}} }^{-1} \mathbf{J}_{{\bm{\eta}_{1}} {\bm{\eta}_{2}}}^{\mathrm{T}}$ describes the loss of information about ${\bm{\eta}_{1}}$  due to uncertainty in the nuisance parameters ${\bm{\eta}_{2}}$.
\end{definition}
\section{Fisher Information for RIS Paths }
\label{section:Fisher_Information_for_RIS_Paths}
 We define the geometric channel parameters $ \bm{\eta}_{1} \triangleq\left[\bm{\theta}_{\mathrm{r}_{\mathrm{u}}}^{\mathrm{T}}, \bm{\phi}_{\mathrm{r}_{\mathrm{u}}} ^{\mathrm{T}}, \bm{\theta}_{\mathrm{t}_{\mathrm{l}}}^{\mathrm{T}}, \bm{\phi}_{\mathrm{t}_{\mathrm{l}}}^{\mathrm{T}},
 \bm{\theta}_{\mathrm{r}_{\mathrm{l}}}^{\mathrm{T}},
 \bm{\phi}_{\mathrm{r}_{\mathrm{l}}}^{\mathrm{T}},
 \bm{\theta}_{\mathrm{t}_{\mathrm{u}}}^{\mathrm{T}},
 \bm{\phi}_{\mathrm{t}_{\mathrm{u}}}^{\mathrm{T}},
 \bm{\tau}^{\mathrm{T}}_{}\right]^{\mathrm{T}}$
and  the nuisance parameter as $ \bm{\eta}_{2} \triangleq\left[\bm{\beta}_{\mathrm{R}}^{\mathrm{T}},
\bm{\beta}_{\mathrm{I}}^{\mathrm{T}}
\right]^{\mathrm{T}}
$. 
To derive the FIM of $\bm{\eta}_{}$, we define the PDF as
$\chi(\mathbf{r}_{t}[n];  \bm{\eta}_{} ) = \chi(\mathbf{r}_{t}[n] |  \bm{\eta}_{} ) \chi(\bm{\eta}_{} ),$
where $\chi(\bm{\eta}_{} ) = \chi(\bm{\eta}_{1}) \chi( \bm{\eta}_{2})$.
The FIM of the channel parameters due to the observation $\mathbf{r}$ is an $ 11 {M}_1 \times 11 {M}_1$ matrix which can be viewed as a collection of ${M}_1 \times {M}_1$ submatrices 
\begin{equation}
\label{equ:FIM_parameter_matrix}
\begin{aligned}
\mathbf{J}_{\eta}^{\mathrm{D}} \triangleq\left[\begin{array}{cccc}
\mathbf{J}_{\bm{\theta}_{\mathrm{r}_{\mathrm{u}}} \bm{\theta}_{\mathrm{r}_{\mathrm{u}}}} & \mathbf{J}_{\bm{\theta}_{\mathrm{r}_{\mathrm{u}}} \bm{\phi}_{\mathrm{r}_{\mathrm{u}}}} & \cdots & \mathbf{J}_{\bm{\theta}_{\mathrm{r}_{\mathrm{u}}} \bm{\beta}_{\mathrm{I}}} \\
\mathbf{J}_{\bm{\theta}_{\mathrm{r}_{\mathrm{u}}} \bm{\phi}_{\mathrm{r}_{\mathrm{u}}}}^{\mathrm{T}} & \ddots & \cdots & \vdots \\
\vdots & \cdots & \ddots & \vdots \\
\mathbf{J}_{\bm{\theta}_{\mathrm{r}_{\mathrm{u}}} \bm{\beta}_{\mathrm{I}}}^{\mathrm{T}} & \cdots & \cdots & \mathbf{J}_{\bm{\beta}_{\mathrm{I}} \bm{\beta}_{\mathrm{I}}}
\end{array}\right]
\end{aligned}
\end{equation}
in which $\mathbf{J}_{\bm{\eta}_{{\mathbf{v}_1}} \bm{\eta}_{{\mathbf{v}_2}}} \triangleq \frac{2}{\sigma^2} \sum_{n = 1}^{N} \sum_{t=1}^{{T}} \Re\left\{\frac{\partial \bm{{\mu}^{}_{t}[n]}^{\mathrm{H}}}{\partial {\bm{\eta}_{}}_{\mathbf{v}_1 }} \frac{\partial \bm{{\mu}^{}_{t}[n]}^{{}}}{\partial {\bm{\eta}_{}}_{\mathbf{v}_2 }}\right\}$, where $\bm{\eta}_{\mathbf{v}_1} $, $\bm{\eta}_{\mathbf{v}_2} $ are both dummy variables that stand for the parameters of interest, and $1/\sigma^2$ is the SNR which incorporates the pathloss and composite noise power. To continue and allow for a compact representation of the FIM, we define the following terms relating to the AoA at the UE
\begin{subequations}[equ:FIM_parameter_terms_los_a]
\begin{align}
{\mathbf{K}}_{\mathrm{r}_{\mathrm{u}}}^{[{m}]} & \triangleq \operatorname{diag}\left(\frac{\partial}{\partial \theta_{\mathrm{r}_{\mathrm{u}}}^{[{m}]}} \Delta_{\mathrm{r}_{\mathrm{u}}}^{\mathrm{T}} \mathbf{k}\left(\theta_{\mathrm{r}_{\mathrm{u}}}^{[{m}]}, \phi_{\mathrm{r}_{\mathrm{u}}}^{[{m}]}\right)\right), \\
{\mathbf{P}}_{\mathrm{r}_{\mathrm{u}}}^{[{m}]}  & \triangleq \operatorname{diag}\left(\frac{\partial}{\partial \phi_{\mathrm{r}_{\mathrm{u}}}^{[{m}]}} \Delta_{\mathrm{r}_{\mathrm{u}}}^{\mathrm{T}} \mathbf{k}\left(\theta_{\mathrm{r}_{\mathrm{u}}}^{[{m}]}, \phi_{\mathrm{r}_{\mathrm{u}}}^{[{m}]}\right)\right), \label{equ:FIM_parameter_terms_los_1} \\
\mathbf{K}_{\mathrm{r}_{\mathrm{u}}} & \triangleq\left[{\mathbf{K}}_{\mathrm{r}_{\mathrm{u}}}^{[1]} \mathbf{a}_{\mathrm{r}_{\mathrm{u}}}^{[1]},{\mathbf{K}}_{\mathrm{r}_{\mathrm{u}}}^{[2]} \mathbf{a}_{\mathrm{r}_{\mathrm{u}}}^{[2]}, \ldots, {\mathbf{K}}_{\mathrm{r}_{\mathrm{u}}}^{[{M_1}]} \mathbf{a}_{\mathrm{r}_{\mathrm{u}}}^{[{M_1}]}\right], \\
\mathbf{P}_{\mathrm{r}_{\mathrm{u}}}  &\triangleq\left[{\mathbf{P}}_{\mathrm{r}_{\mathrm{u}}}^{[1]} \mathbf{a}_{\mathrm{r}_{\mathrm{u}}}^{[1]},{\mathbf{P}}_{\mathrm{r}_{\mathrm{u}}}^{[2]} \mathbf{a}_{\mathrm{r}_{\mathrm{u}}}^{[2]}, \ldots, {\mathbf{P}}_{\mathrm{r}_{\mathrm{u}}}^{[{M_1}]} \mathbf{a}_{\mathrm{r}_{\mathrm{u}}}^{[{M_1}]}\right], \label{equ:FIM_parameter_terms_los_2} \\
\mathbf{A}_{\mathrm{r}_{\mathrm{u}}} & \triangleq\left[\mathbf{a}_{\mathrm{r}_{\mathrm{u}}}^{[1]}, \mathbf{a}_{\mathrm{r}_{\mathrm{u}}}^{[2]}, \ldots, \mathbf{a}_{\mathrm{r}_{\mathrm{u}}}^{[{M_1}]}\right] \label{equ:FIM_parameter_terms_los_3}. 
\end{align}
\end{subequations}
The corresponding terms for the AoD at the BS can be obtained by replacing $r$ with $t$ in (\ref{equ:FIM_parameter_terms_los_1}) and (\ref{equ:FIM_parameter_terms_los_2}). In addition to replacing $\Delta_{\mathrm{r}_{\mathrm{u}}}$ with $\Delta_{\mathrm{l,m}}$, the term related to the elevation AoI defined as ${\mathbf{K}}_{\mathrm{r}_{\mathrm{l}}}^{[{m}]}$ can be obtained by replacing $\theta_{\mathrm{r}_{\mathrm{u}}}^{[{m}]} $ with $\theta_{\mathrm{r}_{\mathrm{l}}}^{[{m}]} $.  A similar term for the azimuth AoI defined as ${\mathbf{P}}_{\mathrm{r}_{\mathrm{l}}}^{[{m}]}$  can be obtained by swapping $\phi_{\mathrm{r}_{\mathrm{u}}}^{[{m}]}$ with $\phi_{\mathrm{r}_{\mathrm{l}}}^{[{m}]}$. The corresponding terms related to the AoR can be obtained from the terms related to the AoI by replacing $r$ with $t$. Other terms related to the AoI and AoR at the RISs during the $q^{\text{th}}$ OFDM symbol are 
\begin{equation}
\label{equ:FIM_parameter_terms_ris_2}
\begin{aligned}
\Tilde{\mathbf{a}}_{{q},{\mathrm{r}_{\mathrm{l}}}}^{[{m}]}  &\triangleq  {\mathbf{\Gamma}}_{q}^{[{m}]} {\mathbf{a}}_{\mathrm{r}_{\mathrm{l}}}^{[{m}]},  \;
\Tilde{\mathbf{a}}_{q,k_{\mathrm{r}_{\mathrm{l}}}}^{[{m}]}  \triangleq  {\mathbf{K}}_{\mathrm{r}_{\mathrm{l}}}^{[{m}]} \; \Tilde{\mathbf{a}}_{{q},{\mathrm{r}_{\mathrm{l}}}}^{[{m}]}, \; {\mathbf{a}}_{k_{\mathrm{t}_{\mathrm{l}}}}^{[{m}]}  \triangleq  {\mathbf{K}}_{\mathrm{t}_{\mathrm{l}}}^{[{m}]} {\mathbf{a}}_{\mathrm{t}_{\mathrm{l}}}^{[{m}]},  \\
\Tilde{\mathbf{a}}_{q,p_{\mathrm{r}_{\mathrm{l}}}}^{[{m}]} & \triangleq  {\mathbf{P}}_{\mathrm{r}_{\mathrm{l}}}^{[{m}]} \Tilde{\mathbf{a}}_{{q},{\mathrm{r}_{\mathrm{l}}}}^{[{m}]}, \;  {\mathbf{a}}_{p_{\mathrm{t}_{\mathrm{l}}}}^{[{m}]}  \triangleq  {\mathbf{P}}_{\mathrm{t}_{\mathrm{l}}}^{[{m}]} {\mathbf{a}}_{\mathrm{t}_{\mathrm{l}}}^{[{m}]}, 
\end{aligned}
\end{equation}
\begin{equation}
\label{equ:FIM_parameter_terms_ris_2b}
\begin{aligned}
\mathbf{k}_{q,\mathrm{l}}  & \triangleq \mathbf{p}_{q,\mathrm{l}}  \triangleq \left[\mathbf{{a}}^{\mathrm{H}[1]}_{\mathrm{t}_{\mathrm{l}}} \Tilde{\mathbf{a}}_{q,\mathrm{r}_{\mathrm{l}}}^{[1]},\mathbf{{a}}^{\mathrm{H}[2]}_{\mathrm{t}_{\mathrm{l}}} \Tilde{\mathbf{a}}_{q,\mathrm{r}_{\mathrm{l}}}^{[2]}, \ldots, \mathbf{{a}}^{\mathrm{H}[M_1]}_{\mathrm{t}_{\mathrm{l}}} \Tilde{\mathbf{a}}_{q,\mathrm{r}_{\mathrm{l}}}^{[{M_1}]}\right]^{\mathrm{H}}, \\
\mathbf{k}_{q,\mathrm{t}_{\mathrm{l}}}  & \triangleq\left[\mathbf{{a}}^{\mathrm{H}[1]}_{k_{\mathrm{t}_{\mathrm{l}}}} \Tilde{\mathbf{a}}_{q,\mathrm{r}_{\mathrm{l}}}^{[1]},\mathbf{{a}}^{\mathrm{H}[2]}_{k_{\mathrm{t}_{\mathrm{l}}}} \Tilde{\mathbf{a}}_{q,\mathrm{r}_{\mathrm{l}}}^{[2]}, \ldots, \mathbf{{a}}^{\mathrm{H}[M_1]}_{k_{\mathrm{t}_{\mathrm{l}}}} \Tilde{\mathbf{a}}_{q,\mathrm{r}_{\mathrm{l}}}^{[{M_1}]}\right]^{\mathrm{H}}, 
\\ \mathbf{p}_{q,\mathrm{t}_{\mathrm{l}}}&\triangleq\left[\mathbf{{a}}^{\mathrm{H}[1]}_{p_{\mathrm{t}_{\mathrm{l}}}} \Tilde{\mathbf{a}}_{q,\mathrm{r}_{\mathrm{l}}}^{[1]},\mathbf{{a}}^{\mathrm{H}[2]}_{p_{\mathrm{t}_{\mathrm{l}}}}  \Tilde{\mathbf{a}}_{q,\mathrm{r}_{\mathrm{l}}}^{[2]}, \ldots, \mathbf{{a}}^{\mathrm{H}[M_1]}_{p_{\mathrm{t}_{\mathrm{l}}}}  \Tilde{\mathbf{a}}_{q,\mathrm{r}_{\mathrm{l}}}^{[{M_1}]}\right]^{\mathrm{H}}, 
\end{aligned}
\end{equation}
\begin{equation}
\label{equ:FIM_parameter_terms_ris_2c}
\begin{aligned}
\mathbf{k}_{q,\mathrm{\mathrm{r}_{\mathrm{l}}}}  & \triangleq\left[\mathbf{{a}}^{\mathrm{H}[1]}_{\mathrm{t}_{\mathrm{l}}} \Tilde{\mathbf{a}}_{q,k_{\mathrm{r}_{\mathrm{l}}}}^{[1]},\mathbf{{a}}^{\mathrm{H}[2]}_{\mathrm{t}_{\mathrm{l}}} \Tilde{\mathbf{a}}_{q,k_{\mathrm{r}_{\mathrm{l}}}}^{[2]}, \ldots, \mathbf{{a}}^{\mathrm{H}[M_1]}_{\mathrm{t}_{\mathrm{l}}} \Tilde{\mathbf{a}}_{q,k_{\mathrm{r}_{\mathrm{l}}}}^{[{M_1}]}\right]^{\mathrm{H}}, \\
\mathbf{p}_{q,\mathrm{\mathrm{r}_{\mathrm{l}}}}   &\triangleq\left[\mathbf{{a}}^{\mathrm{H}[1]}_{\mathrm{t}_{\mathrm{l}}} \Tilde{\mathbf{a}}_{q,p_{\mathrm{r}_{\mathrm{l}}}}^{[1]},\mathbf{{a}}^{\mathrm{H}[2]}_{\mathrm{t}_{\mathrm{l}}} \Tilde{\mathbf{a}}_{q,p_{\mathrm{r}_{\mathrm{l}}}}^{[2]}, \ldots, \mathbf{{a}}^{\mathrm{H}[M_1]}_{\mathrm{t}_{\mathrm{l}}} \Tilde{\mathbf{a}}_{q,p_{\mathrm{r}_{\mathrm{l}}}}^{[{M_1}]}\right]^{\mathrm{H}}. 
\end{aligned}
\end{equation} 
The scalar part of the RIS reflecting coefficients used for multipath separation is arranged as
$\bm{\gamma}^{[{m}]} = [{\gamma}^{[{m}]}_{1}, {\gamma}^{[{m}]}_{2}, \ldots, {\gamma}^{[{m}]}_{T^{'}}]^{\mathrm{T}},$
and can be arranged in a  matrix as
$\mathbf{D}_{\gamma} = [\bm{\gamma}^{[1]}, \bm{\gamma}^{[2]}, \ldots, \bm{\gamma}^{[{M_1}]}].$
This matrix, henceforth referred to as a {\em sequence matrix} provides control in both spatial and temporal domains through the fast-varying part of reflecting coefficients of the RISs. The complex channel gains are arranged in a diagonal matrix
$
\mathbf{B}  \triangleq \operatorname{diag}(\bm{\beta})
$
and the signal factor representing the effect of the transmitted beams  is specified by 
$\left[\mathbf{R}_{k}\right]_{u v} \triangleq \sum_{n = 1}^{N}\left(2 \pi n /\left(N T_{s}\right)\right)^{k} \mathbf{x}[n] \mathbf{x}^{\mathrm{H}}[n] e^{-j 2 \pi n \frac{\tau^{[v]}-\tau^{[u]}}{N T_{s}}}, $
where $k \in \{0,1,2\}$. 
\subsection{Entries of the FIM}

\begin{lemma}
\label{lemma:entries_in_FIM}
The entries in the FIM in (\ref{equ:FIM_parameter_matrix}), which represents the amount of information about the parameter vector, $\bm{\eta}$, present in the received signals, $\mathbf{r}_{t}[n], 
 \quad t=1,2, \ldots, T, \quad n=1,2, \ldots, N,$ have a definite  structure. This structure is presented in (\ref{equ:FIM_general_form}).
 \begin{figure*}
\begin{align}
\label{equ:FIM_general_form}
\mathbf{J}_{\mathbf{v}_1 \mathbf{v}_{2}} &= \frac{2}{\sigma^2}\sum_{q = 1}^{N_{Q}}\Re\{(\text {Rx factor}) \odot( q^{\text{th}} \text { RIS gain}) \odot(\text {RIS correlation}) \odot(\text {Tx factor}) \odot(\text {signal factor})\} = \frac{2}{\sigma^2}\sum_{q = 1}^{N_{Q}}\mathbf{J}_{q,\mathbf{v}_1 \mathbf{v}_{2}}.
\end{align}
\end{figure*}
where $\mathbf{v}_{1}, \mathbf{v}_{2} \in \bm{\eta}$ and $\odot$ represents element-wise matrix multiplication. The FIM, $\mathbf{J}_{q,\mathbf{v}_1 \mathbf{v}_{2}}$, is the contribution due to the $q^{\text{th}}$ OFDM symbol block (see (\ref{equ:far_field_rf_control})). From this equation, we notice that the FIM, $\mathbf{J}_{\mathbf{v}_1 \mathbf{v}_{2}}$, is a summation of the contributions from the OFDM symbol blocks. Focusing on  the submatrices of the FIM in (\ref{equ:FIM_parameter_matrix}) that are related to AoA, we have 
\begin{equation}
\label{equ:FIM_exact_submat_theta_r_u_theta_r_u}
\begin{aligned}
\mathbf{J}_{\bm{\theta}_{\mathrm{r}_{\mathrm{u}}} \bm{\theta}_{\mathrm{r}_{\mathrm{u}}}}
 &=   \frac{2}{\sigma^2}\sum_{q = 1}^{N_{Q}}\Re\left\{
 \left(
 \mathbf{B}^{\mathrm{H}} \mathbf{K}_{\mathrm{r}_{\mathrm{u}}}^{\mathrm{H}} \mathbf{K}_{\mathrm{r}_{\mathrm{u}}} \mathbf{B}\right) \odot\left(\mathbf{k}_{q,\mathrm{l}} \mathbf{k}_{q,\mathrm{l}}^{\mathrm{H}}\right)
 \odot\left(\mathbf{D}_{\mathrm{\gamma}}^{\mathrm{H}} \mathbf{D}_{\mathrm{\gamma}}\right)
 \right. \\
 &\left. \odot \left(\mathbf{A}_{\mathrm{t}_{\mathrm{u}}}^{\mathrm{H}}  \mathbf{F}   \mathbf{F}^{\mathrm{H}} \mathbf{A}_{\mathrm{t}_{\mathrm{u}}}^{}\right)^{\mathrm{T}} \odot \mathbf{R}_{0}\right\}.
\end{aligned}
\end{equation}
The information contribution from the $q^{\text{th}}$ OFDM symbol block can be decomposed into: i) information provided by the receiver specified by some combination of the terms $\{ \mathbf{P}_{\mathrm{r}_{\mathrm{u}}}, \mathbf{K}_{\mathrm{r}_{\mathrm{u}}}, \mathbf{B}, \mathbf{A}_{\mathrm{r}_{\mathrm{u}}} \}$, ii) information provided by the transmitter specified by some combination of the terms $\{ \mathbf{P}_{\mathrm{t}_{\mathrm{u}}}, \mathbf{K}_{\mathrm{t}_{\mathrm{u}}}, \mathbf{F}, \mathbf{A}_{\mathrm{t}_{\mathrm{u}}} \}$, iii) information provided by the RIS gain specified by some combination of the terms $\{ \mathbf{k}_{q,\mathrm{l}}, \mathbf{k}_{q,\mathrm{t}_{\mathrm{l}}}, \mathbf{k}_{q,\mathrm{r}_{\mathrm{l}}}, \mathbf{p}_{q,\mathrm{t}_{\mathrm{l}}}, \mathbf{p}_{q,\mathrm{r}_{\mathrm{l}}}  \}$, iv) the correlation across the RIS specified by the product $\mathbf{D}_{\gamma}^{\mathrm{H}} \mathbf{D}_{\gamma}$, and v) the transmit signal factor. 
\end{lemma}
\begin{proof}
The proof follows after taking the first derivative, and using the definitions in the previous sections.
\end{proof}
In general, all submatrices of the FIM in (\ref{equ:FIM_parameter_matrix}) can be written similarly. These equations are presented in Appendix \ref{appendix:FIM_entries}.
\begin{corollary}
\label{corollary:parallel_FIM}
    When the slow varying RIS coefficients are parallel across the $q^{\text{th}}$ and $(q^{'})^{\text{th}}$ OFDM symbol blocks, i.e. when  $\mathbf{\Gamma}_q^{[{m}]} = \alpha \mathbf{\Gamma}_{q^{'}}^{[{m}]}$, where $\alpha$ is a scalar, the Fisher information contributions across both blocks are identical, i.e. $\mathbf{J}_{q,\mathbf{v}_1 \mathbf{v}_{2}} = \mathbf{J}_{q^{'},\mathbf{v}_1 \mathbf{v}_{2}}$. In this case of parallel RIS coefficients, the RIS coefficients of the $(q^{'})^{\text{th}}$ OFDM symbol blocks do not provide any additional information through the RIS coefficients. However, the additional transmission of the $(q^{'})^{\text{th}}$ OFDM symbol blocks increases the  SNR.
\end{corollary}
\begin{proof}
The proof follows by inspecting Lemma \ref{lemma:entries_in_FIM}, and  by observing that the FIMs are always positive definite.
\end{proof}
Corollary \ref{corollary:parallel_FIM} describes the effect on the Fisher information of having parallel RIS coefficients across multiple OFDM symbol blocks. From this corollary, we notice that the additional OFDM symbol blocks only increase the SNR in the case of parallel RIS coefficients.  
\subsection{Bayesian FIM}
To incorporate any prior knowledge about the unknown but random channel parameters, the Bayesian FIM is also analyzed. The channel parameters are random variables because both the RISs and the UE positions and orientations are viewed as random variables. The channel parameters, $\bm{\eta}_{}$, are assumed to be independent of each other such that the  prior Fisher information matrix, $\mathbf{J}_{ \bm{\eta}_{}}^{\mathrm{P}}$, is an  $11{M}_1 \times 11 {M}_1$  diagonal matrix. The Bayesian FIM of the channel parameters $ \bm{\eta}_{}$  is also an $ 11{M}_1 \times 11 {M}_1$ matrix which contains several ${M}_1 \times {M}_1$ submatrices  such that its entries are written as  
\begin{equation}
    \label{equ:bayesian_FIM_1}
\begin{aligned}
    \Tilde{\mathbf{J}}_{\mathbf{v}_1 \mathbf{v}_2 } &= {\mathbf{J}}_{\mathbf{v}_1 \mathbf{v}_2 } + {\mathbf{J}}_{\mathbf{v}_1 \mathbf{v}_2 }^{\mathrm{P}}, \\
 &= \frac{2}{\sigma^2}\sum_{q = 1}^{N_{Q}}\mathbf{J}_{q,\mathbf{v}_1 \mathbf{v}_{2}} + {\mathbf{J}}_{\mathbf{v}_1 \mathbf{v}_2 }^{\mathrm{P}},
 \end{aligned}
 \end{equation}
where $\mathbf{v}_{1} , \mathbf{v}_{2} \in \bm{\eta}$, $\mathbf{v}_{1} = \mathbf{v}_{2}$, and
\begin{equation}
    \label{equ:bayesian_FIM_2}
\begin{aligned}
\Tilde{\mathbf{J}}_{\mathbf{v}_1 \mathbf{v}_2 } = {\mathbf{J}}_{\mathbf{v}_1 \mathbf{v}_2 } = \frac{2}{\sigma^2}\sum_{q = 1}^{N_{Q}}\mathbf{J}_{q,\mathbf{v}_1 \mathbf{v}_{2}},
 \end{aligned}
 \end{equation}
where $\mathbf{v}_{1} , \mathbf{v}_{2} \in \bm{\eta}$ and $\mathbf{v}_{1} \neq \mathbf{v}_{2}$.

\subsection{Entries in the EFIM with a Unitary RIS Correlation Matrix}
In this subsection, we analyze the structure of the EFIM when the correlation matrix is a unitary matrix. More specifically, we analyze the structure of the FIM when $\mathbf{D}_{\mathrm{\gamma}}^{\mathrm{H}} \mathbf{D}_{\mathrm{\gamma}} = \mathbf{I}_{{M}_1}$. A unitary correlation matrix can help establish orthogonality among the paths received from different RISs, which allows the structure of the Bayesian FIM and EFIM to be analyzed.

\begin{remark}
\label{remark:nusiance_params_diagonal}
The submatrix  $\mathbf{J}_{\bm{\eta}_{2}\bm{\eta}_{2}}^{\mathrm{}}$ and its corresponding entries  $\{\mathbf{J}_{ \bm{\beta}_{\mathrm{R}}  \bm{\beta}_{\mathrm{R}}}, \mathbf{J}_{ \bm{\beta}_{\mathrm{I}}  \bm{\beta}_{\mathrm{I}}}\}$ are diagonal matrices. If the RIS sequence matrix $ \mathbf{D}_{\gamma}$ produces a unitary correlation matrix $\mathbf{D}_{\gamma}^{\mathrm{H}} \mathbf{D}_{\gamma} =\mathbf{I}_{{M}_{1}} $, various  RIS paths can be orthogonalized and  
$\mathbf{J}_{ \bm{\beta}_{\mathrm{R}}  \bm{\beta}_{\mathrm{R}}}
 =  \mathbf{J}_{ \bm{\beta}_{\mathrm{I}}  \bm{\beta}_{\mathrm{I}}}
 =   \frac{2}{\sigma^2}\sum_{q = 1}^{N_{Q}}\left\{\left( \mathbf{A}_{\mathrm{r}_{\mathrm{u}}}^{\mathrm{H}} \mathbf{A}_{\mathrm{r}_{\mathrm{u}}}\right) \odot\left(\mathbf{k}_{q,\mathrm{l}} \mathbf{k}_{q,\mathrm{l}}^{\mathrm{H}}\right)
 \odot\left(\mathbf{D}_{\mathrm{\gamma}}^{\mathrm{H}} \mathbf{D}_{\mathrm{\gamma}}\right) \right. \\
 \left.
\odot\left(\mathbf{A}_{\mathrm{t}_{\mathrm{u}}}^{\mathrm{H}}  \mathbf{F}   \mathbf{F}^{\mathrm{H}} \mathbf{A}_{\mathrm{t}_{\mathrm{u}}}^{}\right)^{\mathrm{T}} 
  \odot \mathbf{R}_{0}\right\}.$
\end{remark}

To  investigate the reduction in information due to uncertainty about the nuisance parameters, we  analyze the term concerning information loss in the Bayesian EFIM $\mathbf{J}_{\bm{\eta}_{1} \bm{\eta}_{1}}^{nu}$. 

\begin{lemma}
\label{lemma:entries_in_nuisance_EFIM}
The terms representing the information loss about the geometric parameters, $\bm{\eta}_1$, due to the nuisance parameters, $\bm{\eta}_2$, have the structure  presented in (\ref{equ:nuisance_FIM_general_form}).
\begin{figure*}
\begin{align}
\label{equ:nuisance_FIM_general_form}
\begin{split}
\mathbf{J}_{\mathbf{v}_{1} \mathbf{v}_{2}}^{nu} &= \left[ \frac{2}{\sigma^2}\sum_{q = 1}^{N_{Q}} { \mathbf{J}}_{q, \bm{\beta}_{\mathrm{R}}  \bm{\beta}_{\mathrm{R}}} + { \mathbf{J}}_{ \bm{\beta}_{\mathrm{R}}  \bm{\beta}_{\mathrm{R}}}^{\mathrm{P}}\right]^{-1}\frac{4}{\sigma^4}\Re\left\{ \sum_{q_1 = 1}^{N_{Q}}\sum_{q_2 = 1}^{N_{Q}}( \mathbf{J}_{q_1,  \mathbf{v}_{1}  \bm{\beta}_{\mathrm{I}}}+ j\mathbf{J}_{q_1,  \mathbf{v}_{1}  \bm{\beta}_{\mathrm{R}}}) (     \mathbf{J}_{q_2,  \mathbf{v}_{2} \bm{\beta}_{\mathrm{I}}}^{\mathrm{}} +j\mathbf{J}_{q_2, \mathbf{v}_{2}  \bm{\beta}_{\mathrm{R}}}^{\mathrm{}}  )^{\mathrm{H}} \right\},\\
&= \left[ \frac{2}{\sigma^2} { \mathbf{J}}_{ \bm{\beta}_{\mathrm{R}}  \bm{\beta}_{\mathrm{R}}} + { \mathbf{J}}_{ \bm{\beta}_{\mathrm{R}}  \bm{\beta}_{\mathrm{R}}}^{\mathrm{P}}\right]^{-1}\frac{4}{\sigma^4}\Re\left\{ \sum_{q_1 = 1}^{N_{Q}}\sum_{q_2 = 1}^{N_{Q}}( \mathbf{J}_{q_1,  \mathbf{v}_{1}  \bm{\beta}_{\mathrm{I}}}+ j\mathbf{J}_{q_1,  \mathbf{v}_{1}  \bm{\beta}_{\mathrm{R}}}) (     \mathbf{J}_{q_2,  \mathbf{v}_{2} \bm{\beta}_{\mathrm{I}}}^{\mathrm{}} +j\mathbf{J}_{q_2, \mathbf{v}_{2}  \bm{\beta}_{\mathrm{R}}}^{\mathrm{}}  )^{\mathrm{H}} \right\}.
\end{split}
\end{align}
\end{figure*}
where $\mathbf{v}_{1}, \mathbf{v}_{2} \in \bm{\eta}_1$.
\end{lemma}
\begin{proof}
    The term concerning information loss, $\mathbf{J}_{\bm{\eta}_{1} \bm{\eta}_{1}}^{nu}$, in the Bayesian EFIM, $\mathbf{J}_{\bm{\eta}_{1}}^{\mathrm{e}}$,  is obtained by applying Definition \ref{definition_EFIM_2} to the Bayesian FIM given by $\mathbf{J}_{\bm{\eta}}^{}$. 
\end{proof}
Lemma \ref{lemma:entries_in_nuisance_EFIM} presents the structure of the information loss about the geometric channel parameters, $\bm{\eta}_1$, due to uncertainty in the nuisance parameters, $\bm{\eta}_2$. In subsequent sections, we use this structure to present estimation limitations under specific configurations of RIS coefficients.
\begin{corollary}
\label{corollary:entries_in_nuisance_EFIM}
An entry in EFIM of the geometric channel  parameters ${\bm{\eta}_{1}}$ is obtained by substituting (\ref{equ:bayesian_FIM_1}), (\ref{equ:bayesian_FIM_2}), and  Lemma \ref{lemma:entries_in_nuisance_EFIM} in Definition \ref{definition_EFIM_2} and it is represented as
\begin{equation}
\label{corollary:definition_EFIM_2}
\begin{aligned}
\mathbf{J}_{\mathbf{v}_1 \mathbf{v}_{2}}^{\mathrm{e}}=\frac{2}{\sigma^2}\sum_{q = 1}^{N_{Q}}\mathbf{J}_{q,\mathbf{v}_1 \mathbf{v}_{2}} + {\mathbf{J}}_{\mathbf{v}_1 \mathbf{v}_2 }^{\mathrm{P}} - \mathbf{J}_{\mathbf{v}_1 \mathbf{v}_{2}}^{nu}
    \end{aligned}
\end{equation}
where $\mathbf{v}_{1} , \mathbf{v}_{2} \in  \bm{\eta}_1$, $\mathbf{v}_{1} = \mathbf{v}_{2}$, and
$\mathbf{J}_{\mathbf{v}_1 \mathbf{v}_{2}}^{\mathrm{e}} = \frac{2}{\sigma^2}\sum_{q = 1}^{N_{Q}}\mathbf{J}_{q,\mathbf{v}_1 \mathbf{v}_{2}}  - \mathbf{J}_{\mathbf{v}_1 \mathbf{v}_{2}}^{nu},$
where $\mathbf{v}_{1} , \mathbf{v}_{2} \in  \bm{\eta}_1$ and $\mathbf{v}_{1} \neq \mathbf{v}_{2}$.
\end{corollary}
The structure of the expressions regarding the information loss terms and the EFIM of the geometric channel parameters given in Lemma \ref{lemma:entries_in_nuisance_EFIM} and Corollary \ref{corollary:entries_in_nuisance_EFIM}, respectively, are difficult to analyze further. However, we can analyze the structure of these expressions under different conditions and state when the estimation problem is infeasible.
\begin{corollary}
    \label{corollary:parallel_nu}
    When the RIS coefficients are parallel, i.e. when for any scalar, $\alpha$, the following condition $\mathbf{\Gamma}_q^{[{m}]} = \alpha \mathbf{\Gamma}_{q^{'}}^{[{m}]}$ holds $ \;  \forall q, q^{'} \in \{1,2,\cdots, N_{Q} \}$,  most of the information loss terms in the EFIM have a definite structure. More specifically, this structure is given by
$
\mathbf{J}_{\mathbf{v}_{1} \mathbf{v}_{2}}^{nu} = \left[ \frac{2}{\sigma^2}N_{Q} { \mathbf{J}}_{q, \bm{\beta}_{\mathrm{R}}  \bm{\beta}_{\mathrm{R}}} + { \mathbf{J}}_{ \bm{\beta}_{\mathrm{R}}  \bm{\beta}_{\mathrm{R}}}^{\mathrm{P}}\right]^{-1}\frac{4}{\sigma^4} N_{Q}{ \mathbf{J}}_{q, \bm{\beta}_{\mathrm{R}}  \bm{\beta}_{\mathrm{R}}}{ \mathbf{J}}_{q, \mathbf{v}_{1} \mathbf{v}_{2}}
$
where $\mathbf{v}_{1}, \mathbf{v}_{2} \in  \bm{\eta}_1$ , $\mathbf{v}_{1} \neq \mathbf{v}_{2}$. Also If $\mathbf{v}_{1} = \bm{\theta}_{\mathrm{r}_{\mathrm{u}}} $ then $\mathbf{v}_{2} \notin \{\bm{\theta}_{\mathrm{r}_{\mathrm{u}}},\bm{\phi}_{\mathrm{r}_{\mathrm{u}}}\} $.  Again for the transmit angles, if  $\mathbf{v}_{1} = \bm{\theta}_{\mathrm{t}_\mathrm{u}} $ then $\mathbf{v}_{2} \notin \{\bm{\theta}_{\mathrm{t}_{\mathrm{u}}},\bm{\phi}_{\mathrm{t}_{\mathrm{u}}}\} $. Information loss terms without this structure  are presented in Appendix \ref{appendix:FIM_nuis_remark_structure}.
\end{corollary}
\begin{proof}
See Appendix \ref{appendix_nusiance_structure}.    
\end{proof}
Corollary \ref{corollary:parallel_nu} presents the structure of the information loss about the geometric channel parameters, $\bm{\eta}_1$, due to uncertainty in the nuisance parameters, $\bm{\eta}_2$ when the RIS coefficients are parallel. With this structure, Corollary \ref{corollary:parallel_EFIM} indicates that all the information about the RIS AoR is lost when there is no prior information about the complex path. Because prior information about the complex path gains is hard to obtain, the RIS AoR can not be estimated with parallel RIS reflection coefficients.
\begin{corollary}
    \label{corollary:parallel_EFIM}
    With parallel RIS coefficients, the submatrices in the EFIM which are related to the AoI and AoR are zero matrices when there is no prior information about the complex path gains and no prior information about the corresponding AoI and AoR.  More specifically, 
$$\mathbf{J}_{  \bm{\theta}_{\mathrm{t}_{\mathrm{l}}}  \bm{\theta}_{\mathrm{t}_{\mathrm{l}}}}^{\mathrm{e}} = \mathbf{J}_{  \bm{\phi}_{\mathrm{t}_{\mathrm{l}}}  \bm{\phi}_{\mathrm{t}_{\mathrm{l}}}}^{\mathrm{e}} = \mathbf{J}_{  \bm{\theta}_{\mathrm{\mathrm{r}_{\mathrm{l}}}}  \bm{\theta}_{\mathrm{\mathrm{r}_{\mathrm{l}}}}}^{\mathrm{e}} = \mathbf{J}_{  \bm{\phi}_{\mathrm{\mathrm{r}_{\mathrm{l}}}}  \bm{\phi}_{\mathrm{\mathrm{r}_{\mathrm{l}}}}}^{\mathrm{e}} = \mathbf{0},
$$
when ${\mathbf{J}}_{  \bm{\theta}_{\mathrm{t}_{\mathrm{l}}}  \bm{\theta}_{\mathrm{t}_{\mathrm{l}}}}^{\mathrm{P}} = {\mathbf{J}}_{  \bm{\phi}_{\mathrm{t}_{\mathrm{l}}}  \bm{\phi}_{\mathrm{t}_{\mathrm{l}}}}^{\mathrm{P}} = {\mathbf{J}}_{  \bm{\theta}_{\mathrm{r}_{\mathrm{l}}}  \bm{\theta}_{\mathrm{r}_{\mathrm{l}}}}^{\mathrm{P}} = {\mathbf{J}}_{  \bm{\phi}_{\mathrm{r}_{\mathrm{l}}}  \bm{\phi}_{\mathrm{r}_{\mathrm{l}}}}^{\mathrm{P}} = { \mathbf{J}}_{ \bm{\beta}_{\mathrm{R}}  \bm{\beta}_{\mathrm{R}}}^{\mathrm{P}} = \mathbf{0}.$
\end{corollary}
\begin{proof}
This is a direct consequence of Corollary \ref{corollary:parallel_nu}.
\end{proof}
With Corollary \ref{corollary:parallel_EFIM}, when the RIS coefficients are parallel, the RIS AoR can not be estimated. This result is clearly stated in the Theorem below.
\begin{theorem}
\label{theorem_AoI_AoR}
Practically, even with infinite number of receive antennas, the azimuth and elevation AoR and the complex path gains of the $m^{\text{th}}$ RIS path can only be jointly estimated when the $m^{\text{th}}$ RIS employs  non-parallel RIS coefficients. This entails explicitly that localization of a single receive antenna UE with reflections from a single RIS is impossible when there is no LOS, and the RIS coefficients are constant across all OFDM symbols.
\end{theorem}
\begin{proof}
The proof follows after noticing that all the angle information of the $m^{\text{th}}$ RIS is lost due to uncertainty in the complex path gain associated with the $m^{\text{th}}$ RIS path (see Corollary     \ref{corollary:parallel_EFIM}).
\end{proof}
\begin{remark}
\label{remark_AoI_AoR_MIMO}
The fundamental difference between a passive RIS and a BS is that no processing is done at a passive RIS. Because there is no processing at a passive RIS, it can only reflect signals into a single stream. Hence, non-parallel streams of information cannot be created during a single OFDM symbol with a passive RIS. This is in stark contrast to a BS that can (with various precoding techniques) produce multiple non-parallel streams \cite{sohrabi2017hybrid}.

Due to this contrast in operation, when there is more than one receive antenna at the UE, the UE can use the non-parallel streams during a single OFDM symbol to estimate the AoD at the BS \cite{8755880,fascista2021downlink}. However, irrespective of the number of receive antennas at the UE, it can not estimate the  AoR of the RIS using only reflections during a single OFDM symbol. Therefore, the  AoR can only be calculated at the UE when the RIS employs  non-parallel coefficients across multiple OFDM symbols. 

It is also important to note that the AoD of the BS can also not be estimated during a single OFDM symbol if the streams created by the precoding strategies are parallel in the spatial domain. 
\end{remark}
In the previous sections, we have discussed the relationship between geometric channel parameters, $\bm{\eta}_1$ and the nuisance parameters, $\bm{\eta}_2$ under both parallel and non-parallel RIS coefficients. In the following sections, we provide mathematical definitions of the AoI azimuth and elevation angles and the AoR azimuth and elevation angles. The next section defines these RIS-related angles with respect to the BS's position, BS's orientation, RISs' position, RISs' orientation, UE's position, and UE's orientation.

\subsection{RIS Related Angle Definitions and Relationships}

To analyze both the angle relationships and derive the FIM for positioning, we define the rotation matrix  $\mathbf{Q}\left(\theta_{0}, \phi_{0}\right)$  given by $\mathbf{Q}\left(\theta_{0}, \phi_{0}\right) =\mathbf{Q}_{z}\left(\phi_{0}\right) \mathbf{Q}_{-x^{}}\left(\theta_{0}\right)$, where $\mathbf{Q}_{z}\left(\phi_{0}\right)$ and $\mathbf{Q}_{-x^{}}\left(\theta_{0}\right)$ define the counter-clockwise rotation around the $z$-axis and the clockwise rotation around the $x$-axis respectively.  We define $\mathbf{g}^{[{m}]}_{{}}=( \mathbf{p}^{[{m}]} -  \mathbf{p}_{\mathrm{BS}})$,
and specify the AoD at the BS as
$\theta_{\mathrm{t}_{\mathrm{u}}}^{[{m}]} =\cos ^{-1}\left({g}^{[{m}]}_{z_{}} /\|\mathbf{g}^{[{m}]}_{{}}\|\right), \; \; 
\phi_{\mathrm{t}_{\mathrm{u}}}^{[{m}]} =\tan ^{-1}\left({g}^{[{m}]}_{y_{}} /{g}^{[{m}]}_{x_{}}\right)$. Next, we translate  the ${m}^{\text{th}}$ RIS to the origin, and the new coordinates of the BS can be written as $\mathbf{c}^{[{m}]}_{{}}=(\mathbf{p}_{\mathrm{BS}}- \mathbf{p}^{[{m}]}), \; \;
\Tilde{\mathbf{c}}^{[{m}]}_{}=\mathbf{Q}^{-1}\left(\theta_{0}^{[{m}]}, \phi_{0}^{[{m}]}\right) \mathbf{c}^{[{m}]}$. With respect to these new coordinates, we can write $\theta_{\mathrm{r}_{\mathrm{l}}}^{[{m}]} =\cos ^{-1}\left({c}^{[{m}]}_{\Tilde{z}} /\|\Tilde{\mathbf{c}}^{[{m}]}_{}\|\right), \; \;
\phi_{\mathrm{r}_{\mathrm{l}}}^{[{m}]} =\tan ^{-1}\left({c}^{[{m}]}_{\Tilde{y}} /{c}^{[{m}]}_{\Tilde{x}}\right)$. Subsequently, the translated coordinates   allow the following definition 
$\mathbf{v}^{[{m}]}_{{}}=(\mathbf{p} - \mathbf{p}^{[{m}]}), \; \;
\Tilde{\mathbf{v}}^{[{m}]}_{}=\mathbf{Q}^{-1}\left(\theta_{0}^{[{m}]}, \phi_{0}^{[{m}]}\right)  \mathbf{v}^{[{m}]}_{{}}$ and we can write $\theta_{\mathrm{t}_{\mathrm{l}}}^{[{m}]} =\cos ^{-1}\left({v}^{[{m}]}_{\Tilde{z}} /\|\Tilde{\mathbf{v}}^{[{m}]}_{}\|\right), \; \;
\phi_{\mathrm{t}_{\mathrm{l}}}^{[{m}]} =\tan ^{-1}\left({v}^{[{m}]}_{\Tilde{y}} /{v}^{[{m}]}_{\Tilde{x}}\right)$.
Similarly, we obtain a new set of coordinates by translating the UE to the origin, and we write the following definitions 
$\mathbf{e}^{[{m}]}_{{}}= -(\mathbf{p} - \mathbf{p}^{[{m}]}), \; \;
\Tilde{\mathbf{e}}^{[{m}]}_{}=\mathbf{Q}^{-1}\left(\theta_{0}^{}, \phi_{0}^{}\right)  \mathbf{e}^{[{m}]}_{{}}$. Hence, we can write $\theta_{\mathrm{r}_{\mathrm{u}}}^{[{m}]} =\cos ^{-1}\left({e}^{[{m}]}_{\Tilde{z}} /\|\Tilde{\mathbf{e}}^{[{m}]}_{}\|\right), \; \;
\phi_{\mathrm{r}_{\mathrm{u}}}^{[{m}]} =\tan ^{-1}\left({e}^{[{m}]}_{\Tilde{y}} /{e}^{[{m}]}_{\Tilde{x}}\right)$.


\subsection{Relationship between AoI and AoR with a Unitary RIS Correlation Matrix} 
Based on the angle relationships and the coordinate translations in the previous section, we show the relationship between the information provided by the AoI and information provided by the AoR for an RIS deployed as a passive uniform rectangular array (URA). To show this relationship, we state the following lemma.   
\begin{lemma}\label{lemma:AoI_AoR}
The matrix specified by
\begin{equation}
\label{equ:AoI_AoR}
\begin{aligned}
\mathbf{V} &= 
\left[\begin{array}{cc}
\bm{\nu}_{1} & \bm{\nu}_{2}  \\\end{array}\right]  \\ &= 
\left[\begin{array}{cc}
\cos(\theta_{\mathrm{t}_{\mathrm{l}}}^{[{m}]}) \sin(\phi_{\mathrm{t}_{\mathrm{l}}}^{[{m}]})   &   \sin(\theta_{\mathrm{t}_{\mathrm{l}}}^{[{m}]}) \cos(\phi_{\mathrm{t}_{\mathrm{l}}}^{[{m}]})  \\
\cos(\theta_{\mathrm{t}_{\mathrm{l}}}^{[{m}]}) \cos(\phi_{\mathrm{t}_{\mathrm{l}}}^{[{m}]}) & -\sin(\theta_{\mathrm{t}_{\mathrm{l}}}^{[{m}]}) \sin(\phi_{\mathrm{t}_{\mathrm{l}}}^{[{m}]})    \\
\end{array}\right] \\
\end{aligned}
\end{equation}
is  a full rank matrix. Hence, the 2D vectors $\bm{\nu}_{3} = \left[\begin{array}{cc} \cos(\theta_{\mathrm{r}_{\mathrm{l}}}^{[{m}]}) \sin(\phi_{\mathrm{r}_{\mathrm{l}}}^{[{m}]}) \; \; \cos(\theta_{\mathrm{r}_{\mathrm{l}}}^{[{m}]}) \sin(\phi_{\mathrm{r}_{\mathrm{l}}}^{[{m}]}) \end{array}\right]^{\mathrm{T}}$ and $\bm{\nu}_{4} = \left[\begin{array}{cc} \sin(\theta_{\mathrm{r}_{\mathrm{l}}}^{[{m}]}) \cos(\phi_{\mathrm{r}_{\mathrm{l}}}^{[{m}]}) \; \; -\sin(\theta_{\mathrm{r}_{\mathrm{l}}}^{[{m}]}) \sin(\phi_{\mathrm{r}_{\mathrm{l}}}^{[{m}]}) \end{array}\right]^{\mathrm{T}}$ can  be obtained as a linear combination of $\bm{\nu}_{1}$ and $\bm{\nu}_{2}$.
\end{lemma}
\begin{proof}
By simple geometry and  based on the angle definitions, we can write (\ref{equ:AoI_AoR}) as 
\begin{equation}
\label{equ:AoI_AoR_equivalent}
\begin{aligned}
\mathbf{V} 
= 
&\left[\begin{array}{cc}
 \frac{{v}^{[{m}]}_{\Tilde{y}} {v}^{[{m}]}_{\Tilde{z}}}{\left(\|\Tilde{\mathbf{v}}^{[{m}]}_{}\| \sqrt{({v}^{[{m}]}_{\Tilde{x}})^2 + ({v}^{[{m}]}_{\Tilde{y}})^2 } \right)} &
\frac{{v}^{[{m}]}_{\Tilde{x}}}{\|\Tilde{\mathbf{v}}^{[{m}]}_{}\|}    \\
\frac{{v}^{[{m}]}_{\Tilde{x}} {v}^{[{m}]}_{\Tilde{z}}}{\left(\|\Tilde{\mathbf{v}}^{[{m}]}_{}\| \sqrt{({v}^{[{m}]}_{\Tilde{x}})^2 + ({v}^{[{m}]}_{\Tilde{y}})^2 } \right)} &\frac{-{v}^{[{m}]}_{\Tilde{y}}}{\|\Tilde{\mathbf{v}}^{[{m}]}_{}\|}    \\ 
\end{array}\right]. 
\end{aligned}
\end{equation}
Based on the property that a rank deficient matrix has a zero determinant, we  obtain $-({v}^{[{m}]}_{\Tilde{y}})^2
 = ({v}^{[{m}]}_{\Tilde{x}})^2$ as the only condition for rank deficiency. Because this rank deficiency condition is not possible, the lemma is proved.
The second part of the Lemma is obvious as $\bm{\nu}_{3}$ and $\bm{\nu}_{4}$ are 2D vectors which can be obtained from  a linear combination of $\bm{\nu}_{1}$ and  $\bm{\nu}_{2}$.
\end{proof}
The following corollaries establish relationships between the information provided by the FIMs of various channel parameters. The first corollary is a vital step in showing dependence among some of the FIMs of the geometric channel parameters.
It establishes a relationship between the following: i) the derivative with respect to the elevation AoR of the exponent in the array response vector due to the reflected signal at the ${m}^{\text{th}}$ RIS specified by $ \mathbf{K}_{\mathrm{t}_{\mathrm{l}}}^{[{m}]}$, ii) the derivative with respect to the azimuth AoR of the exponent in the array response vector due to reflected signal at the same RIS specified by $ \mathbf{P}_{\mathrm{t}_{\mathrm{l}}}^{[{m}]}$, and iii) the derivative with respect to the elevation AoI of the exponent in the array response vector due to incident signal at the same RIS specified by $ \mathbf{K}_{\mathrm{\mathrm{r}_{\mathrm{l}}}}^{[{m}]}$. More specifically, the first corollary indicates that there exists a linear combination of (i) and (ii) that produces (iii). Additionally, in the first corollary, a similar statement is made about the azimuth AoI.
\begin{corollary}
\label{corollary:corollary_ura_p_k_p}
For a RIS deployed as a passive URA with a normal in the z-direction, there exist scalars $\alpha_1$, $\alpha_2$, $\alpha_3$, and $\alpha_4$ such that
\begin{equation}
\label{equ:corollary_p_t_l_k_r_l}
\begin{aligned}
\alpha_1 \mathbf{K}_{\mathrm{t}_{\mathrm{l}}}^{[{m}]} + \alpha_2 \mathbf{P}_{\mathrm{t}_{\mathrm{l}}}^{[{m}]} =  \mathbf{K}_{\mathrm{\mathrm{r}_{\mathrm{l}}}} ^{[{m}]}, \; \;
\alpha_3 \mathbf{K}_{\mathrm{t}_{\mathrm{l}}}^{[{m}]} + \alpha_4 \mathbf{P}_{\mathrm{t}_{\mathrm{l}}}^{[{m}]} =  \mathbf{P}_{\mathrm{\mathrm{r}_{\mathrm{l}}}} ^{[{m}]}.
\end{aligned}
\end{equation}
\end{corollary}

\begin{proof}
\label{equ:proof_corollary_p_t_l_k_r_l}
The diagonal matrices in (\ref{equ:corollary_p_t_l_k_r_l}) have a size of $(N_L^{[{m}]})$ and each element in these matrices can be decomposed into  components in the $\mathrm{x}$ and $\mathrm{y}$ directions. From (\ref{equ:FIM_parameter_terms_ris_2}), the angle component of $ \mathbf{K}_{\mathrm{t}_{\mathrm{l}}}^{[{m}]}$ in the $\mathrm{x}$ and $\mathrm{y}$ direction can be shown to correspond with the elements of  $\bm{\nu}_{1}$. More specifically
$
\mathbf{K}_{\mathrm{t}_{\mathrm{l}}}^{[{m}]} = \pi \mathbf{Y}_{\mathrm{l,m}} \cos(\theta_{\mathrm{t}_{\mathrm{l}}}^{[{m}]}) \sin(\phi_{\mathrm{t}_{\mathrm{l}}}^{[{m}]}) + \pi \mathbf{X}_{\mathrm{l,m}} \cos(\theta_{\mathrm{t}_{\mathrm{l}}}^{[{m}]}) \cos(\phi_{\mathrm{t}_{\mathrm{l}}}^{[{m}]}).
$
Similarly, the angle components of  $ \mathbf{P}_{\mathrm{t}_{\mathrm{l}}}^{[{m}]}$ and $\mathbf{K}_{\mathrm{\mathrm{r}_{\mathrm{l}}}} ^{[{m}]}$ can be shown to equal $\bm{\nu}_{2}$ and $\bm{\nu}_{3}$, respectively. Hence,  $\mathbf{K}_{\mathrm{\mathrm{r}_{\mathrm{l}}}} ^{[{m}]}$ is a linear combination of $\mathbf{K}_{\mathrm{t}_{\mathrm{l}}}^{[{m}]} $ and $\mathbf{P}_{\mathrm{t}_{\mathrm{l}}}^{[{m}]}$. The second equation in this corollary can be proved similarly.
\end{proof}
Corollary \ref{corollary:fim_ura_unitary} shows that the information about the AoI can be obtained as linear combination of the information about the AoR. 

\begin{corollary}
\label{corollary:fim_ura_unitary}
For a unitary RIS correlation matrix and with Corollary \ref{corollary:corollary_ura_p_k_p},  there exist scalars $\alpha_1$, $\alpha_2$, $\alpha_3$, and $\alpha_4$ such that
\begin{equation}
\label{equ:corollary_fim_ura_unitary}
\begin{aligned}
\alpha_1 \mathbf{J}_{\mathbf{v}_1 \bm{\theta}_{\mathrm{t}_{\mathrm{l}}}} + \alpha_2 \mathbf{J}_{\mathbf{v}_1 \bm{\phi}_{\mathrm{t}_{\mathrm{l}}}} =  \mathbf{J}_{\mathbf{v}_1 \bm{\theta}_{\mathrm{r}_{\mathrm{l}}}}, \\ 
\alpha_3 \mathbf{J}_{\mathbf{v}_1 \bm{\theta}_{\mathrm{t}_{\mathrm{l}}}} + \alpha_4 \mathbf{J}_{\mathbf{v}_1 \bm{\phi}_{\mathrm{t}_{\mathrm{l}}}} =  \mathbf{J}_{\mathbf{v}_1 \bm{\phi}_{\mathrm{r}_{\mathrm{l}}}},
\end{aligned}
\end{equation}
where  $\mathbf{v}_1 \in \bm{\eta}$.
\end{corollary}
\begin{proof}
First, note that the FIMs in the above corollary are diagonal matrices. Due to the properties relating Hadamard products with diagonal matrices \cite{horn2012matrix} and Lemma \ref{lemma:entries_in_FIM}, the left-hand side of the above Corollary can be decomposed as 
$\medmath{\frac{2}{\sigma^2}\sum_{q = 1}^{N_{Q}}\Re{\left\{\mathbf{V}_{\mathbf{v}_1} \odot \left[\alpha_1 \left(\mathbf{k}_{q,\mathbf{v}_1} \mathbf{k}_{q,\mathrm{t}_{\mathrm{l}}}^{\mathrm{H}}\right)  + \alpha_2  \left(\mathbf{k}_{q,\mathbf{v}_1} \mathbf{p}_{q,\mathrm{t}_{\mathrm{l}}}^{\mathrm{H}}\right) \right]  \odot\left(\mathbf{D}_{\mathrm{\gamma}}^{\mathrm{H}} \mathbf{D}_{\mathrm{\gamma}}\right)\right\} },} $
where $\mathbf{V}_{\mathbf{v}_1}$ is a dummy diagonal matrix representing the common terms between $\mathbf{J}_{q,\mathbf{v}_1 \bm{\theta}_{\mathrm{t}_{\mathrm{l}}}}$ and $\mathbf{J}_{q,\mathbf{v}_1 \bm{\phi}_{\mathrm{t}_{\mathrm{l}}}}$. Analyzing the diagonal elements of the matrices in the square brackets gives
\begin{equation}
\label{equ:proof_corollary_fim_ura_unitary_1}
\begin{aligned}
\medmath{\frac{2}{\sigma^2}\sum_{q = 1}^{N_{Q}}\Re{\left\{{{v}}_{q,{v}_1}^{[{m}]} \odot  \mathbf{{a}}^{\mathrm{H}[{m}]}_{q,{\mathrm{t}_{\mathrm{l}}}}\left[\alpha_1 \mathbf{K}_{\mathrm{t}_{\mathrm{l}}}^{[{m}]}  + \alpha_2  \mathbf{P}_{\mathrm{t}_{\mathrm{l}}}^{[{m}]} \right]  \Tilde{\mathbf{a}}_{q,\mathrm{r}_{\mathrm{l}}}^{[{m}]} \odot\left(\mathbf{D}_{\mathrm{\gamma}}^{\mathrm{H}} \mathbf{D}_{\mathrm{\gamma}}\right)\right\} }},
\end{aligned}
\end{equation}
where ${{v}}_{q,{v}_1}^{[{m}]}$ represents the common terms on the $m^{\text {th }}$ diagonal .
From  Corollary \ref{corollary:corollary_ura_p_k_p}, the terms in the square brackets equals $ \mathbf{K}_{\mathrm{\mathrm{r}_{\mathrm{l}}}}^{[{m}]}$, hence $\frac{2}{\sigma^2}\sum_{q = 1}^{N_{Q}}\Re{\left\{{{v}}_{q,{v}_1}^{[{m}]} \odot  \mathbf{{a}}^{\mathrm{H}[{m}]}_{q,{\mathrm{t}_{\mathrm{l}}}} \mathbf{K}_{\mathrm{\mathrm{r}_{\mathrm{l}}}}^{[{m}]}   \Tilde{\mathbf{a}}_{q,\mathrm{r}_{\mathrm{l}}}^{[{m}]} \odot\left(\mathbf{D}_{\mathrm{\gamma}}^{\mathrm{H}} \mathbf{D}_{\mathrm{\gamma}}\right)\right\} }$,
and we can write
$\frac{2}{\sigma^2}\sum_{q = 1}^{N_{Q}}\Re{\left\{\mathbf{V}_{\mathbf{v}_1} \odot  \left(\mathbf{k}_{q,\mathbf{v}_1} \mathbf{k}_{q,\mathrm{r}_{\mathrm{l}}}^{\mathrm{H}}\right)  \odot\left(\mathbf{D}_{\mathrm{\gamma}}^{\mathrm{H}} \mathbf{D}_{\mathrm{\gamma}}\right)\right\} } = \mathbf{J}_{\mathbf{v}_1 \bm{\theta}_{\mathrm{r}_{\mathrm{l}}}}.$
The second part of the corollary can be proved similarly.
\end{proof}
\begin{remark}
\label{remark:data_FIM_rank_def}
From Corollary \ref{corollary:fim_ura_unitary}, the FIM of the channel parameters $\bm{\eta}$ specified by $\mathbf{J}_{\bm{\eta}}^{\mathrm{D}}$ is  rank deficient and non-invertible. Also, if there are ${M}_1$ RISs, the resultant FIM $\mathbf{J}_{\bm{\eta}}^{\mathrm{D}}$ has a rank of atmost $11{M}_1 - 2{M}_1$. From Corollary \ref{corollary:fim_ura_unitary}, the AoIs and AoRs can not be estimated separately, irrespective of parallel or non-parallel RIS coefficients. Hence, the model is non-identifiable. The cause of the non-identifiability is that the parameter vector $\bm{\eta}$ is parameter redundant. This can be seen by noting that the parameter vector can be reparameterized in terms of a smaller set of parameters. This reparameterization can be achieved by: i) replacing the AoI and AoR (azimuth and elevation) with the RIS orientation offsets, ii) replacing the azimuth and elevation AoD with the azimuth and elevation angles in the unit vector that points from the BS to the RIS, and iii) replacing the azimuth and elevation AoA with the azimuth and elevation angles in the unit vector pointing from the RIS to the UE.   This reparameterization reduces the size of the parameter vector to $9M_1$, and the length of the parameter vector of the geometric channel parameters decreases to $7M_1$. With this reparameterization, the EFIM of the geometric channel parameters is now a full-rank matrix when non-parallel RIS coefficients are used.
\end{remark}

\subsection{Fisher Information for RIS Paths Plus LOS}
In this subsection, we analyze the FIM of the LOS plus RIS paths.
We define the geometric LOS channel parameters $ \bm{\psi}_{1} \triangleq\left[\theta_{\mathrm{r}_{\mathrm{u}}}^{[0]}, \phi_{\mathrm{r}_{\mathrm{u}}}^{[0]}, \theta_{\mathrm{t}_{\mathrm{u}}}^{[0]}, \phi_{\mathrm{t}_{\mathrm{u}}}^{[0]}, \tau^{[\mathrm{0}]}_{} \right]^{\mathrm{T}}
$ and  the LOS nuisance parameter as $ \bm{\psi}_{2} \triangleq\left[{\beta}_{\mathrm{R}}^{\mathrm{[0]}}, {\beta}_{\mathrm{I}}^{\mathrm{[0]}} 
\right]^{\mathrm{T}}
$. 
The LOS plus RIS geometric channel parameters are defined as $ \bm{\zeta}_{1} \triangleq\left[ \bm{\psi}_{1}^{\mathrm{T}},  \bm{\eta}_{1}^{\mathrm{T}} \right]^{\mathrm{T}}$ and the LOS plus RIS nuisance parameters is defined as $ \bm{\zeta}_{2} \triangleq\left[ \bm{\psi}_{2}^{\mathrm{T}},  \bm{\eta}_{2}^{\mathrm{T}} \right]^{\mathrm{T}}.$ Hence, the LOS plus RIS channel parameters can be written as
$
\begin{aligned}
\begin{array}{llll}
\bm{\zeta}_{} \triangleq\left[ \bm{\zeta}_{1}^{\mathrm{T}},  \bm{\zeta}_{2}^{\mathrm{T}} \right]^{\mathrm{T}}.
\end{array}
\end{aligned}
$
To write the FIM of $\bm{\zeta}_{}$, we define the PDF as
\begin{equation}
\label{equ:los_pdf_joint_channel}
\begin{aligned}
\chi(\mathbf{r}_{t}[n];  \bm{\zeta}_{} ) = \chi(\mathbf{r}_{t}[n] |  \bm{\zeta}_{} ) \chi(\bm{\zeta}_{} ),
\end{aligned}
\end{equation}
where $\chi(\bm{\zeta}_{} ) = \chi(\bm{\psi}_{1}) \chi( \bm{\psi}_{2})\chi(\bm{\eta}_{1}) \chi( \bm{\eta}_{2}) $. The FIM of the LOS and the RIS channel parameters due to observation $\bm{r}$ has the structure 
$
\mathbf{J}_{\zeta}^{\mathrm{D}} \triangleq\left[\begin{array}{c|ccc}
\mathbf{J}_{\psi}^{\mathrm{D}} & \mathbf{J}_{\psi \eta}^{\mathrm{D}} \\
\hline \mathbf{J}_{ \eta \psi}^{\mathrm{D}} & \mathbf{J}_{\eta}^{\mathrm{D}} \\
\end{array}\right],
$
where $\mathbf{J}_{\zeta}^{\mathrm{D}} \in \mathbb{R}^{(11 {M}_1 +7) \times (11 {M}_1 +7)}$,  $\mathbf{J}_{\psi}^{\mathrm{D}} \in \mathbb{R}^{7 \times 7}$, and $\mathbf{J}_{\psi \eta}^{\mathrm{D}} \in \mathbb{R}^{7 \times 11 {M}_1}$. The entries of the latter two matrices are written in Appendix \ref{appendix:EFIM_Los}. The Bayesian FIM $\mathbf{J}_{ \bm{\zeta}_{}}$ can be written as described in Section \ref{section:Fisher_Information_for_RIS_Paths}. Likewise, the equivalent Bayesian EFIM $\mathbf{J}_{ \bm{\zeta}^{}_{1}}^{\mathrm{e}}$ can be written using Definition \ref{definition_EFIM_2}.
\subsection{LOS Related Angle Definitions and Relationships}
This section presents the entries in the transformation matrix that is needed to transform the LOS channel parameters into location parameters. To analyze both the angle relationships and derive the FIM for positioning, we define $\mathbf{g}^{}_{{}}=( \mathbf{p}^{} -  \mathbf{p}_{\mathrm{BS}})$,
and specify the angles of departure at the BS as
$\theta_{\mathrm{t}_{\mathrm{u}}}^{[\mathrm{0}]} =\cos ^{-1}\left({g}^{}_{z_{}} /\|\mathbf{g}^{}_{{}}\|\right), \; \; 
\phi_{\mathrm{t}_{\mathrm{u}}}^{[\mathrm{0}]} =\tan ^{-1}\left({g}^{}_{y_{}} /{g}^{}_{x_{}}\right)$. Next, we translate  the UE to the origin, and the new coordinates of the BS can be written as $\mathbf{e}^{}_{{}}=(\mathbf{p}_{\mathrm{BS}}- \mathbf{p}^{}), \; \;
\Tilde{\mathbf{e}}^{}_{}=\mathbf{Q}^{-1}\left(\theta_{0}^{}, \phi_{0}^{}\right) \mathbf{e}^{}$. With respect to these new coordinates, we can write $\theta_{\mathrm{r}_{\mathrm{u}}}^{[\mathrm{0}]}=\cos ^{-1}\left({e}^{}_{\Tilde{z}} /\|\Tilde{\mathbf{e}}^{}_{}\|\right), \; \;
\phi_{\mathrm{r}_{\mathrm{u}}}^{[\mathrm{0}]} =\tan ^{-1}\left({e}^{}_{\Tilde{y}} /{e}^{}_{\Tilde{x}}\right)$. 
\section{Fisher Information of Location Parameters}
In this section, we derive the FIM and the EFIM of the location parameters. We derive these information matrices for both the case with an arbitrary RIS sequence matrix and the special case with RIS sequences that both generates a unitary correlation matrix and sum to zero (see Assumption \ref{asumption_far_field_rf_control_1}). Based on the FIMs and EFIMs of the location parameters, we also derive expressions for the PEB and OEB for the UE. The location parameters are defined as   $ \bm{\eta}_{L} \triangleq\left[
\mathbf{o}_{}^{\mathrm{T}}, \mathbf{p}_{}^{\mathrm{T}},
 \mathbf{{o}^{\mathrm{T}}}^{[1]}, \mathbf{{p}^{\mathrm{T}}}^{[1]},\mathbf{{o}^{\mathrm{T}}}^{[2]}, \mathbf{{p}^{\mathrm{T}}}^{[2]},\ldots,\mathbf{{o}^{\mathrm{T}}}^{[{M_1}]}, \mathbf{{p}^{\mathrm{T}}}^{[{M_1}]}
\right]^{\mathrm{T}}  \\
\triangleq\left[
\mathbf{o}_{}^{\mathrm{T}}, \mathbf{p}_{}^{\mathrm{T}},
 \mathbf{{q}^{\mathrm{T}}}
\right]^{\mathrm{T}}.$ The PDF $\chi(\mathbf{r}_{t}[n];  \mathbf{\eta}_{L} )$ is obtained as
\begin{equation}
\label{equ:pdf_joint_positioning}
\begin{aligned}
\chi(\mathbf{r}_{t}[n];  \bm{\eta}_{L} ) = \chi(\mathbf{r}_{t}[n] |  \bm{\eta}_{L} ) \chi(\bm{\eta}_{L} )
\end{aligned}
\end{equation}
where 
$\chi(\bm{\eta}_{L} ) = \chi(\mathbf{o}_{}^{}; \mathbf{p}_{}^{}) \prod_{m \in \mathcal{M}_{1}} \chi(\mathbf{o}_{}^{[{m}]}; \mathbf{p}_{}^{[{m}]} | \mathbf{o}_{}^{}, \mathbf{p}_{}^{}) = \chi(\mathbf{o}_{}^{}) \chi( \mathbf{p}_{}^{}) \prod_{m \in \mathcal{M}_{1}} \chi(\mathbf{o}_{}^{[{m}]}| \mathbf{o}_{}^{}, \mathbf{p}_{}^{}) \chi(\mathbf{p}_{}^{[{m}]}| \mathbf{o}_{}^{}, \mathbf{p}_{}^{}).$
Based on the PDF in (\ref{equ:pdf_joint_positioning}), the FIM can be written as $\mathbf{J}_{ \bm{\eta}_{L}}  
= \mathbf{J}_{ \bm{\eta}_{L}} ^{\mathrm{D} } + \mathbf{J}_{ \bm{\eta}_{L}} ^{\mathrm{P} }$.
The parameter vector $\bm{\eta}_{L}$ has a nonlinear relationship with the geometric channel parameters $
\bm{\eta}_{L} = \Upsilon(\bm{\zeta}_{1}).
$ In \cite{kay1993fundamentals}, it was shown that this nonlinear relationship allows the FIM to be written as
\begin{equation}
\label{equ:FIM_position}
\begin{aligned}
\mathbf{J}_{\bm{\eta}_{\mathrm{L}}}^{\mathrm{D}} \triangleq \mathbf{\Upsilon} \mathbf{J}_{\bm{\zeta}_{1}}^{\mathrm{e}} \mathbf{\Upsilon}^{\mathrm{T}} ,
\end{aligned}
\end{equation}
where $\mathbf{\Upsilon}$ is a transformation matrix obtained by finding the gradient of the relationship between the location parameters and the geometric channel parameters.

\subsection{Transformation Matrix from Geometric Channel Parameters to Location Parameters}
The transformation matrix in (\ref{equ:FIM_position})  can be defined as
$$
\mathbf{\Upsilon} \triangleq \frac{\partial \bm{\zeta}_{\mathrm{1}}^{\mathrm{T}}}{\partial \bm{\eta}_{\mathrm{L}}}=\left[\begin{array}{lllll}
\frac{\partial \bm{\psi}^{}_{1}}{\partial \mathbf{o}_{}} &
\frac{\partial \bm{\psi}^{}_{1}}{\partial \mathbf{p}_{}} & \frac{\partial \bm{\psi}^{}_{1}}{\partial \mathbf{q}} \\   
 \frac{\partial \bm{\eta}^{}_{1}}{\partial \mathbf{o}_{}} 
 &
 \frac{\partial \bm{\eta}^{}_{1}}{\partial \mathbf{p}_{}} 
 & \frac{\partial \bm{\eta}^{}_{1}}{\partial \mathbf{q}} \\
\end{array}\right]^{\mathrm{T}}
$$
with the entries of this matrix presented in Appendix \ref{appendix:entries_transformation} and Appendix \ref{appendix:los_entries_transformation}.
\subsection{Bayesian PEB and OEB: Arbitrary Correlation Matrix}
In this section, we derive the FIM for the location parameters  with an arbitrary RIS correlation matrix. This FIM has the structure
\begin{equation}
\label{equ:position_FIM_1}
\begin{aligned}
\mathbf{J}_{\bm{\eta}_{\mathrm{L}}}^{\mathrm{D}} \triangleq \mathbf{\Upsilon} \mathbf{J}_{\bm{\zeta}_{1}}^{\mathrm{e}} \mathbf{\Upsilon}^{\mathrm{T}}  =\left[\begin{array}{cc}
\mathbf{J}_{\bm{\eta}_{L_1}\bm{\eta}_{L_1}} & \mathbf{J}_{\bm{\eta}_{L_1}\bm{\eta}_{L_2}} \\
\mathbf{J}_{\bm{\eta}_{L_1}\bm{\eta}_{L_2}}^{\mathrm{T}} & \mathbf{J}_{\bm{\eta}_{L_2}\bm{\eta}_{L_2}}
\end{array}\right].
    \end{aligned}
\end{equation}
With the assumption of independent prior information about the UE and the RISs, the entries in the Bayesian FIM $\mathbf{J}_{ \bm{\eta}_{L}}$ can be written as $\Tilde{\mathbf{J}}_{ {}_{\mathbf{v}_1 \mathbf{v}_1 }} = {\mathbf{J}}_{ {}_{\mathbf{v}_1 \mathbf{v}_1 }} + {\mathbf{J}}_{ {}_{\mathbf{v}_1 \mathbf{v}_1 }}^{\mathrm{P}},  
\Tilde{\mathbf{J}}_{ {}_{\mathbf{v}_1 \mathbf{v}_2 }} = {\mathbf{J}}_{ {}_{\mathbf{v}_1 \mathbf{v}_2 }}$
where  $\mathbf{v}_{1} , \mathbf{v}_{2} \in \bm{\eta}_L$.  Using Definition \ref{definition_EFIM_2} and (\ref{equ:position_FIM_1}),  
we can write the Bayesian FIM as 
$
\mathbf{J}_{\bm{\eta}_{L}}^{\mathrm{e}}=\mathbf{J}_{_{\bm{\eta}_{L_1}\bm{\eta}_{L_1}}}- \mathbf{J}_{_{\bm{\eta}_{L_1}\bm{\eta}_{L_1}}}^{nu} =\mathbf{J}_{_{\bm{\eta}_{L_1}\bm{\eta}_{L_1}}} - 
\mathbf{J}_{_{\bm{\eta}_{L_1}\bm{\eta}_{L_2}}} \mathbf{J}_{_{\bm{\eta}_{L_2}\bm{\eta}_{L_2}}}^{-1} \mathbf{J}_{_{\bm{\eta}_{L_1}\bm{\eta}_{L_2}}}^{\mathrm{T}}.$  
The term $\mathbf{J}_{_{\bm{\eta}_{L_1}\bm{\eta}_{L_1}}}^{nu}$ accounts for information loss about the UE location due to the uncertainty in the RISs orientation and position. The SPEB and the SOEB with arbitrary RIS correlation matrix are defined as
\begin{equation}
\label{equ:position_EFIM_position_exact}
\begin{aligned}
\operatorname{SOEB} &=\left[\left(\mathbf{J}_{\bm{\eta}_{L}}^{\mathrm{e}}\right)^{-1}\right]_{1,1}+\left[\left(\mathbf{J}_{\bm{\eta}_{L}}^{\mathrm{e}}\right)^{-1}\right]_{2,2}, \\
\operatorname{SPEB} &=\left[\left(\mathbf{J}_{\bm{\eta}_{L}}^{\mathrm{e}}\right)^{-1}\right]_{3,3}+\left[\left(\mathbf{J}_{\bm{\eta}_{L}}^{\mathrm{e}}\right)^{-1}\right]_{4,4}+\left[\left(\mathbf{J}_{\bm{\eta}_{L}}^{\mathrm{e}}\right)^{-1}\right]_{5,5}.
\end{aligned}
\end{equation}

\begin{assumption}
\label{asumption_far_field_rf_control_1}
To ensure that the LOS path can be separated from the RIS paths, we restrict the selection of RIS fast-varying coefficients, $\gamma_{t^{'}}^{[{m}]}$, to coefficients that sum to zero. Hence, the sequence matrix has the following property
\begin{equation}
\label{equ:far_field_rf_control_1}
\begin{aligned}
\mathbf{D}_{\gamma}^{\mathrm{H}} \mathbf{1}_{{M}_1} = \mathbf{0}.
    \end{aligned}
\end{equation}
Please note that this assumption is not very restrictive since practical  discrete codes, such as Hadamard codes and discrete Fourier matrices, satisfy it.
\end{assumption}

\subsection{Bayesian PEB and OEB: Unitary Correlation Matrices and Assumption \ref{asumption_far_field_rf_control_1}}

Under the conditions of unitary RIS correlation matrices, the restriction in (\ref{equ:far_field_rf_control_1}), and independent RIS placements; there is no mutual information between RISs and each path conveys information independently. The restriction also implies that there is no mutual information between the RISs and the LOS path. Hence, the parameters $\bm{\eta}_{1} $ can be rearranged according to paths such $ \Tilde{\bm{\eta}}_{1} \triangleq\left[ {\mathbf{\eta}}^{[1]}_{1}, {\mathbf{\eta}}^{[2]}_{1} , \cdots, {\mathbf{\eta}}^{[{M_1}]}_{1} \right]$ where 
$ 
{\bm{\eta}}^{[{m}]}_{1} \triangleq\left[
{{\theta}_{\mathrm{r}_{\mathrm{u}}}^{[{m}]}}, 
{{\phi}_{\mathrm{r}_{\mathrm{u}}}^{[{m}]}} , {{\theta}_{\mathrm{t}_{\mathrm{l}}}^{[{m}]}} , {\mathbf{\phi}_{\mathrm{t}_{\mathrm{l}}}^{[{m}]}} ,
{ {\theta}_{\mathrm{t}_{\mathrm{u}}}^{[{m}]}} ,
{ {\phi}_{\mathrm{t}_{\mathrm{u}}}^{[{m}]}} ,
{ {\tau}^{[{m}]}} _{}
  \right]^{\mathrm{T}}$\footnote{The parameter vector is not completely reparameterized. Instead, the AoIs are removed from the parameter vector. This constrains the RIS position and orientation to locations that satisfies these AoIs.}. 
The nuisance parameters can be arranged similarly $ {\bm{\eta}_{2}^{[{m}]}} \triangleq\left[{{\beta}_{\mathrm{R}}^{[{m}]}}^{\mathrm{T}},
{{\beta}_{\mathrm{I}}^{[{m}]}}^{\mathrm{T}}
\right]^{\mathrm{T}}$.
The corresponding rearranged EFIM is 
$\Bar{\mathbf{J}}_{ \Tilde{\bm{\eta}}_{1}}^{\mathrm{e}} \triangleq \text{diag}\left[ {\Bar{\mathbf{J}}_{ {{1}}^{}}^{\mathrm{e}}}, {\Bar{\mathbf{J}}_{ {{2}}^{}}^{\mathrm{e}}}, \cdots, {\Bar{\mathbf{J}}_{ {{{M}_1}}^{}}^{\mathrm{e}}} \right],$
where $\Bar{\mathbf{J}}_{ {{m}}^{}}^{\mathrm{e}}  = \Bar{\mathbf{J}}_{ {{{\bm{\eta}}^{[{m}]}_{1}
}} {\bm{\eta}}^{[{m}]}_{1}} - \Bar{\mathbf{J}}_{ {{{\bm{\eta}}^{[{m}]}_{1}
}} {\bm{\eta}}^{[{m}]}_{2}} \Bar{\mathbf{J}}_{ {{{\bm{\eta}}^{[{m}]}_{2}
}} {\bm{\eta}}^{[{m}]}_{2}}^{\mathrm{-1}} \Bar{\mathbf{J}}_{ {{{\bm{\eta}}^{[{m}]}_{1}
}} {\bm{\eta}}^{[{m}]}_{2}}^{\mathrm{T}} \; \; {m} \in \mathcal{M}_1$, and we can write
$
\Bar{\mathbf{J}}_{ \Tilde{\bm{\zeta}}_{1}}^{\mathrm{e}} \triangleq\left[\begin{array}{c|ccc}
\Bar{\mathbf{J}}_{ {\bm{\psi}}_{1}}^{\mathrm{e}}  &  \mathbf{0}  \\
\hline \mathbf{0} & \Bar{\mathbf{J}}_{ \Tilde{\bm{\eta}}_{1}}^{\mathrm{e}}  \\
\end{array}\right]
$
where $\Bar{\mathbf{J}}_{ {\bm{\psi}}_{1}}^{\mathrm{e}}$ presented in the Appendix \ref{appendix:EFIM_Los} is the EFIM obtained by applying Definition \ref{definition_EFIM_2} to the Bayesian FIM $\mathbf{J}_{ \bm{\psi}^{}_{}}^{\mathrm{}}$ of the channel parameters for the LOS path.
Accordingly, the translation matrix $\mathbf{\Upsilon}^{}$ can be written as
\begin{equation}
\label{equ:position_translation_rearranged}
\begin{aligned}
\mathbf{\Upsilon} \triangleq\left[\begin{array}{ccccc}
\overline{\mathbf{\Upsilon}}_{\mathrm{0}}& 
\overline{\mathbf{\Upsilon}}_{\mathrm{1}}&  \overline{\mathbf{\Upsilon}}_{\mathrm{2}} & \cdots & \overline{\mathbf{\Upsilon}}_{{M_1}} \\
 \mathbf{0} &  \overline{\overline{\mathbf{\Upsilon}}}_{\mathrm{1}}  & \cdots & \cdots & \mathbf{0} \\
\vdots & \vdots & \ddots& \ddots & \vdots \\
\mathbf{0} & \mathbf{0} &  \cdots &\cdots & \overline{\overline{\mathbf{\Upsilon}}}_{{M_1}}
\end{array}\right],\end{aligned}
\end{equation}
where $\overline{\mathbf{\Upsilon}}_{\mathrm{0}}$ is a $5  \times 5$ matrix relating the LOS path to the UEs orientation and position,  $\overline{\mathbf{\Upsilon}}_{m}$ is the $5 \times 7$ matrix relating the location of the ${m}^{\text{th}}$  RIS to the UEs orientation and position,  
and $\overline{\overline{\mathbf{\Upsilon}}}_{m}$ is the $5 \times 7$ matrix related  to  $ \mathbf{{o}^{}}^{[{m}]}$ and  $ \mathbf{{p}^{}}^{[{m}]}$. The correspondingly rearranged prior matrix is defined as
${\mathbf{J}}_{ {\bm{\eta}}_{L}}^{\mathrm{P}} \triangleq \text{diag}\left[ {\mathbf{J}}_{{\mathrm{UE}}}^{\mathrm{P}} , {\mathbf{J}}_{{\mathrm{1}}}^{\mathrm{P}}, \cdots, {\mathbf{J}}_{{{M_1}}}^{\mathrm{P}}  \right].$
Hence, the  Bayesian FIM for the positioning parameters is presented in (\ref{equ:positioning_FIM_2}).
\begin{figure*}
\begin{align}
\label{equ:positioning_FIM_2}
\mathbf{J}_{ \bm{\eta}_{L}} =\mathbf{\Upsilon}\Bar{\mathbf{J}}_{ \Tilde{\bm{\zeta}}_{1}}^{\mathrm{e}} \mathbf{\Upsilon}^{\mathrm{T}} + {\mathbf{J}}_{ {\mathbf{\eta}}_{L}}^{\mathrm{P}} =\left[\begin{array}{ccccc}
\sum_{{m}=0}^{{M}_1} \overline{\mathbf{\Upsilon}}_{{m}} \Bar{\mathbf{J}}_{ {{m}}^{}}^{\mathrm{e}} \overline{\mathbf{\Upsilon}}_{{m}}^{\mathrm{T}} + {\mathbf{J}}_{{\mathrm{UE}}}^{\mathrm{P}}& \overline{\mathbf{\Upsilon}}_{\mathrm{1}} \Bar{\mathbf{J}}_{ {{1}}^{}}^{\mathrm{e}} \overline{\overline{\mathbf{\Upsilon}}}_{\mathrm{1}}^{\mathrm{T}}  & \ldots & \overline{\mathbf{\Upsilon}}_{\mathrm{{M}_1}} \Bar{\mathbf{J}}_{ {{{M}_1}}^{}}^{\mathrm{e}} \overline{\overline{\mathbf{\Upsilon}}}_{{M_1}}^{\mathrm{T}}  \\
\overline{\overline{\mathbf{\Upsilon}}}_{\mathrm{1}} \Bar{\mathbf{J}}_{ {{1}}^{}}^{\mathrm{e}} {\overline{\mathbf{\Upsilon}}}_{\mathrm{1}}^{\mathrm{T}}   & 
\overline{\overline{\mathbf{\Upsilon}}}_{\mathrm{1}} \Bar{\mathbf{J}}_{ {{1}}^{}}^{\mathrm{e}} \overline{\overline{\mathbf{\Upsilon}}}_{\mathrm{1}}^{\mathrm{T}} + {\mathbf{J}}_{{\mathrm{1}}}^{\mathrm{P}} & \ldots & 0 \\
\vdots & \vdots & \ddots & \vdots \\
\overline{\overline{\mathbf{\Upsilon}}}_{{M_1}} \Bar{\mathbf{J}}_{ {{{M}_1}}^{}}^{\mathrm{e}} {\overline{\mathbf{\Upsilon}}}_{{M_1}}^{\mathrm{T}} & 0 & \ldots & 
\overline{\overline{\mathbf{\Upsilon}}}_{{M_1}} \Bar{\mathbf{J}}_{ {{{M}_1}}^{}}^{\mathrm{e}} \overline{\overline{\mathbf{\Upsilon}}}_{{M_1}}^{\mathrm{T}} + {\mathbf{J}}_{{{M_1}}}^{\mathrm{P}} 
\end{array}\right].
\end{align}
\end{figure*}

Using Definition \ref{definition_EFIM_2}, the EFIM is presented in (\ref{equ:positioning_FIM_3}).
\begin{figure*}
\begin{align}
\label{equ:positioning_FIM_3}
\mathbf{J}_{\bm{\eta}_{L}}^{\mathrm{e}} &= \overline{\mathbf{\Upsilon}}_{\mathrm{0}} \Bar{\mathbf{J}}_{ {{0}}^{}}^{\mathrm{e}} \overline{\mathbf{\Upsilon}}_{\mathrm{0}}^{\mathrm{T}} +
\sum_{m=1}^{{M}_1} \overline{\mathbf{\Upsilon}}_{{m}} \Bar{\mathbf{J}}_{ {{{m}}}^{}}^{\mathrm{e}} \overline{\mathbf{\Upsilon}}_{{m}}^{\mathrm{T}} + {\mathbf{J}}_{{\mathrm{UE}}}^{\mathrm{P}} - \sum_{m=1}^{{M}_1} \overline{\mathbf{\Upsilon}}_{{m}} \Bar{\mathbf{J}}_{ {{{m}}}^{}}^{\mathrm{e}} \overline{\overline{\mathbf{\Upsilon}}}_{{m}}^{\mathrm{T}} \left(\overline{\overline{\mathbf{\Upsilon}}}_{{m}} \Bar{\mathbf{J}}_{ {{{m}}}^{}}^{\mathrm{e}} \overline{\overline{\mathbf{\Upsilon}}}_{{m}}^{\mathrm{T}} + {\mathbf{J}}_{{{m}}}^{\mathrm{P}}  \right)^{-1} \overline{\overline{\mathbf{\Upsilon}}}_{{m}} \Bar{\mathbf{J}}_{ {{{m}}}^{}}^{\mathrm{e}} {\overline{\mathbf{\Upsilon}}}_{{m}}^{\mathrm{T}}. 
\end{align}
\end{figure*}
Hence, the FIM is partly composed of the FIM provided by the LOS path plus the FIM provided by all RIS paths.  The corresponding SPEB and the SOEB can be obtained using Equation (\ref{equ:position_EFIM_position_exact}).

\section{Numerical Results}
In this section, we evaluate the derived localization bounds with Monte Carlo simulations under different scenarios.  Without loss of generality, we assume that the slow-varying reflection coefficients, $\mathbf{\Gamma}^{[{m}]}_{q}, \; \; \forall q \in \{1,\cdots,N_Q\}$, are randomly generated, and the sequence matrix is a unitary matrix. The RIS coefficients across any two OFDM symbols are non-parallel. Hence, $\mathbf{\Gamma}^{[{m}]}_{q} \neq \mathbf{\Gamma}^{[{m}]}_{q^{'}}, \; \; \text{when } q \neq q^{'}, \; \; \forall q, q^{'} \in \{1,2,\cdots,N_Q \}.$  Note that when LOS is available, Assumption \ref{asumption_far_field_rf_control_1} will be used to ensure that the LOS is orthogonal to the paths generated by the RIS.  We focus on the case with URAs at the BS, RIS, and the UE with their respective normal vectors in the $z$ direction. 
Except stated otherwise, the UE is operating at a frequency of $30$ GHz, the transmit antenna gain is $6$ dB with a transmit power of $5 \; \text{dBm}$, the UE antenna gain is $2$ dB, $N_0 = -174 \; \text{dBm /\ Hz} $ and there are $N = 256$ subcarriers. There are $N_{T} = 4$ transmit antennas, $N_{\mathrm{B}} = N_{T}$ transmit beams, and the UE orientation offset is $(\theta_0^{}, \phi_0^{}) = (10^{\circ}, 0^{\circ})$. The considered area is $100 \; \text{m} \times 100 \; \text{m}  $, and the considered bandwidth is  $0.1 \; \text{GHz}$ with a non-specular reflection pathloss model of \cite{ellingson2021path}
$
1/{\rho^{[{m}]}} =  \frac{\lambda^4 (\cos{\theta^{[{m}]}_{\mathrm{t}_{\mathrm{l}}}} \cos{\theta^{[{m}]}_{\mathrm{r}_{\mathrm{l}}}})^{0.57}}{512\pi^2 (d^{[{m}]}_{\mathrm{r}_{\mathrm{l}}})^2  (d^{[{m}]}_{\mathrm{t}_{\mathrm{l}}})^2}$ for the RIS-paths and  $ 1/{\rho^{[0]}} =  \frac{\lambda^2}{(4\pi d^{}_{\mathrm{r}_{\mathrm{u}}} )^2}
$ for the LOS path \cite{9508872}, where $d^{}_{\mathrm{r}_{\mathrm{u}}}$, $d^{[{m}]}_{\mathrm{r}_{\mathrm{l}}}$, and $d^{[{m}]}_{\mathrm{t}_{\mathrm{l}}}$ are the distances between the BS and the UE,  between the BS and the ${m}^\text{th}$ RIS, and  between the ${m}^\text{th}$ RIS and the UE, respectively. For the case studies in this section, we assume that there is no prior information about the UE position, ${\mathbf{J}}_{{\mathrm{UE}}}^{\mathrm{P}} = \mathbf{0}_{5}$, but the prior information concerning the perturbed RISs is given as  ${\mathbf{J}}_{{{m}}}^{\mathrm{P}} = \frac{0.5}{\sigma^2}\mathbf{I}_{_{5}}, \; \; m \in \mathcal{M}_1^{b}$. The orientation offset of the perturbed RIS is given by $(\theta_0^{[{m}]}, \phi_0^{[{m}]}) = (45^{\circ}, 35^{\circ}), \; \; m \in \mathcal{M}_1^{b}$. In the subsequent figures, the Bayesian error bounds derived from (\ref{equ:positioning_FIM_3}) are plotted. The plots without the term ``LOS" refer to the error bounds with prior information on the perturbed RISs in the absence of an LOS path.   The plots with the term ``LOS" refer to the error bounds with prior information on the perturbed RISs in the presence of an LOS path.  

 In Figs.  \ref{Results:peb_oeb_varying_numb_of_receive_ris} and \ref{Results:peb_oeb_varying_numb_of_ris_ris}, the BS is located at $(0,0, 40 \text{m})$, the RIS and the UE are at a height of $35 \text{m}$ and $5 \text{m}$ respectively. In both figures, $|\mathcal{M}_1^{b}| = 1$.
\begin{figure}[htb!]
\centering
\subfloat[]{\includegraphics[scale=0.205]{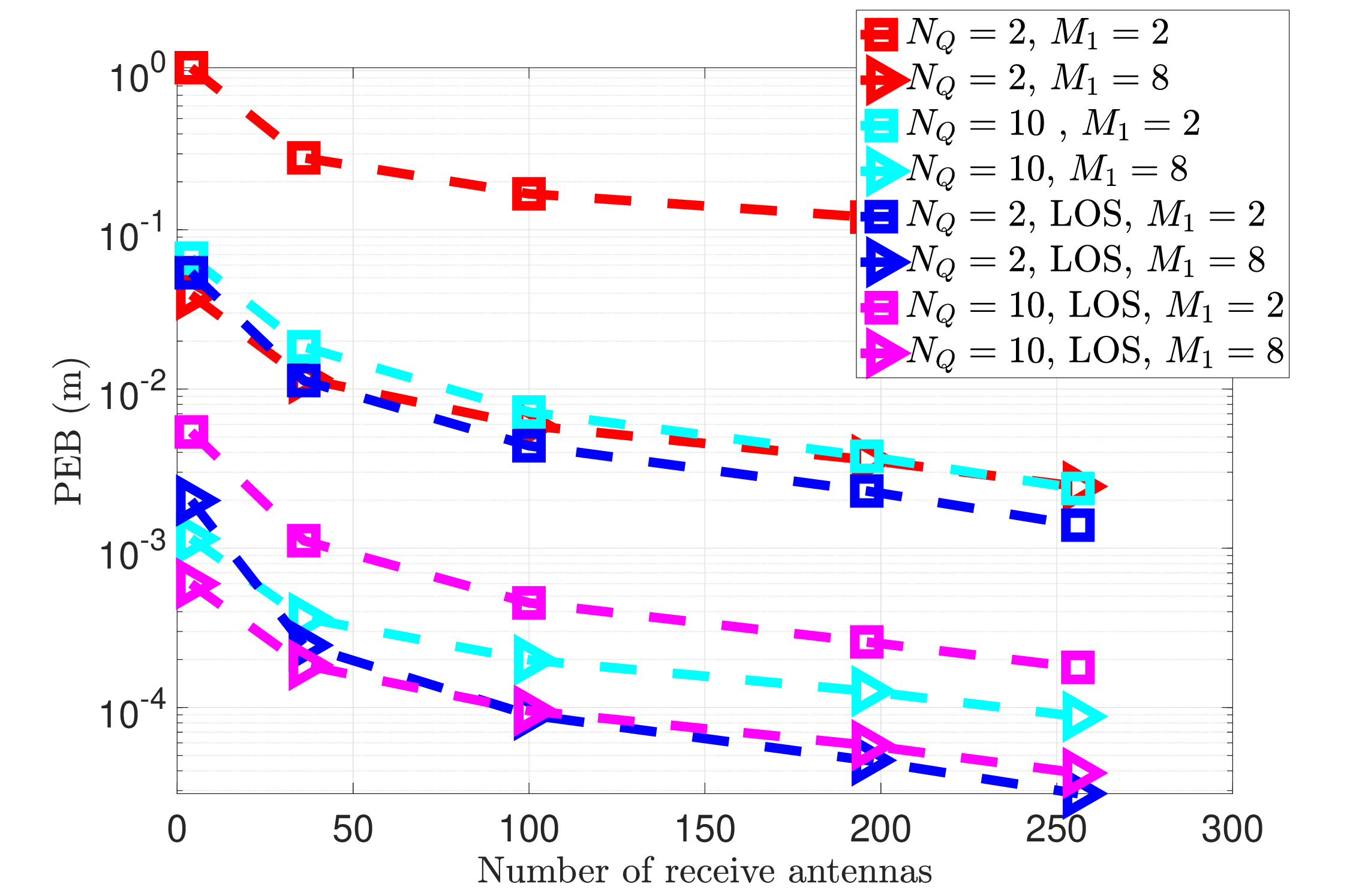}
\label{fig:Results/peb_varying_numb_of_receive_ris}}
\hfil
\subfloat[]{\includegraphics[scale=0.205]{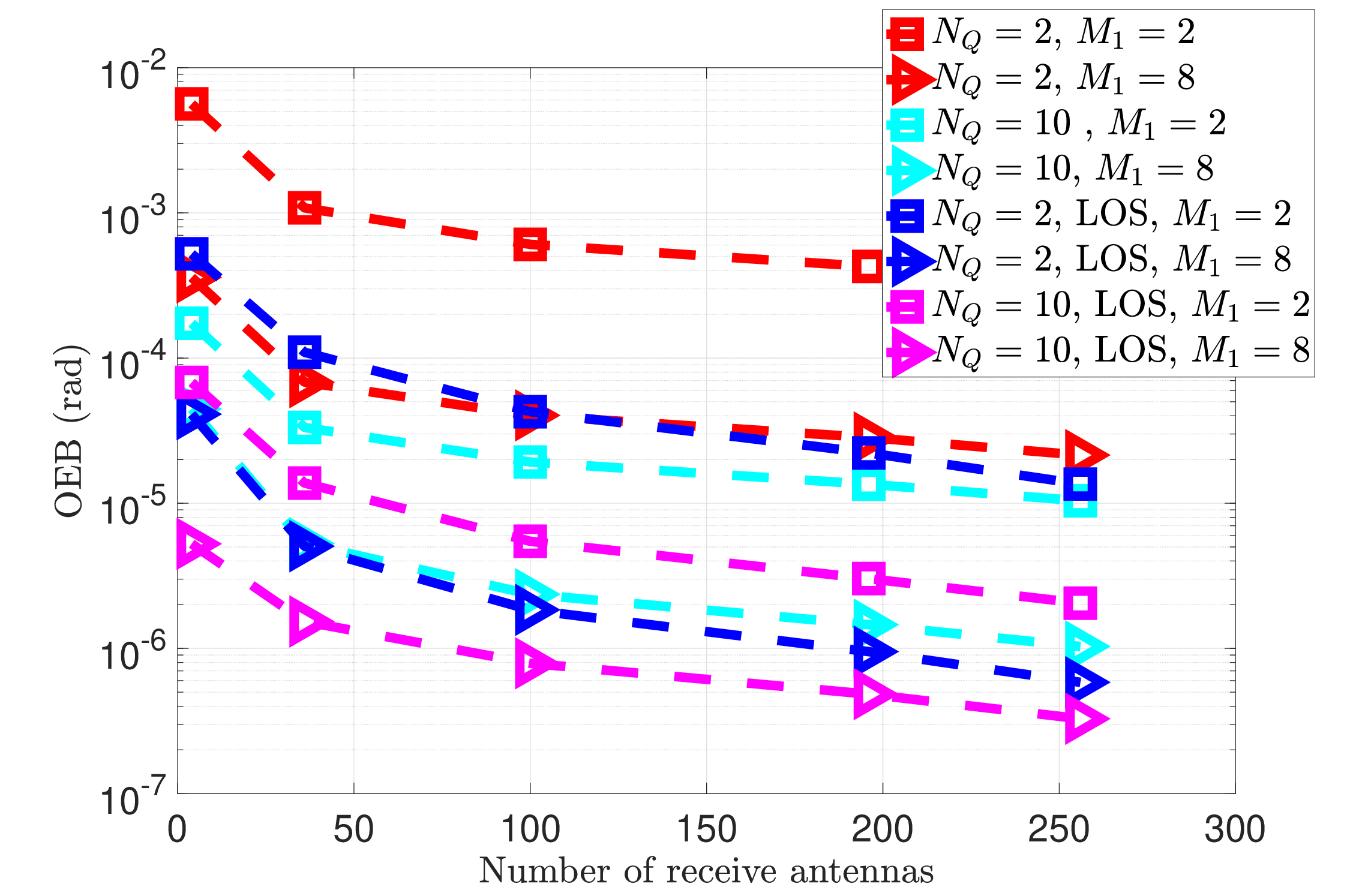}
\label{fig:Results/oeb_varying_numb_of_receive_ris}}
\caption{(a) PEB and (b) OEB with varying number receive antennas,  $|\mathcal{M}_1^{b}| = 1 \; \; \text{with} \; \; {\mathbf{J}}_{{{m}}}^{\mathrm{P}} = \frac{0.5}{\sigma^2}\mathbf{I}_{_{5}}, \; \; m \in \mathcal{M}_1^{b}.$ Each RIS has $N_L^{[m]} = 144$ elements. } 
\label{Results:peb_oeb_varying_numb_of_receive_ris}
\end{figure}
The $({x},{y})$ coordinate of both RIS and UE are randomly generated. 
\begin{figure}[htb!]
\centering
\subfloat[]{\includegraphics[scale=0.205]{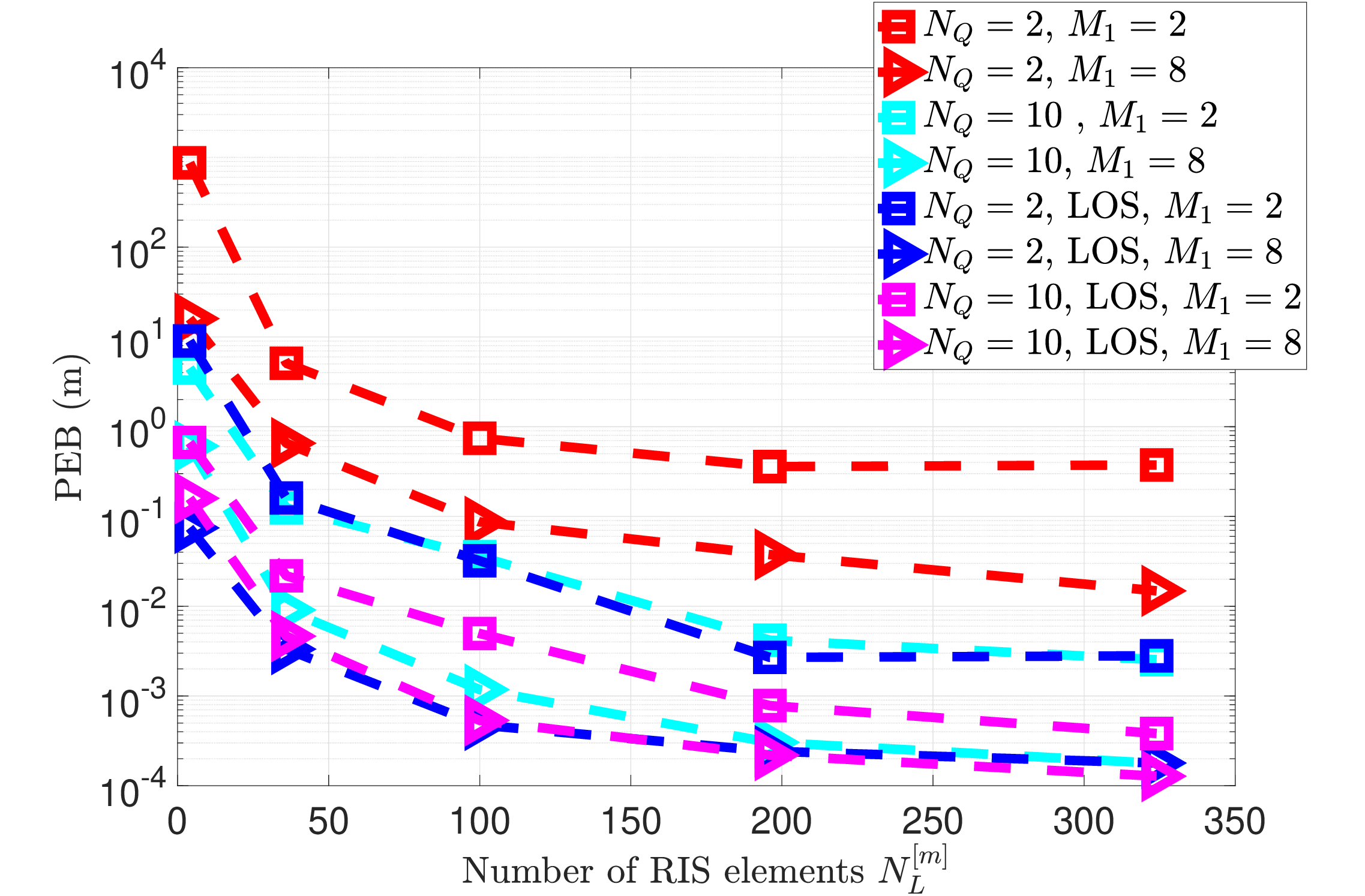}
\label{fig:Results/peb_varying_numb_of_ris_elements_ris}}
\hfil
\subfloat[]{\includegraphics[scale=0.205]{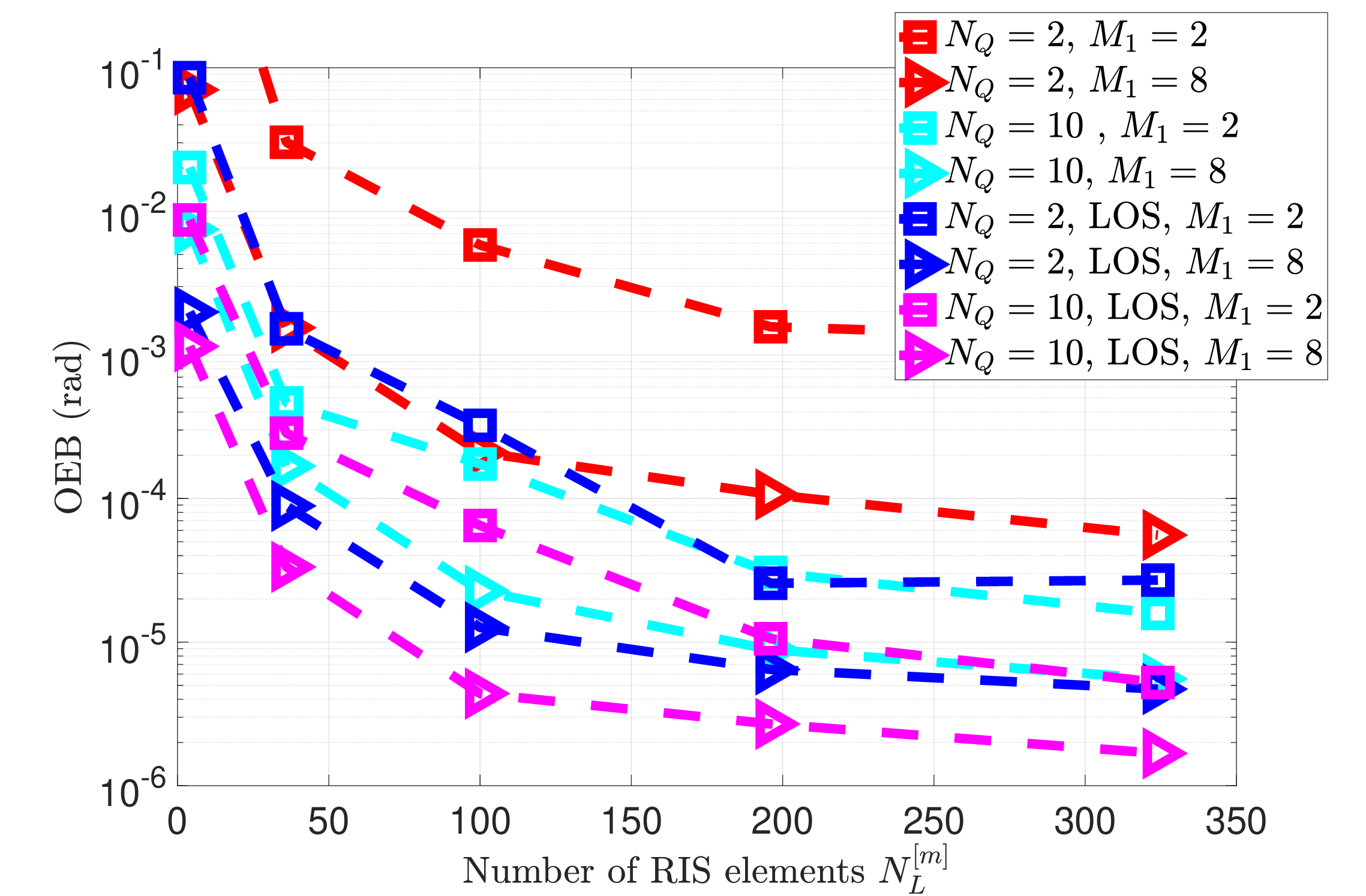}
\label{fig:Results/oeb_varying_numb_of_ris_elements_ris}}
\caption{(a) PEB and (b) OEB with varying number of RIS elements,  $|\mathcal{M}_1^{b}| = 1 \; \; \text{with} \; \; {\mathbf{J}}_{{{m}}}^{\mathrm{P}} = \frac{0.5}{\sigma^2}\mathbf{I}_{_{5}}, \; \; m \in \mathcal{M}_1^{b}.$ There are $N_R = 16$ receive antennas.}
\label{Results:peb_oeb_varying_numb_of_ris_ris}
\end{figure}
In Fig. \ref{Results:peb_oeb_varying_numb_of_receive_ris}, the localization bounds for different numbers of RISs are presented; each RIS has $N_L^{[m]} = 144$ elements. We observe that in general the PEB and OEB decrease for increasing number of OFDM symbols, $N_Q$. The PEB and OEB also decrease for increasing number of RISs. In Fig. \ref{Results:peb_oeb_varying_numb_of_receive_ris}, there is a noticeable decrease in the error bounds when there is an LOS path present. 

In Fig. \ref{Results:peb_oeb_varying_numb_of_ris_ris}, the number of receive antennas is set to $N_R = 16.$ We observe that the PEB and OEB decrease as the number of RIS elements increases. In general, we also notice that the error bounds reduce as a function of the number of OFDM symbols, $N_Q$, and the number of RISs.
It is important to emphasize that in Figs. \ref{Results:peb_oeb_varying_numb_of_receive_ris} and \ref{Results:peb_oeb_varying_numb_of_ris_ris}, we present results for the scenarios where LOS paths are present in addition to the RIS paths. These scenarios provide lower bounds for the systems without LOS. Note that the gap between the localization bounds in the RIS-only scenarios and the RIS plus LOS scenarios reduces as the number of RISs increases.  
\section{Conclusion}
In this paper, the effect of the multipath created due to multiple RISs on wireless-enabled localization has been investigated from a Bayesian perspective. This was achieved by viewing the position and orientation of the RISs as prior information to assist in downlink UE localization.  We derived the FIM for the RIS-related channel parameters and showed that the FIM can be decomposed into a sum of the FIMs provided by the OFDM symbols. The FIM provided by each OFDM symbol can be  decomposed into: i) information provided by the receiver, ii) information provided by the transmitter, and iii) information provided by the RIS components. We showed that the information provided by the RIS can be further decomposed into a correlation matrix and an information matrix representing the gains due to the RIS. Through this decomposition, we observed that with parallel RIS reflection coefficients across OFDM symbols, the information matrix produced by the RIS coefficients during the additional OFDM symbols does not provide any further information. Next, for parallel RIS coefficients, we showed through the derivations of the Bayesian EFIM that all information about the RIS-related angle channel parameters is lost when the complex path gains are unknown. We noted that this loss of information has severe implications and can hinder localization. Next, we transformed the Bayesian EFIM of the geometric channel parameters to the Bayesian FIM for localization. We obtained the Bayesian EFIM of localization by considering RISs, with position and orientation offsets. As a future work, we intend to simultaneously analyze the impact of different UE architectures and RIS location offset on localization performance. 
	\appendix 
	\subsection{ Entries of the  FIM }
	The entries of the FIM in (\ref{equ:FIM_parameter_matrix}) are presented below
	\label{appendix:FIM_entries}
	\begin{subequations}[equ:FIM_exact_submat_theta_r_u]
\begin{align}
\begin{split}
\medmath{{\mathbf{J}_{\bm{\theta}_{\mathrm{r}_{\mathrm{u}}} \bm{\phi}_{\mathrm{r}_{\mathrm{u}}}}}}
 &=   \medmath{\frac{2}{\sigma^2}\sum_{q = 1}^{N_{Q}}\Re }\left\{\medmath{\left(\mathbf{B}^{\mathrm{H}} \mathbf{K}_{\mathrm{r}_{\mathrm{u}}}^{\mathrm{H}} \mathbf{P}_{\mathrm{r}_{\mathrm{u}}} \mathbf{B}\right) \odot\left(\mathbf{k}_{q,\mathrm{l}} \mathbf{k}_{q,\mathrm{l}}^{\mathrm{H}}\right)
 \odot\left(\mathbf{D}_{\mathrm{\gamma}}^{\mathrm{H}} \mathbf{D}_{\mathrm{\gamma}}\right)} \right. \\ & \left.  \medmath{\odot \left(\mathbf{A}_{\mathrm{t}_{\mathrm{u}}}^{\mathrm{H}}  \mathbf{F}   \mathbf{F}^{\mathrm{H}} \mathbf{A}_{\mathrm{t}_{\mathrm{u}}}^{}\right)^{\mathrm{T}} \odot \mathbf{R}_{0}}\right\}
  \label{equ:FIM_exact_submat_theta_r_u_phi_r_u} \end{split}
 \\
 \begin{split}
\medmath{\mathbf{J}_{\bm{\theta}_{\mathrm{r}_{\mathrm{u}}} \bm{\theta}_{\mathrm{t}_{\mathrm{l}}}}}
 &=   \medmath{-\frac{2}{\sigma^2}\sum_{q = 1}^{N_{Q}}\Re}\left\{\medmath{\left(\mathbf{B}^{\mathrm{H}} \mathbf{K}_{\mathrm{r}_{\mathrm{u}}}^{\mathrm{H}} \mathbf{A}_{\mathrm{r}_{\mathrm{u}}} \mathbf{B}\right) \odot\left(\mathbf{k}_{q,\mathrm{l}} \mathbf{k}_{q,\mathrm{t}_{\mathrm{l}}}^{\mathrm{H}}\right)
 \odot\left(\mathbf{D}_{\mathrm{\gamma}}^{\mathrm{H}} \mathbf{D}_{\mathrm{\gamma}}\right) } \right. \\ & \left. \medmath{
  \odot\left(\mathbf{A}_{\mathrm{t}_{\mathrm{u}}}^{\mathrm{H}}  \mathbf{F}   \mathbf{F}^{\mathrm{H}} \mathbf{A}_{\mathrm{t}_{\mathrm{u}}}^{}\right)^{\mathrm{T}} \odot \mathbf{R}_{0}}\right\} \label{equ:FIM_exact_submat_theta_r_u_theta_t_l} \end{split}\\
   \begin{split}
 \medmath{\mathbf{J}_{\bm{\theta}_{\mathrm{r}_{\mathrm{u}}} \bm{\phi}_{\mathrm{t}_{\mathrm{l}}}}}
 &=   \medmath{-\frac{2}{\sigma^2}\sum_{q = 1}^{N_{Q}}\Re}\left\{\medmath{\left(\mathbf{B}^{\mathrm{H}} \mathbf{K}_{\mathrm{r}_{\mathrm{u}}}^{\mathrm{H}} \mathbf{A}_{\mathrm{r}_{\mathrm{u}}} \mathbf{B}\right) }\odot\left(\mathbf{k}_{q,\mathrm{l}} \mathbf{p}_{q,\mathrm{t}_{\mathrm{l}}}^{\mathrm{H}}\right)
 \odot\left(\mathbf{D}_{\mathrm{\gamma}}^{\mathrm{H}} \mathbf{D}_{\mathrm{\gamma}}\right)   \right. \\ & \left. \medmath{\odot\left(\mathbf{A}_{\mathrm{t}_{\mathrm{u}}}^{\mathrm{H}}  \mathbf{F}   \mathbf{F}^{\mathrm{H}} \mathbf{A}_{\mathrm{t}_{\mathrm{u}}}^{}\right)^{\mathrm{T}} \odot \mathbf{R}_{0}}\right\} \label{equ:FIM_exact_submat_theta_r_u_phi_t_l}  \end{split}
\\
  \begin{split}
\medmath{\mathbf{J}_{\bm{\theta}_{\mathrm{r}_{\mathrm{u}}} \bm{\theta}_{\mathrm{r}_{\mathrm{l}}}}}
 &=  \medmath{\frac{2}{\sigma^2}\sum_{q = 1}^{N_{Q}}\Re}\left\{\medmath{\left(\mathbf{B}^{\mathrm{H}} \mathbf{K}_{\mathrm{r}_{\mathrm{u}}}^{\mathrm{H}} \mathbf{A}_{\mathrm{r}_{\mathrm{u}}} \mathbf{B}\right) \odot\left(\mathbf{k}_{q,\mathrm{l}} \mathbf{k}_{q,\mathrm{r}_{\mathrm{l}}}^{\mathrm{H}}\right)
\odot\left(\mathbf{D}_{\mathrm{\gamma}}^{\mathrm{H}} \mathbf{D}_{\mathrm{\gamma}}\right)} \right. \\ & \left. \medmath{\odot\left(\mathbf{A}_{\mathrm{t}_{\mathrm{u}}}^{\mathrm{H}}  \mathbf{F}   \mathbf{F}^{\mathrm{H}} \mathbf{A}_{\mathrm{t}_{\mathrm{u}}}^{}\right)^{\mathrm{T}} \odot \mathbf{R}_{0}}\right\},  \label{equ:FIM_exact_submat_theta_r_u_theta_r_l}  \end{split}
\\
\medmath{\mathbf{J}_{\bm{\theta}_{\mathrm{r}_{\mathrm{u}}} \bm{\phi}_{\mathrm{r}_{\mathrm{l}}}}}
 &=   \medmath{\frac{2}{\sigma^2}\sum_{q = 1}^{N_{Q}}\Re}\left\{\medmath{\left(\mathbf{B}^{\mathrm{H}} \mathbf{K}_{\mathrm{r}_{\mathrm{u}}}^{\mathrm{H}} \mathbf{A}_{\mathrm{r}_{\mathrm{u}}} \mathbf{B}\right) \odot\left(\mathbf{k}_{q,\mathrm{l}} \mathbf{p}_{q,\mathrm{r}_{\mathrm{l}}}^{\mathrm{H}}\right)
 \odot\left(\mathbf{D}_{\mathrm{\gamma}}^{\mathrm{H}} \mathbf{D}_{\mathrm{\gamma}}\right)} \right. \\ & \left.
  \odot\medmath{\left(\mathbf{A}_{\mathrm{t}_{\mathrm{u}}}^{\mathrm{H}}  \mathbf{F}   \mathbf{F}^{\mathrm{H}} \mathbf{A}_{\mathrm{t}_{\mathrm{u}}}^{}\right)^{\mathrm{T}} \odot \mathbf{R}_{0}}\right\},  \label{equ:FIM_exact_submat_theta_r_u_phi_r_l} \\
\medmath{\mathbf{J}_{\bm{\theta}_{\mathrm{r}_{\mathrm{u}}} \bm{\theta}_{\mathrm{t}_{\mathrm{u}}}}}
 &=   \medmath{-\frac{2}{\sigma^2}\sum_{q = 1}^{N_{Q}}\Re}\left\{\medmath{\left(\mathbf{B}^{\mathrm{H}} \mathbf{K}_{\mathrm{r}_{\mathrm{u}}}^{\mathrm{H}} \mathbf{A}_{\mathrm{r}_{\mathrm{u}}} \mathbf{B}\right) \odot\left(\mathbf{k}_{q,\mathrm{l}} \mathbf{k}_{q,\mathrm{l}}^{\mathrm{H}}\right)
 \odot\left(\mathbf{D}_{\mathrm{\gamma}}^{\mathrm{H}} \mathbf{D}_{\mathrm{\gamma}}\right)} \right. \\ & \left.  \medmath{\odot\left(\mathbf{K}_{\mathrm{t}_{\mathrm{u}}}^{\mathrm{H}}  \mathbf{F}   \mathbf{F}^{\mathrm{H}} \mathbf{A}_{\mathrm{t}_{\mathrm{u}}}^{}\right)^{\mathrm{T}} \odot \mathbf{R}_{0}}\right\},  \label{equ:FIM_exact_submat_theta_r_u_theta_t_u} \\
 \medmath{\mathbf{J}_{\bm{\theta}_{\mathrm{r}_{\mathrm{u}}} \bm{\phi}_{\mathrm{t}_{\mathrm{u}}}}}
 &=   \medmath{-\frac{2}{\sigma^2}\sum_{q = 1}^{N_{Q}}\Re}\left\{\medmath{\left(\mathbf{B}^{\mathrm{H}} \mathbf{K}_{\mathrm{r}_{\mathrm{u}}}^{\mathrm{H}} \mathbf{A}_{\mathrm{r}_{\mathrm{u}}} \mathbf{B}\right) \odot\left(\mathbf{k}_{q,\mathrm{l}} \mathbf{k}_{q,\mathrm{l}}^{\mathrm{H}}\right)
 \odot\left(\mathbf{D}_{\mathrm{\gamma}}^{\mathrm{H}} \mathbf{D}_{\mathrm{\gamma}}\right)} \right. \\ & \left. \medmath{
  \odot\left(\mathbf{P}_{\mathrm{t}_{\mathrm{u}}}^{\mathrm{H}}  \mathbf{F}   \mathbf{F}^{\mathrm{H}} \mathbf{A}_{\mathrm{t}_{\mathrm{u}}}^{}\right)^{\mathrm{T}} \odot \mathbf{R}_{0}}\right\},  \label{equ:FIM_exact_submat_theta_r_u_phi_t_u} \\
  \begin{split}
\medmath{\mathbf{J}_{\bm{\theta}_{\mathrm{r}_{\mathrm{u}}} \bm{\tau}}}
 &=   \medmath{\frac{2}{\sigma^2}\sum_{q = 1}^{N_{Q}}\Re}\left\{\medmath{\left(\mathbf{B}^{\mathrm{H}} \mathbf{K}_{\mathrm{r}_{\mathrm{u}}}^{\mathrm{H}} \mathbf{A}_{\mathrm{r}_{\mathrm{u}}} \mathbf{B}\right) \odot\left(\mathbf{k}_{q,\mathrm{l}} \mathbf{k}_{q,\mathrm{l}}^{\mathrm{H}}\right)
 \odot\left(\mathbf{D}_{\mathrm{\gamma}}^{\mathrm{H}} \mathbf{D}_{\mathrm{\gamma}}\right)}
 \right. \\ & \left. 
  \medmath{\odot\left(\mathbf{A}_{\mathrm{t}_{\mathrm{u}}}^{\mathrm{H}}  \mathbf{F}   \mathbf{F}^{\mathrm{H}} \mathbf{A}_{\mathrm{t}_{\mathrm{u}}}^{}\right)^{\mathrm{T}} \odot \mathbf{R}_{1}}\right\}, \label{equ:FIM_exact_submat_theta_r_u_tau} \end{split} \\
  \begin{split}
\medmath{\mathbf{J}_{\bm{\theta}_{\mathrm{r}_{\mathrm{u}}} \bm{\beta}_{\mathrm{I}}}}  & \medmath{+j\mathbf{J}_{\bm{\theta}_{\mathrm{r}_{\mathrm{u}}} \bm{\beta}_{\mathrm{R}}}}
 =    \medmath{-\frac{2}{\sigma^2}\sum_{q = 1}^{N_{Q}}}\left\{\medmath{\left(\mathbf{B}^{\mathrm{H}} \mathbf{K}_{\mathrm{r}_{\mathrm{u}}}^{\mathrm{H}} \mathbf{A}_{\mathrm{r}_{\mathrm{u}}} \right) \odot\left(\mathbf{k}_{q,\mathrm{l}} \mathbf{k}_{q,\mathrm{l}}^{\mathrm{H}}\right)
 \odot\left(\mathbf{D}_{\mathrm{\gamma}}^{\mathrm{H}} \mathbf{D}_{\mathrm{\gamma}}\right)} \right. \\ & \left. 
 \odot \medmath{\left(\mathbf{A}_{\mathrm{t}_{\mathrm{u}}}^{\mathrm{H}}  \mathbf{F}   \mathbf{F}^{\mathrm{H}} \mathbf{A}_{\mathrm{t}_{\mathrm{u}}}^{}\right)^{\mathrm{T}} \odot \mathbf{R}_{0}}\right\}   \label{equ:FIM_exact_submat_theta_r_u_beta}
   \end{split}. 
\end{align}
\end{subequations}
The submatrices related to $\bm{\phi}_{\mathrm{r}_{\mathrm{u}}}$ can be obtained by replacing $\mathbf{K}_{\mathrm{r}_{\mathrm{u}}}$ with $\mathbf{P}_{\mathrm{r}_{\mathrm{u}}}$ in the appropriate equations. The submatrices related to $\bm{\theta}_{\mathrm{t}_{\mathrm{l}}}$ are derived next
	\begin{subequations}[equ:FIM_exact_submat_theta_t_l]
\begin{align}
\begin{split}
\medmath{\mathbf{J}_{\bm{\theta}_{\mathrm{t}_{\mathrm{l}}} \bm{\theta}_{\mathrm{t}_{\mathrm{l}}}}}
 &=   \medmath{\frac{2}{\sigma^2}\sum_{q = 1}^{N_{Q}}\Re} \left\{\medmath{\left(\mathbf{B}^{\mathrm{H}} \mathbf{A}_{\mathrm{r}_{\mathrm{u}}}^{\mathrm{H}} \mathbf{A}_{\mathrm{r}_{\mathrm{u}}} \mathbf{B}\right) \odot\left(\mathbf{k}_{q,\mathrm{t}_{\mathrm{l}}} \mathbf{k}_{q,\mathrm{t}_{\mathrm{l}}}^{\mathrm{H}}\right)
 \odot\left(\mathbf{D}_{\mathrm{\gamma}}^{\mathrm{H}} \mathbf{D}_{\mathrm{\gamma}}\right) } \right. \\ & \left. \medmath{\odot\left(\mathbf{A}_{\mathrm{t}_{\mathrm{u}}}^{\mathrm{H}}  \mathbf{F}   \mathbf{F}^{\mathrm{H}} \mathbf{A}_{\mathrm{t}_{\mathrm{u}}}^{}\right)^{\mathrm{T}} \odot \mathbf{R}_{0}}\right\} \label{equ:FIM_exact_submat_theta_t_l_theta_t_l} 
 \end{split}
\\
\begin{split}
\medmath{\mathbf{J}_{\bm{\theta}_{\mathrm{t}_{\mathrm{l}}} \bm{\phi}_{\mathrm{t}_{\mathrm{l}}}}}
 &=   \medmath{\frac{2}{\sigma^2}\sum_{q = 1}^{N_{Q}}\Re}\left\{\medmath{\left(\mathbf{B}^{\mathrm{H}} \mathbf{A}_{\mathrm{r}_{\mathrm{u}}}^{\mathrm{H}} \mathbf{A}_{\mathrm{r}_{\mathrm{u}}} \mathbf{B}\right) \odot\left(\mathbf{k}_{q,\mathrm{t}_{\mathrm{l}}} \mathbf{p}_{q,\mathrm{t}_{\mathrm{l}}}^{\mathrm{H}}\right)
 \odot\left(\mathbf{D}_{\mathrm{\gamma}}^{\mathrm{H}} \mathbf{D}_{\mathrm{\gamma}}\right)}
 \right. \\ & \left.  \medmath{\odot\left(\mathbf{A}_{\mathrm{t}_{\mathrm{u}}}^{\mathrm{H}}  \mathbf{F}   \mathbf{F}^{\mathrm{H}} \mathbf{A}_{\mathrm{t}_{\mathrm{u}}}^{}\right)^{\mathrm{T}} \odot \mathbf{R}_{0}}\right\} \label{equ:FIM_exact_submat_theta_t_l_phi_t_l} 
  \end{split}
\\
\begin{split}
\medmath{\mathbf{J}_{\bm{\theta}_{\mathrm{t}_{\mathrm{l}}} \bm{\theta}_{\mathrm{r}_{\mathrm{l}}}}}
 &=   \medmath{-\frac{2}{\sigma^2}\sum_{q = 1}^{N_{Q}}\Re}\left\{\medmath{\left(\mathbf{B}^{\mathrm{H}} \mathbf{A}_{\mathrm{r}_{\mathrm{u}}}^{\mathrm{H}} \mathbf{A}_{\mathrm{r}_{\mathrm{u}}} \mathbf{B}\right) \odot\left(\mathbf{k}_{q,\mathrm{t}_{\mathrm{l}}} \mathbf{k}_{q,\mathrm{r}_{\mathrm{l}}}^{\mathrm{H}}\right)
 \odot\left(\mathbf{D}_{\mathrm{\gamma}}^{\mathrm{H}} \mathbf{D}_{\mathrm{\gamma}}\right)}  \right. \\ & \left.  \medmath{
  \odot\left(\mathbf{A}_{\mathrm{t}_{\mathrm{u}}}^{\mathrm{H}}  \mathbf{F}   \mathbf{F}^{\mathrm{H}} \mathbf{A}_{\mathrm{t}_{\mathrm{u}}}^{}\right)^{\mathrm{T}} \odot \mathbf{R}_{0}}\right\} \label{equ:FIM_exact_submat_theta_t_l_theta_r_l}   \end{split}
\\
\begin{split}
\medmath{\mathbf{J}_{\bm{\theta}_{\mathrm{t}_{\mathrm{l}}} \bm{\phi}_{\mathrm{r}_{\mathrm{l}}}}}
 &=   \medmath{-\frac{2}{\sigma^2}\sum_{q = 1}^{N_{Q}}\Re} \left\{\medmath{\left(\mathbf{B}^{\mathrm{H}} \mathbf{A}_{\mathrm{r}_{\mathrm{u}}}^{\mathrm{H}} \mathbf{A}_{\mathrm{r}_{\mathrm{u}}} \mathbf{B}\right) \odot\left(\mathbf{k}_{q,\mathrm{t}_{\mathrm{l}}} \mathbf{p}_{q,\mathrm{r}_{\mathrm{l}}}^{\mathrm{H}}\right)
 \odot\left(\mathbf{D}_{\mathrm{\gamma}}^{\mathrm{H}} \mathbf{D}_{\mathrm{\gamma}}\right)} 
   \right. \\ & \left. 
 \medmath{\odot\left(\mathbf{A}_{\mathrm{t}_{\mathrm{u}}}^{\mathrm{H}}  \mathbf{F}   \mathbf{F}^{\mathrm{H}} \mathbf{A}_{\mathrm{t}_{\mathrm{u}}}^{}\right)^{\mathrm{T}} \odot \mathbf{R}_{0}}\right\},  \label{equ:FIM_exact_submat_theta_t_l_phi_r_l} \end{split}
\\
\begin{split}
\medmath{\mathbf{J}_{\bm{\theta}_{\mathrm{t}_{\mathrm{l}}} \bm{\theta}_{\mathrm{t}_{\mathrm{u}}}}}
 &=   \medmath{\frac{2}{\sigma^2}\sum_{q = 1}^{N_{Q}}\Re}\left\{\medmath{\left(\mathbf{B}^{\mathrm{H}} \mathbf{A}_{\mathrm{r}_{\mathrm{u}}}^{\mathrm{H}} \mathbf{A}_{\mathrm{r}_{\mathrm{u}}} \mathbf{B}\right) \odot\left(\mathbf{k}_{q,\mathrm{t}_{\mathrm{l}}} \mathbf{k}_{q,\mathrm{l}}^{\mathrm{H}}\right)
 \odot\left(\mathbf{D}_{\mathrm{\gamma}}^{\mathrm{H}} \mathbf{D}_{\mathrm{\gamma}}\right)}    \right. \\ & \left. \medmath{
 \odot\left(\mathbf{K}_{\mathrm{t}_{\mathrm{u}}}^{\mathrm{H}}  \mathbf{F}   \mathbf{F}^{\mathrm{H}} \mathbf{A}_{\mathrm{t}_{\mathrm{u}}}^{}\right)^{\mathrm{T}} \odot \mathbf{R}_{0}}\right\},  \label{equ:FIM_exact_submat_theta_t_l_theta_t_u}
 \end{split}
\\
\begin{split}
\medmath{\mathbf{J}_{\bm{\theta}_{\mathrm{t}_{\mathrm{l}}} \bm{\phi}_{\mathrm{t}_{\mathrm{u}}}}}
 &=   \medmath{\frac{2}{\sigma^2}\sum_{q = 1}^{N_{Q}}\Re}\left\{\medmath{\left(\mathbf{B}^{\mathrm{H}} \mathbf{A}_{\mathrm{r}_{\mathrm{u}}}^{\mathrm{H}} \mathbf{A}_{\mathrm{r}_{\mathrm{u}}} \mathbf{B}\right) \odot\left(\mathbf{k}_{q,\mathrm{t}_{\mathrm{l}}} \mathbf{k}_{q,\mathrm{l}}^{\mathrm{H}}\right)
 \odot\left(\mathbf{D}_{\mathrm{\gamma}}^{\mathrm{H}} \mathbf{D}_{\mathrm{\gamma}}\right)}\right. \\ & \left.\medmath{
 \odot\left(\mathbf{P}_{\mathrm{t}_{\mathrm{u}}}^{\mathrm{H}}  \mathbf{F}   \mathbf{F}^{\mathrm{H}} \mathbf{A}_{\mathrm{t}_{\mathrm{u}}}^{}\right)^{\mathrm{T}} \odot \mathbf{R}_{0}}\right\},  \label{equ:FIM_exact_submat_theta_t_l_phi_t_u} \end{split} \\
 \begin{split}
\medmath{\mathbf{J}_{\bm{\theta}_{\mathrm{t}_{\mathrm{l}}} \bm{\tau}}}
 &=   \medmath{-\frac{2}{\sigma^2}\sum_{q = 1}^{N_{Q}}\Re}\left\{\medmath{\left(\mathbf{B}^{\mathrm{H}} \mathbf{A}_{\mathrm{r}_{\mathrm{u}}}^{\mathrm{H}} \mathbf{A}_{\mathrm{r}_{\mathrm{u}}} \mathbf{B}\right) \odot\left(\mathbf{k}_{q,\mathrm{t}_{\mathrm{l}}} \mathbf{k}_{q,\mathrm{l}}^{\mathrm{H}}\right)
 \odot\left(\mathbf{D}_{\mathrm{\gamma}}^{\mathrm{H}} \mathbf{D}_{\mathrm{\gamma}}\right)}\right. \\ & \left. \medmath{
  \odot\left(\mathbf{A}_{\mathrm{t}_{\mathrm{u}}}^{\mathrm{H}}  \mathbf{F}   \mathbf{F}^{\mathrm{H}} \mathbf{A}_{\mathrm{t}_{\mathrm{u}}}^{}\right)^{\mathrm{T}} \odot \mathbf{R}_{1}}\right\}, \label{equ:FIM_exact_submat_theta_t_l_tau} \end{split}\\
  \begin{split}
\medmath{\mathbf{J}_{\bm{\theta}_{\mathrm{t}_{\mathrm{l}}} \bm{\beta}_{\mathrm{I}}}} &\medmath{+j\mathbf{J}_{\bm{\theta}_{\mathrm{t}_{\mathrm{l}}} \bm{\beta}_{\mathrm{R}}}}
 =   \medmath{\frac{2}{\sigma^2}\sum_{q = 1}^{N_{Q}}} \left\{\medmath{\left(\mathbf{B}^{\mathrm{H}} \mathbf{A}_{\mathrm{r}_{\mathrm{u}}}^{\mathrm{H}} \mathbf{A}_{\mathrm{r}_{\mathrm{u}}} \right) \odot\left(\mathbf{k}_{q,\mathrm{t}_{\mathrm{l}}} \mathbf{k}_{q,\mathrm{l}}^{\mathrm{H}}\right)}\right. \\ & \left. \medmath{
 \odot\left(\mathbf{D}_{\mathrm{\gamma}}^{\mathrm{H}} \mathbf{D}_{\mathrm{\gamma}}\right) 
  \odot\left(\mathbf{A}_{\mathrm{t}_{\mathrm{u}}}^{\mathrm{H}}  \mathbf{F}   \mathbf{F}^{\mathrm{H}} \mathbf{A}_{\mathrm{t}_{\mathrm{u}}}^{}\right)^{\mathrm{T}} \odot \mathbf{R}_{0}}\right\}. \label{equ:FIM_exact_submat_theta_t_l_beta}
  \end{split}
\end{align}
\end{subequations}
The submatrices related to $\bm{\phi}_{\mathrm{t}_{\mathrm{l}}}$ can be obtained by replacing $\mathbf{k}_{q,\mathrm{t}_{\mathrm{l}}}$ with $\mathbf{p}_{q,\mathrm{t}_{\mathrm{l}}}$ in the appropriate equations. The submatrices related to $\bm{\theta}_{\mathrm{r}_{\mathrm{l}}}$ are obtained as
\begin{subequations}[equ:FIM_exact_submat_theta_r_l]
\begin{align}
  \begin{split}
\medmath{\mathbf{J}_{\bm{\theta}_{\mathrm{r}_{\mathrm{l}}} \bm{\theta}_{\mathrm{r}_{\mathrm{l}}}}}
 &=   \medmath{\frac{2}{\sigma^2}\sum_{q = 1}^{N_{Q}}\Re} \left\{\medmath{\left(\mathbf{B}^{\mathrm{H}} \mathbf{A}_{\mathrm{r}_{\mathrm{u}}}^{\mathrm{H}} \mathbf{A}_{\mathrm{r}_{\mathrm{u}}} \mathbf{B}\right) \odot\left(\mathbf{k}_{q,\mathrm{r}_{\mathrm{l}}} \mathbf{k}_{q,\mathrm{r}_{\mathrm{l}}}^{\mathrm{H}}\right)
 \odot\left(\mathbf{D}_{\mathrm{\gamma}}^{\mathrm{H}} \mathbf{D}_{\mathrm{\gamma}}\right)} \right. \\ & \left.  \medmath{
 \odot\left(\mathbf{A}_{\mathrm{t}_{\mathrm{u}}}^{\mathrm{H}}  \mathbf{F}   \mathbf{F}^{\mathrm{H}} \mathbf{A}_{\mathrm{t}_{\mathrm{u}}}^{}\right)^{\mathrm{T}} \odot \mathbf{R}_{0}}\right\}, \label{equ:FIM_exact_submat_theta_r_l_theta_r_l}   \end{split}
\\
\begin{split}
\medmath{\mathbf{J}_{\bm{\theta}_{\mathrm{r}_{\mathrm{l}}} \bm{\phi}_{\mathrm{r}_{\mathrm{l}}}}}
 &=   \medmath{\frac{2}{\sigma^2}\sum_{q = 1}^{N_{Q}}\Re} \left\{\medmath{\left(\mathbf{B}^{\mathrm{H}} \mathbf{A}_{\mathrm{r}_{\mathrm{u}}}^{\mathrm{H}} \mathbf{A}_{\mathrm{r}_{\mathrm{u}}} \mathbf{B}\right) \odot\left(\mathbf{k}_{q,\mathrm{r}_{\mathrm{l}}} \mathbf{p}_{q,\mathrm{r}_{\mathrm{l}}}^{\mathrm{H}}\right)
 \odot\left(\mathbf{D}_{\mathrm{\gamma}}^{\mathrm{H}} \mathbf{D}_{\mathrm{\gamma}}\right)} \right. \\ & \left. 
  \medmath{\odot\left(\mathbf{A}_{\mathrm{t}_{\mathrm{u}}}^{\mathrm{H}}  \mathbf{F}   \mathbf{F}^{\mathrm{H}} \mathbf{A}_{\mathrm{t}_{\mathrm{u}}}^{}\right)^{\mathrm{T}} \odot \mathbf{R}_{0}}\right\}, \label{equ:FIM_exact_submat_theta_r_l_phi_r_l} \end{split}\\
  \begin{split}
\medmath{\mathbf{J}_{\bm{\theta}_{\mathrm{r}_{\mathrm{l}}} \bm{\theta}_{\mathrm{t}_{\mathrm{u}}}}}
 &=  \medmath{ -\frac{2}{\sigma^2}\sum_{q = 1}^{N_{Q}}\Re}\left\{\medmath{\left(\mathbf{B}^{\mathrm{H}} \mathbf{A}_{\mathrm{r}_{\mathrm{u}}}^{\mathrm{H}} \mathbf{A}_{\mathrm{r}_{\mathrm{u}}} \mathbf{B}\right) \odot\left(\mathbf{k}_{q,\mathrm{r}_{\mathrm{l}}} \mathbf{k}_{q,\mathrm{l}}^{\mathrm{H}}\right)
 \odot\left(\mathbf{D}_{\mathrm{\gamma}}^{\mathrm{H}} \mathbf{D}_{\mathrm{\gamma}}\right)}  \right. \\ & \left.  \medmath{
  \odot\left(\mathbf{K}_{\mathrm{t}_{\mathrm{u}}}^{\mathrm{H}}  \mathbf{F}   \mathbf{F}^{\mathrm{H}} \mathbf{A}_{\mathrm{t}_{\mathrm{u}}}^{}\right)^{\mathrm{T}} \odot \mathbf{R}_{0}}\right\}, \label{equ:FIM_exact_submat_theta_r_l_theta_t_u} \end{split}\\
  \begin{split}
\medmath{\mathbf{J}_{\bm{\theta}_{\mathrm{r}_{\mathrm{l}}} \bm{\phi}_{\mathrm{t}_{\mathrm{u}}}}}
 &=   \medmath{-\frac{2}{\sigma^2}\sum_{q = 1}^{N_{Q}}\Re}\left\{\medmath{\left(\mathbf{B}^{\mathrm{H}} \mathbf{A}_{\mathrm{r}_{\mathrm{u}}}^{\mathrm{H}} \mathbf{A}_{\mathrm{r}_{\mathrm{u}}} \mathbf{B}\right) \odot\left(\mathbf{k}_{q,\mathrm{r}_{\mathrm{l}}} \mathbf{k}_{q,\mathrm{l}}^{\mathrm{H}}\right)
 \odot\left(\mathbf{D}_{\mathrm{\gamma}}^{\mathrm{H}} \mathbf{D}_{\mathrm{\gamma}}\right)
  }  \right. \\ & \left.  \medmath{\odot\left(\mathbf{P}_{\mathrm{t}_{\mathrm{u}}}^{\mathrm{H}}  \mathbf{F}   \mathbf{F}^{\mathrm{H}} \mathbf{A}_{\mathrm{t}_{\mathrm{u}}}^{}\right)^{\mathrm{T}} \odot \mathbf{R}_{0}}\right\}, \label{equ:FIM_exact_submat_theta_r_l_phi_t_u} 
  \end{split}\\
  \begin{split}
\medmath{\mathbf{J}_{\bm{\theta}_{\mathrm{r}_{\mathrm{l}}} \bm{\tau}}}
 &=   \medmath{\frac{2}{\sigma^2}\sum_{q = 1}^{N_{Q}}\Re}\left\{\medmath{\left(\mathbf{B}^{\mathrm{H}} \mathbf{A}_{\mathrm{r}_{\mathrm{u}}}^{\mathrm{H}} \mathbf{A}_{\mathrm{r}_{\mathrm{u}}} \mathbf{B}\right) \odot\left(\mathbf{k}_{q,\mathrm{r}_{\mathrm{l}}} \mathbf{k}_{q,\mathrm{l}}^{\mathrm{H}}\right)
 \odot\left(\mathbf{D}_{\mathrm{\gamma}}^{\mathrm{H}} \mathbf{D}_{\mathrm{\gamma}}\right)} \right. \\ & \left. \medmath{
  \odot\left(\mathbf{A}_{\mathrm{t}_{\mathrm{u}}}^{\mathrm{H}}  \mathbf{F}   \mathbf{F}^{\mathrm{H}} \mathbf{A}_{\mathrm{t}_{\mathrm{u}}}^{}\right)^{\mathrm{T}} \odot \mathbf{R}_{1}}\right\}, \label{equ:FIM_exact_submat_theta_r_l_tau}
  \end{split}\\
  \begin{split}
\medmath{\mathbf{J}_{\bm{\theta}_{\mathrm{r}_{\mathrm{l}}} \bm{\beta}_{\mathrm{I}}}} &\medmath{+j\mathbf{J}_{\bm{\theta}_{\mathrm{r}_{\mathrm{l}}} \bm{\beta}_{\mathrm{R}}}}
 =   \medmath{-\frac{2}{\sigma^2}\sum_{q = 1}^{N_{Q}}}\left\{\medmath{\left(\mathbf{B}^{\mathrm{H}} \mathbf{A}_{\mathrm{r}_{\mathrm{u}}}^{\mathrm{H}} \mathbf{A}_{\mathrm{r}_{\mathrm{u}}} \right) \odot\left(\mathbf{k}_{q,\mathrm{r}_{\mathrm{l}}} \mathbf{k}_{q,\mathrm{l}}^{\mathrm{H}}\right)
 }   \right. \\ & \left.  
 \medmath{\odot\left(\mathbf{D}_{\mathrm{\gamma}}^{\mathrm{H}} \mathbf{D}_{\mathrm{\gamma}}\right)\odot\left(\mathbf{A}_{\mathrm{t}_{\mathrm{u}}}^{\mathrm{H}}  \mathbf{F}   \mathbf{F}^{\mathrm{H}} \mathbf{A}_{\mathrm{t}_{\mathrm{u}}}^{}\right)^{\mathrm{T}} \odot \mathbf{R}_{0}}\right\} \label{equ:FIM_exact_submat_theta_r_l_beta}.
 \end{split}
\end{align}
\end{subequations}
The submatrices related to $\bm{\phi}_{\mathrm{r}_{\mathrm{l}}}$ can be obtained by replacing $\mathbf{k}_{q,\mathrm{r}_{\mathrm{l}}}$ with $\mathbf{p}_{q,\mathrm{r}_{\mathrm{l}}}$ in the appropriate equations. 
The submatrices related to $\bm{\theta}_{\mathrm{t}_{\mathrm{u}}}$ are obtained as
\begin{subequations}[equ:FIM_exact_submat_theta_t_u]
\begin{align}
\begin{split}
\medmath{\mathbf{J}_{\bm{\theta}_{\mathrm{t}_{\mathrm{u}}} \bm{\theta}_{\mathrm{t}_{\mathrm{u}}}}}
 &=   \medmath{\frac{2}{\sigma^2}\sum_{q = 1}^{N_{Q}}\Re} \left\{\medmath{\left(\mathbf{B}^{\mathrm{H}} \mathbf{A}_{\mathrm{r}_{\mathrm{u}}}^{\mathrm{H}} \mathbf{A}_{\mathrm{r}_{\mathrm{u}}} \mathbf{B}\right) \odot\left(\mathbf{k}_{q,\mathrm{l}} \mathbf{k}_{q,\mathrm{l}}^{\mathrm{H}}\right)
 \odot\left(\mathbf{D}_{\mathrm{\gamma}}^{\mathrm{H}} \mathbf{D}_{\mathrm{\gamma}}\right)} \right. \\ & \left. \medmath{
  \odot\left(\mathbf{K}_{\mathrm{t}_{\mathrm{u}}}^{\mathrm{H}}  \mathbf{F}   \mathbf{F}^{\mathrm{H}} \mathbf{K}_{\mathrm{t}_{\mathrm{u}}}^{}\right)^{\mathrm{T}} \odot \mathbf{R}_{0}}\right\}, \label{equ:FIM_exact_submat_theta_t_u_theta_t_u} \end{split} \\
  \begin{split}
\medmath{\mathbf{J}_{\bm{\theta}_{\mathrm{t}_{\mathrm{u}}} \bm{\phi}_{\mathrm{t}_{\mathrm{u}}}}}
 &=  \medmath{ \frac{2}{\sigma^2}\sum_{q = 1}^{N_{Q}}\Re}\left\{\medmath{\left(\mathbf{B}^{\mathrm{H}} \mathbf{A}_{\mathrm{r}_{\mathrm{u}}}^{\mathrm{H}} \mathbf{A}_{\mathrm{r}_{\mathrm{u}}} \mathbf{B}\right) \odot\left(\mathbf{k}_{q,\mathrm{l}} \mathbf{k}_{q,\mathrm{l}}^{\mathrm{H}}\right)
 \odot\left(\mathbf{D}_{\mathrm{\gamma}}^{\mathrm{H}} \mathbf{D}_{\mathrm{\gamma}}\right)} \right. \\ & \left. \medmath{
  \odot\left(\mathbf{P}_{\mathrm{t}_{\mathrm{u}}}^{\mathrm{H}}  \mathbf{F}   \mathbf{F}^{\mathrm{H}} \mathbf{K}_{\mathrm{t}_{\mathrm{u}}}^{}\right)^{\mathrm{T}} \odot \mathbf{R}_{0}}\right\}, \label{equ:FIM_exact_submat_theta_t_u_phi_t_u}
  \end{split}\\
  \begin{split}
\medmath{\mathbf{J}_{\bm{\theta}_{\mathrm{t}_{\mathrm{u}}} \bm{\tau}} }
 &=  \medmath{ -\frac{2}{\sigma^2}\sum_{q = 1}^{N_{Q}}\Re} \left\{\medmath{\left(\mathbf{B}^{\mathrm{H}} \mathbf{A}_{\mathrm{r}_{\mathrm{u}}}^{\mathrm{H}} \mathbf{A}_{\mathrm{r}_{\mathrm{u}}} \mathbf{B}\right) \odot\left(\mathbf{k}_{q,\mathrm{l}} \mathbf{k}_{q,\mathrm{l}}^{\mathrm{H}}\right)
 \odot\left(\mathbf{D}_{\mathrm{\gamma}}^{\mathrm{H}} \mathbf{D}_{\mathrm{\gamma}}\right)} \right. \\ & \left. \medmath{
  \odot\left(\mathbf{A}_{\mathrm{t}_{\mathrm{u}}}^{\mathrm{H}}  \mathbf{F}   \mathbf{F}^{\mathrm{H}} \mathbf{K}_{\mathrm{t}_{\mathrm{u}}}^{}\right)^{\mathrm{T}} \odot \mathbf{R}_{1}}\right\}, \label{equ:FIM_exact_submat_theta_t_u_tau} 
  \end{split}\\
  \begin{split}
\medmath{\mathbf{J}_{\bm{\theta}_{\mathrm{t}_{\mathrm{u}}} \bm{\beta}_{\mathrm{I}}}} &\medmath{+j\mathbf{J}_{\bm{\theta}_{\mathrm{t}_{\mathrm{u}}} \bm{\beta}_{\mathrm{R}}}} =   \medmath{\frac{2}{\sigma^2}\sum_{q = 1}^{N_{Q}}}\left\{\medmath{\left(\mathbf{B}^{\mathrm{H}} \mathbf{A}_{\mathrm{r}_{\mathrm{u}}}^{\mathrm{H}} \mathbf{A}_{\mathrm{r}_{\mathrm{u}}} \right) \odot\left(\mathbf{k}_{q,\mathrm{l}} \mathbf{k}_{q,\mathrm{l}}^{\mathrm{H}}\right)
 \odot\left(\mathbf{D}_{\mathrm{\gamma}}^{\mathrm{H}} \mathbf{D}_{\mathrm{\gamma}}\right)} \right. \\ & \left. \medmath{
  \odot\left(\mathbf{A}_{\mathrm{t}_{\mathrm{u}}}^{\mathrm{H}}  \mathbf{F}   \mathbf{F}^{\mathrm{H}} \mathbf{K}_{\mathrm{t}_{\mathrm{u}}}^{}\right)^{\mathrm{T}} \odot \mathbf{R}_{0}}\right\}. \label{equ:FIM_exact_submat_theta_t_u_beta}
  \end{split}
\end{align}
\end{subequations}
The submatrices related to $\bm{\phi}_{\mathrm{t}_{\mathrm{u}}}$ can be obtained by replacing $\mathbf{K}_{\mathrm{t}_{\mathrm{u}}}^{}$ with $\mathbf{P}_{\mathrm{t}_{\mathrm{u}}}^{}$ in the appropriate equations. 
The submatrices related to the delay are obtained as
\begin{subequations}[equ:FIM_exact_submat_tau]
\begin{align}
\begin{split}
    \medmath{\mathbf{J}_{\bm{\tau} \bm{\tau}}}
 &=   \medmath{\frac{2}{\sigma^2}\sum_{q = 1}^{N_{Q}}\Re}\left\{\medmath{\left(\mathbf{B}^{\mathrm{H}} \mathbf{A}_{\mathrm{r}_{\mathrm{u}}}^{\mathrm{H}} \mathbf{A}_{\mathrm{r}_{\mathrm{u}}} \mathbf{B}\right) \odot\left(\mathbf{k}_{q,\mathrm{l}} \mathbf{k}_{q,\mathrm{l}}^{\mathrm{H}}\right)
 \odot\left(\mathbf{D}_{\mathrm{\gamma}}^{\mathrm{H}} \mathbf{D}_{\mathrm{\gamma}}\right) } \right. \\ & \left. \medmath{
 \odot\left(\mathbf{A}_{\mathrm{t}_{\mathrm{u}}}^{\mathrm{H}}  \mathbf{F}   \mathbf{F}^{\mathrm{H}} \mathbf{A}_{\mathrm{t}_{\mathrm{u}}}^{}\right)^{\mathrm{T}} \odot \mathbf{R}_{2}}\right\}, \label{equ:FIM_exact_submat_tau_tau} 
 \end{split}
\\
\begin{split}
\medmath{\mathbf{J}_{\bm{\tau} \bm{\beta}_{\mathrm{I}}}} &\medmath{+j\mathbf{J}_{\bm{\tau} \bm{\beta}_{\mathrm{R}}}} =   \medmath{-\frac{2}{\sigma^2}\sum_{q = 1}^{N_{Q}}}\left\{ \medmath{\left(\mathbf{B}^{\mathrm{H}} \mathbf{A}_{\mathrm{r}_{\mathrm{u}}}^{\mathrm{H}} \mathbf{A}_{\mathrm{r}_{\mathrm{u}}} \right) \odot\left(\mathbf{k}_{q,\mathrm{l}} \mathbf{k}_{q,\mathrm{l}}^{\mathrm{H}}\right)
 \odot\left(\mathbf{D}_{\mathrm{\gamma}}^{\mathrm{H}} \mathbf{D}_{\mathrm{\gamma}}\right)} \right. \\ & \left. \medmath{
  \odot\left(\mathbf{A}_{\mathrm{t}_{\mathrm{u}}}^{\mathrm{H}}  \mathbf{F}   \mathbf{F}^{\mathrm{H}} \mathbf{A}_{\mathrm{t}_{\mathrm{u}}}^{}\right)^{\mathrm{T}} \odot \mathbf{R}_{1}}\right\}. \label{equ:FIM_exact_submat_tau_beta}
  \end{split}
\end{align}
\end{subequations}
The submatrices related to the channel complex gain  are obtained as
\begin{subequations}[equ:FIM_exact_submat_beta_r]
\begin{align}
\begin{split}
\medmath{\mathbf{J}_{ \bm{\beta}_{\mathrm{R}}  \bm{\beta}_{\mathrm{R}}}}
 &= \medmath{\mathbf{J}_{ \bm{\beta}_{\mathrm{I}}  \bm{\beta}_{\mathrm{I}}}
 =   \frac{2}{\sigma^2}\sum_{q = 1}^{N_{Q}}\Re}\left\{\medmath{\left( \mathbf{A}_{\mathrm{r}_{\mathrm{u}}}^{\mathrm{H}} \mathbf{A}_{\mathrm{r}_{\mathrm{u}}}\right) \odot\left(\mathbf{k}_{q,\mathrm{l}} \mathbf{k}_{q,\mathrm{l}}^{\mathrm{H}}\right)
 \odot\left(\mathbf{D}_{\mathrm{\gamma}}^{\mathrm{H}} \mathbf{D}_{\mathrm{\gamma}}\right)}
   \right. \\ & \left. \medmath{\odot\left(\mathbf{A}_{\mathrm{t}_{\mathrm{u}}}^{\mathrm{H}}  \mathbf{F}   \mathbf{F}^{\mathrm{H}} \mathbf{A}_{\mathrm{t}_{\mathrm{u}}}^{}\right)^{\mathrm{T}} \odot \mathbf{R}_{0}}\right\}, \label{equ:FIM_exact_submat_beta_r_beta_r} \end{split}
 \\
    \begin{split}
\medmath{\mathbf{J}_{ \bm{\beta}_{\mathrm{R}}  \bm{\beta}_{\mathrm{I}}}}
 &=   \medmath{\frac{2}{\sigma^2}\sum_{q = 1}^{N_{Q}}\Re} \left\{\medmath{\left( j \mathbf{A}_{\mathrm{r}_{\mathrm{u}}}^{\mathrm{H}} \mathbf{A}_{\mathrm{r}_{\mathrm{u}}}\right) \odot\left(\mathbf{k}_{q,\mathrm{l}} \mathbf{k}_{q,\mathrm{l}}^{\mathrm{H}}\right)
 \odot\left(\mathbf{D}_{\mathrm{\gamma}}^{\mathrm{H}} \mathbf{D}_{\mathrm{\gamma}}\right) }  \right. \\ & \left.   \medmath{  
  \odot\left(\mathbf{A}_{\mathrm{t}_{\mathrm{u}}}^{\mathrm{H}}  \mathbf{F}   \mathbf{F}^{\mathrm{H}} \mathbf{A}_{\mathrm{t}_{\mathrm{u}}}^{}\right)^{\mathrm{T}} \odot \mathbf{R}_{0} }\right\}. \label{equ:FIM_exact_submat_beta_r_beta_I}
  \end{split}
\end{align}
\end{subequations}

	\subsection{ Entries of the  EFIM with Unitary RIS Sequences and Parallel RIS Configurations}
	\label{appendix:FIM_nuis_remark_structure}The information loss terms in the EFIM expression that have a different structure than that described in Corollary \ref{corollary:parallel_nu}  are presented in this section. The information loss terms  related to the receive angles in both elevation  and azimuth are written as
$$\medmath{\mathbf{J}_{ \bm{\theta}_{\mathrm{r}_{\mathrm{u}}}  \bm{\theta}_{\mathrm{r}_{\mathrm{u}}}}^{\mathrm{nu}}
 = \Tilde{\mathbf{J}}_{ \bm{\beta}_{\mathrm{R}}  \bm{\beta}_{\mathrm{R}}}^{-1}  \frac{4}{\sigma^4}N_{Q} \left[  \Re \left\{ \|\left(\mathbf{B}^{\mathrm{H}} \mathbf{K}_{\mathrm{r}_{\mathrm{u}}}^{\mathrm{H}} \mathbf{A}_{\mathrm{r}_{\mathrm{u}}} \right)\|_0^2    \odot  \mathbf{V}_{q,\mathrm{r}_{\mathrm{u}}} \odot  \mathbf{V}_{q,\mathrm{r}_{\mathrm{u}}} \right\} \right],} \\
$$ 
$$  \medmath{\mathbf{J}_{ \bm{\phi}_{\mathrm{r}_{\mathrm{u}}}  \bm{\phi}_{\mathrm{r}_{\mathrm{u}}}}^{\mathrm{nu}}
 =  \Tilde{\mathbf{J}}_{ \bm{\beta}_{\mathrm{R}}  \bm{\beta}_{\mathrm{R}}}^{-1}  \frac{4}{\sigma^4}N_{Q}\left[  \Re \left\{ \|\left(\mathbf{B}^{\mathrm{H}} \mathbf{P}_{\mathrm{r}_{\mathrm{u}}}^{\mathrm{H}} \mathbf{A}_{\mathrm{r}_{\mathrm{u}}} \right)   \|_0^2    \odot  \mathbf{V}_{q,\mathrm{r}_{\mathrm{u}}} \odot \mathbf{V}_{q,\mathrm{r}_{\mathrm{u}}} \right\} \right],}$$
and
\begin{equation}
\begin{aligned}
 \medmath{\mathbf{J}_{ \bm{\theta}_{\mathrm{r}_{\mathrm{u}}}  \bm{\phi}_{\mathrm{r}_{\mathrm{u}}}}^{\mathrm{nu}}}
 &=   \medmath{\Tilde{\mathbf{J}}_{ \bm{\beta}_{\mathrm{R}}  \bm{\beta}_{\mathrm{R}}}^{-1}  \frac{4}{\sigma^4}N_{Q}  \Re} \left\{ \medmath{\left(\mathbf{B}^{\mathrm{H}} \mathbf{K}_{\mathrm{r}_{\mathrm{u}}}^{\mathrm{H}} \mathbf{A}_{\mathrm{r}_{\mathrm{u}}}  \right)  \odot \left(\mathbf{B}^{\mathrm{H}} \mathbf{P}_{\mathrm{r}_{\mathrm{u}}}^{\mathrm{H}} \mathbf{A}_{\mathrm{r}_{\mathrm{u}}} \right)^{\mathrm{H}}} \right. \\ & \left.  \medmath{\odot  \mathbf{V}_{q,\mathrm{r}_{\mathrm{u}}} \odot \mathbf{V}_{q,\mathrm{r}_{\mathrm{u}}}} \right\},
 \end{aligned}
\end{equation}
where $
 \medmath{\mathbf{V}_{q,\mathrm{r}_{\mathrm{u}}} =   \left(\mathbf{k}_{q,\mathrm{l}} \mathbf{k}_{q,\mathrm{l}}^{\mathrm{H}}\right)
 \odot\left(\mathbf{D}_{\mathrm{\gamma}}^{\mathrm{H}} \mathbf{D}_{\mathrm{\gamma}}\right)
 \odot\left(\mathbf{A}_{\mathrm{t}_{\mathrm{u}}}^{\mathrm{H}}  \mathbf{F}   \mathbf{F}^{\mathrm{H}} \mathbf{A}_{\mathrm{t}_{\mathrm{u}}}^{}\right)^{\mathrm{T}} \odot \mathbf{R}_{0}.}
$
The terms related to the transmit angles are written as
\begin{equation}
\label{equ:appendix_FIM_nuis_equivalent_submat_theta_t_u}
\begin{aligned}
&\medmath{\mathbf{J}_{ \bm{\theta}_{\mathrm{\mathrm{t}_{\mathrm{u}}}}  \bm{\theta}_{\mathrm{\mathrm{t}_{\mathrm{u}}}}}^{\mathrm{nu}}} \\ &= \medmath{   \Tilde{\mathbf{J}}_{ \bm{\beta}_{\mathrm{R}}  \bm{\beta}_{\mathrm{R}}}^{-1} \frac{4}{\sigma^4}N_{Q} \left[  \Re \left\{ \| \left(\mathbf{B}^{\mathrm{H}} \mathbf{A}_{\mathrm{r}_{\mathrm{u}}}^{\mathrm{H}} \mathbf{A}_{\mathrm{r}_{\mathrm{u}}} \right) \|_0^2     \odot  \mathbf{V}_{q,\mathrm{k}_{\mathrm{t}_\mathrm{u}}}^{\mathrm{H}} \odot  \mathbf{V}_{q,\mathrm{k}_{\mathrm{t}_\mathrm{u}}} \right\} \right]}, \\
 &\medmath{ \mathbf{J}_{ \bm{\theta}_{\mathrm{\mathrm{t}_{\mathrm{u}}}}  \bm{\phi}_{\mathrm{\mathrm{t}_{\mathrm{u}}}}}^{\mathrm{nu}}}
 \\ &=  \medmath{\Tilde{\mathbf{J}}_{ \bm{\beta}_{\mathrm{R}}  \bm{\beta}_{\mathrm{R}}}^{-1} \frac{4}{\sigma^4}N_{Q}  \left[ \Re \left\{ \| \left(\mathbf{B}^{\mathrm{H}} \mathbf{A}_{\mathrm{r}_{\mathrm{u}}}^{\mathrm{H}} \mathbf{A}_{\mathrm{r}_{\mathrm{u}}} \right) \|_0^2      \odot  \mathbf{V}_{q,\mathrm{k}_{\mathrm{t}_\mathrm{u}}} \odot  \mathbf{V}_{q,\mathrm{p}_{\mathrm{t}_\mathrm{u}}}^{\mathrm{H}} \right\} \right],}  \\
  &\medmath{\mathbf{J}_{ \bm{\phi}_{\mathrm{\mathrm{t}_{\mathrm{u}}}}  \bm{\phi}_{\mathrm{\mathrm{t}_{\mathrm{u}}}}}^{\mathrm{nu}}} \\ &=  \medmath{\Tilde{\mathbf{J}}_{ \bm{\beta}_{\mathrm{R}}  \bm{\beta}_{\mathrm{R}}}^{-1}  \frac{4}{\sigma^4}N_{Q} \left[ \Re \left\{ \| \left(\mathbf{B}^{\mathrm{H}} \mathbf{A}_{\mathrm{r}_{\mathrm{u}}}^{\mathrm{H}} \mathbf{A}_{\mathrm{r}_{\mathrm{u}}} \right) \|_0^2  \odot  \mathbf{V}_{q,\mathrm{p}_{\mathrm{t}_\mathrm{u}}}^{\mathrm{H}} \odot  \mathbf{V}_{q,\mathrm{p}_{\mathrm{t}_\mathrm{u}}} \right\} \right],}
\end{aligned}
\end{equation}
where
$$
\begin{aligned}
 \medmath{\mathbf{V}_{q,\mathrm{k}_{\mathrm{t}_\mathrm{u}}} =   \left(  \left(\mathbf{A}_{\mathrm{t}_{\mathrm{u}}}^{\mathrm{H}}  \mathbf{F}   \mathbf{F}^{\mathrm{H}} \mathbf{K}_{\mathrm{t}_{\mathrm{u}}}^{}\right)^{\mathrm{T}}  \right) \odot \left(\mathbf{k}_{q,\mathrm{l}} \mathbf{k}_{q,\mathrm{l}}^{\mathrm{H}}\right)
 \odot\left(\mathbf{D}_{\mathrm{\gamma}}^{\mathrm{H}} \mathbf{D}_{\mathrm{\gamma}}\right)
  \odot \mathbf{R}_{0}}, \\
   \medmath{\mathbf{V}_{q,\mathrm{p}_{\mathrm{t}_\mathrm{u}}} =   \left(  \left(\mathbf{A}_{\mathrm{t}_{\mathrm{u}}}^{\mathrm{H}}  \mathbf{F}   \mathbf{F}^{\mathrm{H}} \mathbf{P}_{\mathrm{t}_{\mathrm{u}}}^{}\right)^{\mathrm{T}}  \right) \odot \left(\mathbf{k}_{q,\mathrm{l}} \mathbf{k}_{q,\mathrm{l}}^{\mathrm{H}}\right)
 \odot\left(\mathbf{D}_{\mathrm{\gamma}}^{\mathrm{H}} \mathbf{D}_{\mathrm{\gamma}}\right)
  \odot \mathbf{R}_{0}.}
  \end{aligned}
$$
The term related to the ToA is written as
\begin{equation}
\label{equ:appendix_FIM_nuis_equivalent_submat_toa}
\begin{aligned}
\medmath{\mathbf{J}_{ \mathbf{\mathrm{\tau}} \mathbf{\mathrm{\tau}}}^{\mathrm{nu}}
 =  \frac{4}{\sigma^4}N_{Q} \Tilde{\mathbf{J}}_{ \bm{\beta}_{\mathrm{R}}  \bm{\beta}_{\mathrm{R}}}^{-1} \odot \left[ \mathbf{J}_{q,\bm{\tau} \bm{\beta}_{\mathrm{I}}} +j\mathbf{J}_{q,\bm{\tau} \bm{\beta}_{\mathrm{R}}}\right].}    \\ 
\end{aligned}
\end{equation}

\subsection{Proof of Corollary \ref{corollary:parallel_nu}}
\label{appendix_nusiance_structure}
In this proof, we focus on the information loss due to the nuisance parameters which concerns the cross correlation between the receive elevation angle and the elevation angle of reflection
\begin{equation}
\label{equ:FIM_nuis_equivalent_submat_theta_r_u_theta_t_l}
\begin{aligned}
\medmath{\mathbf{J}_{ \bm{\theta}_{\mathrm{r}_{\mathrm{u}}}  \bm{\theta}_{\mathrm{t}_{\mathrm{l}}}}^{\mathrm{nu}}}
 &= \medmath{\mathbf{J}_{ \bm{\theta}_{\mathrm{r}_{\mathrm{u}}}  \bm{\beta}_{\mathrm{R}}} \Tilde{\mathbf{J}}_{ \bm{\beta}_{\mathrm{R}}  \bm{\beta}_{\mathrm{R}}}^{-1}  \mathbf{J}_{ \bm{\theta}_{\mathrm{t}_{\mathrm{l}}}  \bm{\beta}_{\mathrm{R}}}^{\mathrm{T}} + \mathbf{J}_{ \bm{\theta}_{\mathrm{r}_{\mathrm{u}}}  \bm{\beta}_{\mathrm{I}}}  \Tilde{\mathbf{J}}_{ \bm{\beta}_{\mathrm{I}}  \bm{\beta}_{\mathrm{I}}}^{-1}  \mathbf{J}_{ \bm{\theta}_{\mathrm{t}_{\mathrm{l}}}  \bm{\beta}_{\mathrm{I}}}^{\mathrm{T}}} \\ &= \medmath{\Tilde{\mathbf{J}}_{ \bm{\beta}_{\mathrm{R}}  \bm{\beta}_{\mathrm{R}}}^{-1}\left[ \mathbf{J}_{ \bm{\theta}_{\mathrm{r}_{\mathrm{u}}}  \bm{\beta}_{\mathrm{R}}}   \mathbf{J}_{ \bm{\theta}_{\mathrm{t}_{\mathrm{l}}}  \bm{\beta}_{\mathrm{R}}}^{\mathrm{T}} + \mathbf{J}_{ \bm{\theta}_{\mathrm{r}_{\mathrm{u}}}  \bm{\beta}_{\mathrm{I}}}   \mathbf{J}_{ \bm{\theta}_{\mathrm{t}_{\mathrm{l}}}  \bm{\beta}_{\mathrm{I}}}^{\mathrm{T}} \right],} \\
\end{aligned}
\end{equation}
which is a consequence  of Remark \ref{remark:nusiance_params_diagonal}. Now, applying basic complex analysis, $\Im{(\nu_1)}\Im{(\nu_2)} + \Re{(\nu_1)}\Re{(\nu_2)} = \Re{(\nu_1 \nu_2^{\mathrm{H}})}= \Re{(\nu_1^{\mathrm{H}} \nu_2^{})}$, we have
\begin{equation}
\label{equ:FIM_nuis_equivalent_submat_theta_r_u_theta_t_l_1}
\begin{aligned}
&\medmath{\mathbf{J}_{ \bm{\theta}_{\mathrm{r}_{\mathrm{u}}}  \bm{\theta}_{\mathrm{t}_{\mathrm{l}}}}^{\mathrm{nu}} = } \\ &
  \medmath{\Tilde{\mathbf{J}}_{ \bm{\beta}_{\mathrm{R}}  \bm{\beta}_{\mathrm{R}}}^{-1} \frac{4}{\sigma^4}N_{Q}\Re\left\{\left[ ( \mathbf{J}_{q, \bm{\theta}_{\mathrm{r}_{\mathrm{u}}}  \bm{\beta}_{\mathrm{I}}}+ j\mathbf{J}_{q, \bm{\theta}_{\mathrm{r}_{\mathrm{u}}}  \bm{\beta}_{\mathrm{R}}}) (     \mathbf{J}_{q, \bm{\theta}_{\mathrm{t}_{\mathrm{l}}}  \bm{\beta}_{\mathrm{I}}}^{\mathrm{}} +j\mathbf{J}_{q, \bm{\theta}_{\mathrm{t}_{\mathrm{l}}}  \bm{\beta}_{\mathrm{R}}}^{\mathrm{}}  )^{\mathrm{H}} \right]\right\}.} \\
\end{aligned}
\end{equation}
Now, substituting (\ref{equ:FIM_exact_submat_theta_r_u_beta}) and (\ref{equ:FIM_exact_submat_theta_t_l_beta}) in the above equation gives
\begin{equation}
\label{equ:FIM_nuis_equivalent_submat_theta_r_u_theta_t_l_3}
\begin{aligned}
&\medmath{\mathbf{J}_{ \bm{\theta}_{\mathrm{r}_{\mathrm{u}}}  \bm{\theta}_{\mathrm{t}_{\mathrm{l}}}}^{\mathrm{nu}} =  
\Tilde{\mathbf{J}}_{ \bm{\beta}_{\mathrm{R}}  \bm{\beta}_{\mathrm{R}}}^{-1}\left[ -\frac{4N_{Q}}{\sigma^4}\Re\left\{\left(\mathbf{B}^{\mathrm{H}} \mathbf{K}_{\mathrm{r}_{\mathrm{u}}}^{\mathrm{H}} \mathbf{A}_{\mathrm{r}_{\mathrm{u}}} \right) \odot\left(\mathbf{k}_{q,\mathrm{l}} \mathbf{k}_{q,\mathrm{l}}^{\mathrm{H}}\right)  \right. \right.}\\
 &\medmath{\left.\left. \odot
\left(\mathbf{D}_{\mathrm{\gamma}}^{\mathrm{H}} \mathbf{D}_{\mathrm{\gamma}}\right)  \odot\left(\mathbf{A}_{\mathrm{t}_{\mathrm{u}}}^{\mathrm{H}}  \mathbf{F}   \mathbf{F}^{\mathrm{H}} \mathbf{A}_{\mathrm{t}_{\mathrm{u}}}^{}\right)^{\mathrm{T}} \odot \mathbf{R}_{0} \odot \left( \mathbf{A}_{\mathrm{r}_{\mathrm{u}}}^{\mathrm{H}} \mathbf{A}_{\mathrm{r}_{\mathrm{u}}}\mathbf{B}^{\mathrm{}} \right)  
  \odot\left(\mathbf{k}_{q,\mathrm{l}}^{\mathrm{}}\mathbf{k}_{q,\mathrm{t}_{\mathrm{l}}}^{\mathrm{H}} \right)  \right. \right.}\\
 &\medmath{\left.\left. 
 \odot\left(\mathbf{D}_{\mathrm{\gamma}}^{\mathrm{H}} \mathbf{D}_{\mathrm{\gamma}}\right)^{\mathrm{}} 
\left(\mathbf{A}_{\mathrm{t}_{\mathrm{u}}}^{\mathrm{H}}  \mathbf{F}   \mathbf{F}^{\mathrm{H}} \mathbf{A}_{\mathrm{t}_{\mathrm{u}}}^{}\right)^{\mathrm{T}} \odot \mathbf{R}_{0}
 \right\} \right] }.
\end{aligned}
\end{equation}
Rearranging the terms and applying the properties relating Hadamard products with diagonal matrices \cite{horn2012matrix} produces
\begin{equation}
\label{equ:FIM_nuis_equivalent_submat_theta_r_u_theta_t_l_4}
\begin{aligned}
&\medmath{\mathbf{J}_{ \bm{\theta}_{\mathrm{r}_{\mathrm{u}}}  \bm{\theta}_{\mathrm{t}_{\mathrm{l}}}}^{\mathrm{nu}} =  
\Tilde{\mathbf{J}}_{ \bm{\beta}_{\mathrm{R}}  \bm{\beta}_{\mathrm{R}}}^{-1}\left[ -\frac{4N_{Q}}{\sigma^4}\Re\left\{\left(\mathbf{B}^{\mathrm{H}} \mathbf{K}_{\mathrm{r}_{\mathrm{u}}}^{\mathrm{H}} \mathbf{A}_{\mathrm{r}_{\mathrm{u}}}\mathbf{B}^{\mathrm{}} \right) \odot\left(\mathbf{k}_{q,\mathrm{l}}^{\mathrm{}}\mathbf{k}_{q,\mathrm{t}_{\mathrm{l}}}^{\mathrm{H}} \right)  \right. \right.}\\
 &\medmath{\left.\left. \odot
\left(\mathbf{D}_{\mathrm{\gamma}}^{\mathrm{H}} \mathbf{D}_{\mathrm{\gamma}}\right)  \odot\left(\mathbf{A}_{\mathrm{t}_{\mathrm{u}}}^{\mathrm{H}}  \mathbf{F}   \mathbf{F}^{\mathrm{H}} \mathbf{A}_{\mathrm{t}_{\mathrm{u}}}^{}\right)^{\mathrm{T}} \odot \mathbf{R}_{0}  \odot \left( \mathbf{A}_{\mathrm{r}_{\mathrm{u}}}^{\mathrm{H}} \mathbf{A}_{\mathrm{r}_{\mathrm{u}}} \right)  
  \odot \left(\mathbf{k}_{q,\mathrm{l}} \mathbf{k}_{q,\mathrm{l}}^{\mathrm{H}}\right)  \right. \right. }\\
 &\medmath{ \left.\left. 
 \odot\left(\mathbf{D}_{\mathrm{\gamma}}^{\mathrm{H}} \mathbf{D}_{\mathrm{\gamma}}\right)^{\mathrm{}} \odot
\left(\mathbf{A}_{\mathrm{t}_{\mathrm{u}}}^{\mathrm{H}}  \mathbf{F}   \mathbf{F}^{\mathrm{H}} \mathbf{A}_{\mathrm{t}_{\mathrm{u}}}^{}\right)^{\mathrm{T}} \odot \mathbf{R}_{0}
 \right\} \right]}.
\end{aligned}
\end{equation}
Using Remark \ref{remark:nusiance_params_diagonal} yields
\begin{equation}
\label{equ:FIM_nuis_equivalent_submat_theta_r_u_theta_t_l_5}
\begin{aligned}
&\medmath{\mathbf{J}_{ \bm{\theta}_{\mathrm{r}_{\mathrm{u}}}  \bm{\theta}_{\mathrm{t}_{\mathrm{l}}}}^{\mathrm{nu}} =  
\Tilde{\mathbf{J}}_{ \bm{\beta}_{\mathrm{R}}  \bm{\beta}_{\mathrm{R}}}^{-1} {\mathbf{J}}_{q, \bm{\beta}_{\mathrm{R}}  \bm{\beta}_{\mathrm{R}}}^{}\left[ -\frac{4N_{Q}}{\sigma^4}\Re\left\{\left(\mathbf{B}^{\mathrm{H}} \mathbf{K}_{\mathrm{r}_{\mathrm{u}}}^{\mathrm{H}} \mathbf{A}_{\mathrm{r}_{\mathrm{u}}}\mathbf{B}^{\mathrm{}} \right) \odot\left(\mathbf{k}_{q,\mathrm{l}}^{\mathrm{}}\mathbf{k}_{q,\mathrm{t}_{\mathrm{l}}}^{\mathrm{H}} \right)  \right. \right.}\\
 &\medmath{\left.\left. \odot
\left(\mathbf{D}_{\mathrm{\gamma}}^{\mathrm{H}} \mathbf{D}_{\mathrm{\gamma}}\right)  \odot\left(\mathbf{A}_{\mathrm{t}_{\mathrm{u}}}^{\mathrm{H}}  \mathbf{F}   \mathbf{F}^{\mathrm{H}} \mathbf{A}_{\mathrm{t}_{\mathrm{u}}}^{}\right)^{\mathrm{T}} \odot \mathbf{R}_{0}
 \right\} \right]}.
\end{aligned}
\end{equation}
Now, recognizing  that the term in the square brackets is $\mathbf{J}_{q, \bm{\theta}_{\mathrm{r}_{\mathrm{u}}}  \bm{\theta}_{\mathrm{t}_{\mathrm{l}}}}^{\mathrm{}} $, we have the structure
$
\mathbf{J}_{ \bm{\theta}_{\mathrm{r}_{\mathrm{u}}}  \bm{\theta}_{\mathrm{t}_{\mathrm{l}}}}^{\mathrm{nu}}
  = \frac{4N_{Q}}{\sigma^4}\Tilde{\mathbf{J}}_{ \bm{\beta}_{\mathrm{R}}  \bm{\beta}_{\mathrm{R}}}^{-1} \mathbf{J}_{q, \bm{\beta}_{\mathrm{R}}  \bm{\beta}_{\mathrm{R}}}^{}\mathbf{J}_{q, \bm{\theta}_{\mathrm{r}_{\mathrm{u}}}  \bm{\theta}_{\mathrm{t}_{\mathrm{l}}}}^{}.$
Here, $\Tilde{\mathbf{J}}_{ \bm{\beta}_{\mathrm{R}}  \bm{\beta}_{\mathrm{R}}} = \frac{2}{\sigma^2}N_{Q} { \mathbf{J}}_{q, \bm{\beta}_{\mathrm{R}}  \bm{\beta}_{\mathrm{R}}} + { \mathbf{J}}_{ \bm{\beta}_{\mathrm{R}}  \bm{\beta}_{\mathrm{R}}}^{\mathrm{P}}.$ Other terms that have the same structure as those in Corollary \ref{corollary:parallel_nu} can be derived similarly.
\subsection{Entries of the FIM between the LOS path and the RIS paths}
\label{appendix:FIM_los_RIS}
Defining $\mathbf{R}_{0}^{\mathrm{L}_1}$ as the signal factor considering a LOS path and the RIS paths; the entries in the FIM $\mathbf{J}_{\psi \eta}^{\mathrm{D}} \in \mathbb{R}^{7 \times 11 {M}_1}$ are presented below
	\begin{subequations}[equ:FIM_exact_submat_theta_r_u_0]
\begin{align}
\begin{split}
\medmath{\mathbf{J}_{\bm{\theta}_{\mathrm{r}_{\mathrm{u}}}^{\mathrm{[0]}} \bm{\theta}_{\mathrm{r}_{\mathrm{u}}}}}
 &=   \medmath{\frac{2}{\sigma^2}\sum_{q = 1}^{N_{Q}}\Re}\left\{\left(\beta^{{\mathrm{H}}[\mathrm{0}]} \mathbf{a}_{\mathrm{r}_{\mathrm{u}}}^{{\mathrm{H}}[0] } {\mathbf{K}}_{\mathrm{r}_{\mathrm{u}}}^{{\mathrm{H}}[0]}   \mathbf{K}_{\mathrm{r}_{\mathrm{u}}} \mathbf{B}\right) \odot\left( \mathbf{k}_{q,\mathrm{l}}^{\mathrm{H}}\right) \right. \\ & \left. \medmath{
 \odot\left(\mathbf{1}_{{M}_1}^{\mathrm{H}}\mathbf{D}_{\mathrm{\gamma}}^{} \right)  \odot \left(\mathbf{A}_{\mathrm{t}_{\mathrm{u}}}^{\mathrm{H}}  \mathbf{F}   \mathbf{F}^{\mathrm{H}} \mathbf{a}_{\mathrm{t}_{\mathrm{u}}}^{[\mathrm{0}]}\right)^{\mathrm{T}} \odot \mathbf{R}_{0}^{\mathrm{L}_1}}\right\}, \label{equ:FIM_exact_submat_theta_r_u_0_theta_r_u} \end{split}\\
 \begin{split}
\medmath{\mathbf{J}_{\bm{\theta}_{\mathrm{r}_{\mathrm{u}}}^{\mathrm{[0]}} \bm{\phi}_{\mathrm{r}_{\mathrm{u}}}}}
 &=   \medmath{\frac{2}{\sigma^2}\sum_{q = 1}^{N_{Q}}\Re}\left\{\medmath{\left(\beta^{{\mathrm{H}}[\mathrm{0}]} \mathbf{a}_{\mathrm{r}_{\mathrm{u}}}^{{\mathrm{H}}[0] } {\mathbf{K}}_{\mathrm{r}_{\mathrm{u}}}^{{\mathrm{H}}[0]}   \mathbf{P}_{\mathrm{r}_{\mathrm{u}}} \mathbf{B}\right) \odot\left( \mathbf{k}_{q,\mathrm{l}}^{\mathrm{H}}\right)}  \right. \\ & \left.  \medmath{
 \odot\left(\mathbf{1}_{{M}_1}^{\mathrm{H}}\mathbf{D}_{\mathrm{\gamma}}^{} \right) \odot \left(\mathbf{A}_{\mathrm{t}_{\mathrm{u}}}^{\mathrm{H}}  \mathbf{F}   \mathbf{F}^{\mathrm{H}} \mathbf{a}_{\mathrm{t}_{\mathrm{u}}}^{[\mathrm{0}]}\right)^{\mathrm{T}} \odot \mathbf{R}_{0}^{\mathrm{L}_1}}\right\}, \label{equ:FIM_exact_submat_theta_r_u_0_phi_r_u} 
  \end{split}
 \\
  \begin{split}
\medmath{\mathbf{J}_{\bm{\theta}_{\mathrm{r}_{\mathrm{u}}}^{\mathrm{[0]}} \bm{\theta}_{\mathrm{t}_{\mathrm{l}}}}}
 &= \medmath{ - \frac{2}{\sigma^2}\sum_{q = 1}^{N_{Q}}\Re}\left\{\medmath{\left(\beta^{{\mathrm{H}}[\mathrm{0}]} \mathbf{a}_{\mathrm{r}_{\mathrm{u}}}^{{\mathrm{H}}[0] } {\mathbf{K}}_{\mathrm{r}_{\mathrm{u}}}^{{\mathrm{H}}[0]}   \mathbf{A}_{\mathrm{r}_{\mathrm{u}}} \mathbf{B}\right) \odot\left( \mathbf{k}_{q,\mathrm{t}_{\mathrm{l}}}^{\mathrm{H}}\right)} \right. \\ & \left. \medmath{
 \odot\left(\mathbf{1}_{{M}_1}^{\mathrm{H}}\mathbf{D}_{\mathrm{\gamma}}^{} \right) \odot \left(\mathbf{A}_{\mathrm{t}_{\mathrm{u}}}^{\mathrm{H}}  \mathbf{F}   \mathbf{F}^{\mathrm{H}} \mathbf{a}_{\mathrm{t}_{\mathrm{u}}}^{[\mathrm{0}]}\right)^{\mathrm{T}} \odot \mathbf{R}_{0}^{\mathrm{L}_1}}\right\}, \label{equ:FIM_exact_submat_theta_r_u_0_theta_t_l} 
  \end{split}
 \\
   \begin{split}
\medmath{\mathbf{J}_{\bm{\theta}_{\mathrm{r}_{\mathrm{u}}}^{\mathrm{[0]}} \bm{\phi}_{\mathrm{t}_{\mathrm{l}}}}}
 &=  \medmath{- \frac{2}{\sigma^2}\sum_{q = 1}^{N_{Q}}\Re}\left\{\medmath{\left(\beta^{{\mathrm{H}}[\mathrm{0}]} \mathbf{a}_{\mathrm{r}_{\mathrm{u}}}^{{\mathrm{H}}[0] } {\mathbf{K}}_{\mathrm{r}_{\mathrm{u}}}^{{\mathrm{H}}[0]}   \mathbf{A}_{\mathrm{r}_{\mathrm{u}}} \mathbf{B}\right) \odot\left( \mathbf{p}_{q,\mathrm{t}_{\mathrm{l}}}^{\mathrm{H}}\right)}
 \right. \\ & \left. \medmath{\odot\left(\mathbf{1}_{{M}_1}^{\mathrm{H}}\mathbf{D}_{\mathrm{\gamma}}^{} \right) \odot \left(\mathbf{A}_{\mathrm{t}_{\mathrm{u}}}^{\mathrm{H}}  \mathbf{F}   \mathbf{F}^{\mathrm{H}} \mathbf{a}_{\mathrm{t}_{\mathrm{u}}}^{[\mathrm{0}]}\right)^{\mathrm{T}} \odot \mathbf{R}_{0}^{\mathrm{L}_1}}\right\}, \label{equ:FIM_exact_submat_theta_r_u_0_phi_t_l}
  \end{split}
 \\
 \begin{split}
\medmath{\mathbf{J}_{\bm{\theta}_{\mathrm{r}_{\mathrm{u}}}^{\mathrm{[0]}} \bm{\theta}_{\mathrm{r}_{\mathrm{l}}}}}
 &=   \medmath{\frac{2}{\sigma^2}\sum_{q = 1}^{N_{Q}}\Re}\left\{\medmath{\left(\beta^{{\mathrm{H}}[\mathrm{0}]} \mathbf{a}_{\mathrm{r}_{\mathrm{u}}}^{{\mathrm{H}}[0] } {\mathbf{K}}_{\mathrm{r}_{\mathrm{u}}}^{{\mathrm{H}}[0]}   \mathbf{A}_{\mathrm{r}_{\mathrm{u}}} \mathbf{B}\right) \odot\left( \mathbf{k}_{q,\mathrm{r}_{\mathrm{l}}}^{\mathrm{H}}\right)}
\right. \\ & \left. \medmath{ \odot\left(\mathbf{1}_{{M}_1}^{\mathrm{H}}\mathbf{D}_{\mathrm{\gamma}}^{} \right) \odot \left(\mathbf{A}_{\mathrm{t}_{\mathrm{u}}}^{\mathrm{H}}  \mathbf{F}   \mathbf{F}^{\mathrm{H}} \mathbf{a}_{\mathrm{t}_{\mathrm{u}}}^{[\mathrm{0}]}\right)^{\mathrm{T}} \odot \mathbf{R}_{0}^{\mathrm{L}_1}}\right\}, \label{equ:FIM_exact_submat_theta_r_u_0_theta_r_l}
  \end{split}
 \\
  \begin{split}
 \medmath{\mathbf{J}_{\bm{\theta}_{\mathrm{r}_{\mathrm{u}}}^{\mathrm{[0]}} \bm{\phi}_{\mathrm{r}_{\mathrm{l}}}}}
 &=  \medmath{ \frac{2}{\sigma^2}\sum_{q = 1}^{N_{Q}}\Re}\left\{\medmath{\left(\beta^{{\mathrm{H}}[\mathrm{0}]} \mathbf{a}_{\mathrm{r}_{\mathrm{u}}}^{{\mathrm{H}}[0] } {\mathbf{K}}_{\mathrm{r}_{\mathrm{u}}}^{{\mathrm{H}}[0]}   \mathbf{A}_{\mathrm{r}_{\mathrm{u}}} \mathbf{B}\right) \odot\left( \mathbf{p}_{q,\mathrm{r}_{\mathrm{l}}}^{\mathrm{H}}\right) 
} \right. \\ & \left.  \medmath{ \odot\left(\mathbf{1}_{{M}_1}^{\mathrm{H}}\mathbf{D}_{\mathrm{\gamma}}^{} \right) \odot \left(\mathbf{A}_{\mathrm{t}_{\mathrm{u}}}^{\mathrm{H}}  \mathbf{F}   \mathbf{F}^{\mathrm{H}} \mathbf{a}_{\mathrm{t}_{\mathrm{u}}}^{[\mathrm{0}]}\right)^{\mathrm{T}} \odot \mathbf{R}_{0}^{\mathrm{L}_1}}\right\}, \label{equ:FIM_exact_submat_theta_r_u_0_phi_r_l}   \end{split}
 \\
   \begin{split}
\medmath{\mathbf{J}_{\bm{\theta}_{\mathrm{r}_{\mathrm{u}}}^{\mathrm{[0]}} \bm{\theta}_{\mathrm{t}_{\mathrm{u}}}}}
 &=   \medmath{-\frac{2}{\sigma^2}\sum_{q = 1}^{N_{Q}}\Re}\left\{\medmath{\left(\beta^{{\mathrm{H}}[\mathrm{0}]} \mathbf{a}_{\mathrm{r}_{\mathrm{u}}}^{{\mathrm{H}}[0] } {\mathbf{K}}_{\mathrm{r}_{\mathrm{u}}}^{{\mathrm{H}}[0]}   \mathbf{A}_{\mathrm{r}_{\mathrm{u}}} \mathbf{B}\right) \odot\left( \mathbf{k}_{q,\mathrm{l}}^{\mathrm{H}}\right)} \right. \\ & \left.  \medmath{
 \odot\left(\mathbf{1}_{{M}_1}^{\mathrm{H}}\mathbf{D}_{\mathrm{\gamma}}^{} \right) \odot \left(\mathbf{K}_{\mathrm{t}_{\mathrm{u}}}^{\mathrm{H}}  \mathbf{F}   \mathbf{F}^{\mathrm{H}} \mathbf{a}_{\mathrm{t}_{\mathrm{u}}}^{[\mathrm{0}]}\right)^{\mathrm{T}} \odot \mathbf{R}_{0}^{\mathrm{L}_1}}\right\}, \label{equ:FIM_exact_submat_theta_r_u_0_theta_t_u}  \end{split}
 \\
    \begin{split}
\medmath{\mathbf{J}_{\bm{\theta}_{\mathrm{r}_{\mathrm{u}}}^{\mathrm{[0]}} \bm{\phi}_{\mathrm{t}_{\mathrm{u}}}}}
 &=   \medmath{-\frac{2}{\sigma^2}\sum_{q = 1}^{N_{Q}}\Re}\left\{\medmath{\left(\beta^{{\mathrm{H}}[\mathrm{0}]} \mathbf{a}_{\mathrm{r}_{\mathrm{u}}}^{{\mathrm{H}}[0] } {\mathbf{K}}_{\mathrm{r}_{\mathrm{u}}}^{{\mathrm{H}}[0]}   \mathbf{A}_{\mathrm{r}_{\mathrm{u}}} \mathbf{B}\right) \odot\left( \mathbf{k}_{q,\mathrm{l}}^{\mathrm{H}}\right)
}  \right. \\ & \left. \medmath{ \odot\left(\mathbf{1}_{{M}_1}^{\mathrm{H}}\mathbf{D}_{\mathrm{\gamma}}^{} \right) \odot \left(\mathbf{P}_{\mathrm{t}_{\mathrm{u}}}^{\mathrm{H}}  \mathbf{F}   \mathbf{F}^{\mathrm{H}} \mathbf{a}_{\mathrm{t}_{\mathrm{u}}}^{[\mathrm{0}]}\right)^{\mathrm{T}} \odot \mathbf{R}_{0}^{\mathrm{L}_1}}\right\}, \label{equ:FIM_exact_submat_theta_r_u_0_phi_t_u}
  \end{split}
 \\
 \begin{split}
\medmath{\mathbf{J}_{\bm{\theta}_{\mathrm{r}_{\mathrm{u}}}^{\mathrm{[0]}} \bm{\tau}}}
 &=   \medmath{\frac{2}{\sigma^2}\sum_{q = 1}^{N_{Q}}\Re} \left\{\medmath{\left(\beta^{{\mathrm{H}}[\mathrm{0}]} \mathbf{a}_{\mathrm{r}_{\mathrm{u}}}^{{\mathrm{H}}[0] } {\mathbf{K}}_{\mathrm{r}_{\mathrm{u}}}^{{\mathrm{H}}[0]}   \mathbf{A}_{\mathrm{r}_{\mathrm{u}}} \mathbf{B}\right) \odot\left( \mathbf{k}_{q,\mathrm{l}}^{\mathrm{H}}\right)}  \right. \\ & \left.  \medmath{
 \odot\left(\mathbf{1}_{{M}_1}^{\mathrm{H}}\mathbf{D}_{\mathrm{\gamma}}^{} \right) \odot \left(\mathbf{A}_{\mathrm{t}_{\mathrm{u}}}^{\mathrm{H}}  \mathbf{F}   \mathbf{F}^{\mathrm{H}} \mathbf{a}_{\mathrm{t}_{\mathrm{u}}}^{[\mathrm{0}]}\right)^{\mathrm{T}} \odot \mathbf{R}_{1}^{\mathrm{L}}}\right\}, \label{equ:FIM_exact_submat_theta_r_u_0_tau} 
   \end{split}
\\
  \begin{split}
\medmath{\mathbf{J}_{\bm{\theta}_{\mathrm{r}_{\mathrm{u}}}^{\mathrm{[0]}} \bm{\beta}_{\mathrm{I}}} }&+\medmath{j\mathbf{J}_{\bm{\theta}_{\mathrm{r}_{\mathrm{u}}}^{\mathrm{[0]}} \bm{\beta}_{\mathrm{R}}}}
 =    \medmath{-\frac{2}{\sigma^2}\sum_{q = 1}^{N_{Q}}}\left\{\medmath{\left(\beta^{{\mathrm{H}}[\mathrm{0}]} \mathbf{a}_{\mathrm{r}_{\mathrm{u}}}^{{\mathrm{H}}[0] } {\mathbf{K}}_{\mathrm{r}_{\mathrm{u}}}^{{\mathrm{H}}[0]}   \mathbf{A}_{\mathrm{r}_{\mathrm{u}}}  \right) \odot\left( \mathbf{k}_{q,\mathrm{l}}^{\mathrm{H}}\right)}  \right. \\ & \left. \medmath{
 \odot\left(\mathbf{1}_{{M}_1}^{\mathrm{H}}\mathbf{D}_{\mathrm{\gamma}}^{}\right)
 \odot\left(\mathbf{A}_{\mathrm{t}_{\mathrm{u}}}^{\mathrm{H}}  \mathbf{F}   \mathbf{F}^{\mathrm{H}} \mathbf{a}_{\mathrm{t}_{\mathrm{u}}}^{[\mathrm{0}]}\right)^{\mathrm{T}} \odot \mathbf{R}_{0}^{\mathrm{L}_1}}\right\}.\label{equ:FIM_exact_submat_theta_r_u_0_beta}
   \end{split}
 \end{align}
\end{subequations}
Other entries related to $\bm{\psi}_{}/\{\bm{\theta}_{\mathrm{r}_{\mathrm{u}}}^{\mathrm{[0]}}\}$ can be obtained similarly by making appropriate matrix substitutions.

\subsection{Entries of the EFIM of the LOS}
\label{appendix:EFIM_Los}
 Defining $\mathbf{R}_{0}^{\mathrm{L}_2}$ as the signal factor considering just the LOS path; the entries in the FIM $\mathbf{J}_{\psi}^{\mathrm{D}} \in \mathbb{R}^{7 \times 7}$ are presented below
	\begin{subequations}[equ:FIM_exact_submat_theta_r_u_0_1]
\begin{align}
\begin{split}
\medmath{\mathbf{J}_{\bm{\theta}_{\mathrm{r}_{\mathrm{u}}}^{\mathrm{[0]}} \bm{\theta}_{\mathrm{r}_{\mathrm{u}}}^{\mathrm{[0]}}}}
 &=   \medmath{\frac{2}{\sigma^2}\Re}\left\{\medmath{\left(T\left| \beta^{{\mathrm{}}[\mathrm{0}]} \right|^2  \mathbf{a}_{\mathrm{r}_{\mathrm{u}}}^{{\mathrm{H}}[0] } {\mathbf{K}}_{\mathrm{r}_{\mathrm{u}}}^{{\mathrm{H}}[0]} {\mathbf{K}}_{\mathrm{r}_{\mathrm{u}}}^{{}[0]} \mathbf{a}_{\mathrm{r}_{\mathrm{u}}}^{{}[0] } \right)} \right. \\ & \left.   \medmath{
 \odot \left(\mathbf{a}_{\mathrm{t}_{\mathrm{u}}}^{\mathrm{H} [\mathrm{0}]}  \mathbf{F}   \mathbf{F}^{\mathrm{H}}  \mathbf{a}_{\mathrm{t}_{\mathrm{u}}}^{[\mathrm{0}]}\right)^{\mathrm{T}} \odot \mathbf{R}_{0}^{\mathrm{L}_2}}\right\}, \label{equ:FIM_exact_submat_theta_r_u_0_theta_r_u_0} \end{split} \\
\begin{split}
\medmath{\mathbf{J}_{\bm{\theta}_{\mathrm{r}_{\mathrm{u}}}^{\mathrm{[0]}} \bm{\phi}_{\mathrm{r}_{\mathrm{u}}}^{\mathrm{[0]}}}}
 &=  \medmath{ \frac{2}{\sigma^2}\Re}\left\{\medmath{\left(T\left| \beta^{{\mathrm{}}[\mathrm{0}]} \right|^2  \mathbf{a}_{\mathrm{r}_{\mathrm{u}}}^{{\mathrm{H}}[0] } {\mathbf{K}}_{\mathrm{r}_{\mathrm{u}}}^{{\mathrm{H}}[0]} {\mathbf{P}}_{\mathrm{r}_{\mathrm{u}}}^{{}[0]} \mathbf{a}_{\mathrm{r}_{\mathrm{u}}}^{{}[0] } \right) } \right. \\ & \left.   \medmath{ 
 \odot \left(\mathbf{a}_{\mathrm{t}_{\mathrm{u}}}^{\mathrm{H} [\mathrm{0}]}  \mathbf{F}   \mathbf{F}^{\mathrm{H}}  \mathbf{a}_{\mathrm{t}_{\mathrm{u}}}^{[\mathrm{0}]}\right)^{\mathrm{T}} \odot \mathbf{R}_{0}^{\mathrm{L}_2}}\right\}, \label{equ:FIM_exact_submat_theta_r_u_0_phi_r_u_0} \end{split} \\
\begin{split}
\medmath{\mathbf{J}_{\bm{\theta}_{\mathrm{r}_{\mathrm{u}}}^{\mathrm{[0]}} \bm{\theta}_{\mathrm{t}_{\mathrm{u}}}^{\mathrm{[0]}}}}
 &=  \medmath{ -\frac{2}{\sigma^2}\Re}\left\{\medmath{\left(T\left| \beta^{{\mathrm{}}[\mathrm{0}]} \right|^2  \mathbf{a}_{\mathrm{r}_{\mathrm{u}}}^{{\mathrm{H}}[0] } {\mathbf{K}}_{\mathrm{r}_{\mathrm{u}}}^{{\mathrm{H}}[0]}  \mathbf{a}_{\mathrm{r}_{\mathrm{u}}}^{{}[0] } \right) }  \right. \\ & \left.   \medmath{
 \odot \left(\mathbf{a}_{\mathrm{t}_{\mathrm{u}}}^{\mathrm{H} [\mathrm{0}]} {\mathbf{K}}_{\mathrm{t}_{\mathrm{u}}}^{\mathrm{H}[0]}  \mathbf{F}   \mathbf{F}^{\mathrm{H}}  \mathbf{a}_{\mathrm{t}_{\mathrm{u}}}^{[\mathrm{0}]}\right)^{\mathrm{T}} \odot \mathbf{R}_{0}^{\mathrm{L}_2}}\right\}, \label{equ:FIM_exact_submat_theta_r_u_0_theta_t_u_0} \end{split} \\
  \begin{split}
\medmath{\mathbf{J}_{\bm{\theta}_{\mathrm{r}_{\mathrm{u}}}^{\mathrm{[0]}} \bm{\phi}_{\mathrm{t}_{\mathrm{u}}}^{\mathrm{[0]}}}}
 &=   \medmath{-\frac{2}{\sigma^2}\Re}\left\{\medmath{\left(T\left| \beta^{{\mathrm{}}[\mathrm{0}]} \right|^2  \mathbf{a}_{\mathrm{r}_{\mathrm{u}}}^{{\mathrm{H}}[0] } {\mathbf{K}}_{\mathrm{r}_{\mathrm{u}}}^{{\mathrm{H}}[0]}  \mathbf{a}_{\mathrm{r}_{\mathrm{u}}}^{{}[0] } \right) }  \right. \\ & \left.   \medmath{
 \odot \left(\mathbf{a}_{\mathrm{t}_{\mathrm{u}}}^{\mathrm{H} [\mathrm{0}]} {\mathbf{P}}_{\mathrm{t}_{\mathrm{u}}}^{\mathrm{H}[0]}  \mathbf{F}   \mathbf{F}^{\mathrm{H}}  \mathbf{a}_{\mathrm{t}_{\mathrm{u}}}^{[\mathrm{0}]}\right)^{\mathrm{T}} \odot \mathbf{R}_{0}^{\mathrm{L}_2}}\right\}, \label{equ:FIM_exact_submat_theta_r_u_0_phi_t_u_0}  \end{split} \\
  \begin{split}
\medmath{\mathbf{J}_{\bm{\theta}_{\mathrm{r}_{\mathrm{u}}}^{\mathrm{[0]}} \bm{\tau}^{\mathrm{[0]}}}}
 &=  \medmath{ \frac{2}{\sigma^2}\Re}\left\{\medmath{\left(T\left| \beta^{{\mathrm{}}[\mathrm{0}]} \right|^2  \mathbf{a}_{\mathrm{r}_{\mathrm{u}}}^{{\mathrm{H}}[0] } {\mathbf{K}}_{\mathrm{r}_{\mathrm{u}}}^{{\mathrm{H}}[0]}  \mathbf{a}_{\mathrm{r}_{\mathrm{u}}}^{{}[0] } \right) }  \right. \\ & \left.   \medmath{
 \odot \left(\mathbf{a}_{\mathrm{t}_{\mathrm{u}}}^{\mathrm{H} [\mathrm{0}]}   \mathbf{F}   \mathbf{F}^{\mathrm{H}}  \mathbf{a}_{\mathrm{t}_{\mathrm{u}}}^{[\mathrm{0}]}\right)^{\mathrm{T}} \odot \mathbf{R}_{1}^{\mathrm{L}_2}}\right\}, \label{equ:FIM_exact_submat_theta_r_u_0_tau_0} \end{split}\\
  \begin{split}
\medmath{\mathbf{J}_{\bm{\theta}_{\mathrm{r}_{\mathrm{u}}}^{\mathrm{[0]}}\bm{\beta}_{\mathrm{I}}^{\mathrm{[0]}} }} &\medmath{+j\mathbf{J}_{\bm{\theta}_{\mathrm{r}_{\mathrm{u}}}^{\mathrm{[0]}}\bm{\beta}_{\mathrm{R}}^{\mathrm{[0]}} }}
 =   \medmath{- \frac{2}{\sigma^2}}\left\{\medmath{\left(T \beta^{{\mathrm{H}}[\mathrm{0}]}   \mathbf{a}_{\mathrm{r}_{\mathrm{u}}}^{{\mathrm{H}}[0] } {\mathbf{K}}_{\mathrm{r}_{\mathrm{u}}}^{{\mathrm{H}}[0]}  \mathbf{a}_{\mathrm{r}_{\mathrm{u}}}^{{}[0] } \right) } \right. \\ & \left.   \medmath{
 \odot \left(\mathbf{a}_{\mathrm{t}_{\mathrm{u}}}^{\mathrm{H} [\mathrm{0}]}   \mathbf{F}   \mathbf{F}^{\mathrm{H}}  \mathbf{a}_{\mathrm{t}_{\mathrm{u}}}^{[\mathrm{0}]}\right)^{\mathrm{T}} \odot \mathbf{R}_{0}^{\mathrm{L}_2}}\right\}. \label{equ:FIM_exact_submat_theta_r_u_0_beta_0}
 \end{split}
 \end{align}
\end{subequations}
The other entries related to $\bm{\psi}_{}/\{\bm{\theta}_{\mathrm{r}_{\mathrm{u}}}^{\mathrm{[0]}}  \}$ can be obtained similarly by making appropriate matrix substitutions. Subsequently, the EFIM $\Bar{\mathbf{J}}_{ {\bm{\psi}}_{1}}^{\mathrm{e}}$ can be obtained as $\Bar{\mathbf{J}}_{ {\bm{\psi}}_{1}}^{\mathrm{e}}  = {\mathbf{J}}_{ {{{\bm{\psi}}^{}_{1}
}} {\bm{\psi}}^{}_{1}} - {\mathbf{J}}_{ {{{\bm{\psi}}^{}_{1}
}} {\bm{\psi}}^{}_{2}} {\mathbf{J}}_{ {{{\bm{\psi}}^{}_{2}
}} {\bm{\psi}}^{}_{2}}^{\mathrm{-1}} {\mathbf{J}}_{ {{{\bm{\psi}}^{}_{1}
}} {\bm{\psi}}^{}_{2}}^{\mathrm{T}} $. 
\subsection{RIS Related Entries in the Transformation Matrix}
\label{appendix:entries_transformation}
 The non-zero terms related to the angle of departure at the BS are derived. 
The non-zero derivatives of $\theta_{\mathrm{t}_{\mathrm{u}}}^{[{m}]}$ are $\frac{\partial {\theta}_{\mathrm{t}_{\mathrm{u}}}^{[{m}]}}{\partial \mathbf{p}^{[{m}]}} = \frac{ \left[p_{x}^{[{m}]} g_{z}^{[{m}]}, p_{y}^{[{m}]} g_{z}^{[{m}]},\|\mathbf{g}^{[{m}]}_{{}}\|^2 -p_{z}^{[{m}]} g_{z}^{[{m}]} \right]^{\mathrm{T}} }{\|\mathbf{g}^{[{m}]}_{{}}\|^2\sqrt{({g}^{[{m}]}_{x_{}})^2+({g}^{[{m}]}_{y_{}})^2}} 
$
and the non-zero derivatives  of $\phi_{\mathrm{t}_{\mathrm{u}}}^{[{m}]}$ are $\frac{\partial {\phi}_{\mathrm{t}_{\mathrm{u}}}^{[{m}]}}{\partial \mathbf{p}^{[{m}]}} = \frac{ \left[ -g_{y}^{[{m}]} , g_{x}^{[{m}]}, 0 \right]^{\mathrm{T}} }{{({g}^{[{m}]}_{x_{}})^2+({g}^{[{m}]}_{y_{}})^2}} 
$. The non-zero terms related to the angle of arrival at the RIS, $( \mathbf{\theta}_{\mathrm{r}_{\mathrm{l}}}^{[{m}]}, \mathbf{\phi}_{\mathrm{r}_{\mathrm{l}}} ^{[{m}]})$ are derived next. The non-zero derivatives of $\theta_{\mathrm{r}_{\mathrm{l}}}^{[{m}]}$ are $\frac{\partial {\theta}_{\mathrm{r}_{\mathrm{l}}}^{[{m}]}}{\partial {\theta}^{[{m}]}_{0}} =-\frac{{c}^{[{m}]}_{\Tilde{y}}}{\sqrt{({c}^{[{m}]}_{\Tilde{x}})^2+({c}^{[{m}]}_{\Tilde{y}})^2}}$, $\frac{\partial {\theta}_{\mathrm{r}_{\mathrm{l}}}^{[{m}]}}{\partial {\phi}^{[{m}]}_{0}} = \frac{ {c}^{[{m}]}_{\Tilde{x}} \sin \theta_{0}^{[{m}]}}{\sqrt{({c}^{[{m}]}_{\Tilde{x}})^2+({c}^{[{m}]}_{\Tilde{y}})^2}}$, and


$$
\frac{\partial {\theta}_{\mathrm{r}_{\mathrm{l}}}^{[{m}]}}{\partial \mathbf{p}^{[{m}]}} = \frac{ 1}{\sqrt{({c}^{[{m}]}_{\Tilde{x}})^2+({c}^{[{m}]}_{\Tilde{y}})^2}} \left( \mathbf{q}_{3}^{[{m}]} - \frac{{c}^{[{m}]}_{\Tilde{z}}  \mathbf{c}^{[{m}]}}{\|\Tilde{\mathbf{c}}^{[{m}]}_{}\|^2} \right).
$$

The non-zero derivatives of ${\phi}_{\mathrm{r}_{\mathrm{l}}}^{[{m}]}$ are $\frac{\partial {\phi}_{\mathrm{r}_{\mathrm{l}}}^{[{m}]}}{\partial {\theta}^{[{m}]}_{0}} =-\frac{{c}^{[{m}]}_{\Tilde{x}} {c}^{[{m}]}_{\Tilde{z}}}{{({c}^{[{m}]}_{\Tilde{x}})^2+({c}^{[{m}]}_{\Tilde{y}})^2}}$, $\frac{\partial {\phi}_{\mathrm{r}_{\mathrm{l}}}^{[{m}]}}{\partial \mathbf{p}^{[{m}]}} =-\frac{(\mathbf{q}_{2}^{[{m}]}{\mathbf{q}_{1}^{[{m}]}}^{\mathrm{T}} - \mathbf{q}_{1}^{[{m}]}{\mathbf{q}_{2}^{[{m}]}}^{\mathrm{T}})\mathbf{c}^{[{m}]}}{{({c}^{[{m}]}_{\Tilde{x}})^2+({c}^{[{m}]}_{\Tilde{y}})^2}}$,

and
$$
\frac{\partial {\phi}_{\mathrm{r}_{\mathrm{l}}}^{[{m}]}}{\partial {\phi}^{[{m}]}_{0}} =\frac{-({c}^{[{m}]}_{\Tilde{x}})^2 \cos \theta_{0}^{[{m}]}+\left(c_{x}^{[{m}]} \sin \phi_{0}^{[{m}]}-c_{y}^{[{m}]} \cos \phi_{0}^{[{m}]}\right) {({c}^{[{m}]}_{\Tilde{y}})}}{{({c}^{[{m}]}_{\Tilde{x}})^2+({c}^{[{m}]}_{\Tilde{y}})^2}}.
$$

Except, the non-zero derivatives of $\theta_{\mathrm{t}_{\mathrm{l}}}^{[{m}]}$ and ${\phi}_{\mathrm{t}_{\mathrm{l}}}^{[{m}]}$ related to the UE position which are given by
$
\frac{\partial {\theta}_{\mathrm{t}_{\mathrm{l}}}^{[{m}]}}{\partial \mathbf{p}^{}} = -\frac{ 1}{\sqrt{({v}^{[{m}]}_{\Tilde{x}})^2+({v}^{[{m}]}_{\Tilde{y}})^2}} \left( \mathbf{q}_{3}^{[{m}]} - \frac{{v}^{[{m}]}_{\Tilde{z}}  \mathbf{v}^{[{m}]}}{\|\Tilde{\mathbf{v}}^{[{m}]}_{}\|^2} \right), 
$
$
\frac{\partial {\phi}_{\mathrm{t}_{\mathrm{l}}}^{[{m}]}}{\partial \mathbf{p}^{}} =\frac{(\mathbf{q}_{2}^{[{m}]}{\mathbf{q}_{1}^{[{m}]}}^{\mathrm{T}} - \mathbf{q}_{1}^{[{m}]}{\mathbf{q}_{2}^{[{m}]}}^{\mathrm{T}})\mathbf{v}^{[{m}]}}{{({v}^{[{m}]}_{\Tilde{x}})^2+({v}^{[{m}]}_{\Tilde{y}})^2}},
$
all other non-zero derivatives can be obtained by replacing $c$ in the derivatives for $\theta_{\mathrm{r}_{\mathrm{l}}}^{[{m}]}$ and ${\phi}_{\mathrm{r}_{\mathrm{l}}}^{[{m}]}$ with $v$. The non-zero derivatives of $\theta_{\mathrm{r}_{\mathrm{u}}}^{[{m}]}$ are $\frac{\partial {\theta}_{\mathrm{r}_{\mathrm{u}}}^{[{m}]}}{\partial {\theta}^{}_{0}} =-\frac{{e}^{[{m}]}_{\Tilde{y}}}{\sqrt{({e}^{[{m}]}_{\Tilde{x}})^2+({e}^{[{m}]}_{\Tilde{y}})^2}}$, $\frac{\partial {\theta}_{\mathrm{r}_{\mathrm{u}}}^{[{m}]}}{\partial {\phi}^{}_{0}} = \frac{ {e}^{[{m}]}_{\Tilde{x}} \sin \theta_{0}^{}}{\sqrt{({e}^{[{m}]}_{\Tilde{x}})^2+({e}^{[{m}]}_{\Tilde{y}})^2}}$, $\frac{\partial {\theta}_{\mathrm{r}_{\mathrm{u}}}^{[{m}]}}{\partial \mathbf{p}^{[{m}]}} = -\frac{ 1}{\sqrt{({e}^{[{m}]}_{\Tilde{x}})^2+({e}^{[{m}]}_{\Tilde{y}})^2}} \left( \mathbf{q}_{3}^{[{m}]} - \frac{{e}^{[{m}]}_{\Tilde{z}}  \mathbf{e}^{[{m}]}}{\|\Tilde{\mathbf{e}}^{[{m}]}_{}\|^2} \right) $, and $\frac{\partial {\theta}_{\mathrm{r}_{\mathrm{u}}}^{[{m}]}}{\partial \mathbf{p}^{}} = \frac{ 1}{\sqrt{({e}^{[{m}]}_{\Tilde{x}})^2+({e}^{[{m}]}_{\Tilde{y}})^2}} \left( \mathbf{q}_{3}^{[{m}]} - \frac{{e}^{[{m}]}_{\Tilde{z}}  \mathbf{e}^{[{m}]}}{\|\Tilde{\mathbf{e}}^{[{m}]}_{}\|^2} \right). \\
$



The non-zero derivatives of ${\phi}_{\mathrm{r}_{\mathrm{u}}}^{[{m}]}$ are $\frac{\partial {\phi}_{\mathrm{r}_{\mathrm{u}}}^{[{m}]}}{\partial {\theta}^{}_{0}} =-\frac{{e}^{[{m}]}_{\Tilde{x}} {e}^{[{m}]}_{\Tilde{z}}}{{({e}^{[{m}]}_{\Tilde{x}})^2+({e}^{[{m}]}_{\Tilde{y}})^2}}$,  $\frac{\partial {\phi}_{\mathrm{r}_{\mathrm{u}}}^{[{m}]}}{\partial \mathbf{p}^{[{m}]}} =\frac{(\mathbf{q}_{2}^{[{m}]}{\mathbf{q}_{1}^{[{m}]}}^{\mathrm{T}} - \mathbf{q}_{1}^{[{m}]}{\mathbf{q}_{2}^{[{m}]}}^{\mathrm{T}})\mathbf{e}^{[{m}]}}{{({e}^{[{m}]}_{\Tilde{x}})^2+({e}^{[{m}]}_{\Tilde{y}})^2}}
$, $\frac{\partial {\phi}_{\mathrm{r}_{\mathrm{u}}}^{[{m}]}}{\partial \mathbf{p}^{}} =-\frac{(\mathbf{q}_{2}^{[{m}]}{\mathbf{q}_{1}^{[{m}]}}^{\mathrm{T}} - \mathbf{q}_{1}^{[{m}]}{\mathbf{q}_{2}^{[{m}]}}^{\mathrm{T}})\mathbf{e}^{[{m}]}}{{({e}^{[{m}]}_{\Tilde{x}})^2+({e}^{[{m}]}_{\Tilde{y}})^2}}
$, and $\medmath{\frac{\partial {\phi}_{\mathrm{r}_{\mathrm{u}}}^{[{m}]}}{\partial {\phi}^{}_{0}} =\frac{-({e}^{[{m}]}_{\Tilde{x}})^2 \cos \theta_{0}^{}+\left(e_{x}^{[{m}]} \sin \phi_{0}^{[{m}]} - e_{y}^{[{m}]} \cos \phi_{0}^{[{m}]}\right) {({e}^{[{m}]}_{\Tilde{y}})}}{{({e}^{[{m}]}_{\Tilde{x}})^2+({e}^{[{m}]}_{\Tilde{y}})^2}}}.$




The delay related derivatives are $\medmath{\frac{\partial {\tau}_{\mathrm{}}^{[{m}]}}{\partial \mathbf{p}^{[{m}]}} =\frac{\mathbf{p}^{[{m}]} -\mathbf{p}_{\mathrm{BS}}}{c \|\mathbf{p}^{[{m}]} -\mathbf{p}_{\mathrm{BS}}  \|} 
-\frac{\mathbf{p} - \mathbf{p}^{[{m}]}}{c \|\mathbf{p} - \mathbf{p}^{[{m}]} \|} $ and $\frac{\partial {\tau}_{\mathrm{}}^{}}{\partial \mathbf{p}^{}} =\frac{\mathbf{p} - \mathbf{p}^{[{m}]}}{c \|\mathbf{p} - \mathbf{p}^{[{m}]} \|}.}$

\subsection{LOS related Entries in the Transformation Matrix}
\label{appendix:los_entries_transformation}
The non-zero terms related to the LOS angles of departure and LOS angles of arrival can be obtained by making appropriate substitutions in the derivatives of $(\theta_{\mathrm{t}_{\mathrm{u}}}^{[{m}]}, \phi_{\mathrm{t}_{\mathrm{u}}}^{[{m}]})$ and derivatives of $(\theta_{\mathrm{r}_{\mathrm{u}}}^{[{m}]}, \phi_{\mathrm{r}_{\mathrm{u}}}^{[{m}]})$ respectively. The delay related derivatives are $\frac{\partial {\tau}_{\mathrm{}}^{[0]}}{\partial \mathbf{p}^{}} =\frac{\mathbf{p} - \mathbf{p}_{\mathrm{BS}}}{c \|\mathbf{p} - \mathbf{p}_{\mathrm{BS}} \|}.$
\bibliographystyle{IEEEtran}
\bibliography{refs}

\begin{thebibliography}{10}
\providecommand{\url}[1]{#1}
\csname url@samestyle\endcsname
\providecommand{\newblock}{\relax}
\providecommand{\bibinfo}[2]{#2}
\providecommand{\BIBentrySTDinterwordspacing}{\spaceskip=0pt\relax}
\providecommand{\BIBentryALTinterwordstretchfactor}{4}
\providecommand{\BIBentryALTinterwordspacing}{\spaceskip=\fontdimen2\font plus
\BIBentryALTinterwordstretchfactor\fontdimen3\font minus
  \fontdimen4\font\relax}
\providecommand{\BIBforeignlanguage}[2]{{%
\expandafter\ifx\csname l@#1\endcsname\relax
\typeout{** WARNING: IEEEtran.bst: No hyphenation pattern has been}%
\typeout{** loaded for the language `#1'. Using the pattern for}%
\typeout{** the default language instead.}%
\else
\language=\csname l@#1\endcsname
\fi
#2}}
\providecommand{\BIBdecl}{\relax}
\BIBdecl

\bibitem{emenonye2023_ICC_conf_workshop}
D.-R. Emenonye, H.~S. Dhillon, and R.~M. Buehrer, ``Limitations of {RIS}-aided
  localization: Inspecting the relationships between channel parameters,''
  \emph{accepted to {Proc., IEEE Intl. Conf. on Commun. (ICC) Workshop}}, 2023.

\bibitem{9366805}
L.~Wei, C.~Huang, G.~C. Alexandropoulos, C.~Yuen, Z.~Zhang, and M.~Debbah,
  ``Channel estimation for {RIS}-empowered multi-user {MISO} wireless
  communications,'' \emph{IEEE Trans. on Commun.}, vol.~69, no.~6, pp.
  4144--4157, Mar 2021.

\bibitem{9779586}
L.~Wei, C.~Huang, G.~C. Alexandropoulos, W.~E.~I. Sha, Z.~Zhang, M.~Debbah, and
  C.~Yuen, ``Multi-user holographic {MIMO} surfaces: Channel modeling and
  spectral efficiency analysis,'' \emph{IEEE Journal on Sel. Areas in Signal
  Processing}, vol.~16, no.~5, pp. 1112--1124, May 2022.

\bibitem{8741198}
C.~Huang, A.~Zappone, G.~C. Alexandropoulos, M.~Debbah, and C.~Yuen,
  ``Reconfigurable intelligent surfaces for energy efficiency in wireless
  communication,'' \emph{IEEE Trans. on Wireless Commun.}, vol.~18, no.~8, pp.
  4157--4170, Aug. 2019.

\bibitem{8982186}
H.~Guo, Y.-C. Liang, J.~Chen, and E.~G. Larsson, ``Weighted sum-rate
  maximization for reconfigurable intelligent surface aided wireless
  networks,'' \emph{IEEE Trans. on Wireless Commun.}, vol.~19, no.~5, pp.
  3064--3076, May 2020.

\bibitem{8811733}
Q.~Wu and R.~Zhang, ``Intelligent reflecting surface enhanced wireless network
  via joint active and passive beamforming,'' \emph{IEEE Trans. on Wireless
  Commun.}, vol.~18, no.~11, pp. 5394--5409, Nov. 2019.

\bibitem{keykhosravi2021multi}
K.~Keykhosravi and H.~Wymeersch, ``Multi-{RIS} discrete-phase encoding for
  interpath-interference-free channel estimation,'' \emph{arXiv:2106.07065},
  2021.

\bibitem{9528041}
K.~Keykhosravi, M.~F. Keskin, S.~Dwivedi, G.~Seco-Granados, and H.~Wymeersch,
  ``Semi-passive {3D} positioning of multiple {RIS}-enabled users,'' \emph{IEEE
  Trans. on Veh. Technology}, vol.~70, no.~10, pp. 11\,073--11\,077, Oct. 2021.

\bibitem{fascista2021ris}
A.~Fascista, A.~Coluccia, H.~Wymeersch, and G.~Seco-Granados, ``{RIS}-aided
  joint localization and synchronization with a single-antenna {mmWave}
  receiver,'' in \emph{Proc., IEEE Intl. Conf. on Acoustics, Speech, and Sig.
  Proc. (ICASSP)}.\hskip 1em plus 0.5em minus 0.4em\relax IEEE, 2021, pp.
  4455--4459.

\bibitem{wymeersch2020beyond}
H.~Wymeersch and B.~Denis, ``Beyond {5G} wireless localization with
  reconfigurable intelligent surfaces,'' in \emph{Proc., IEEE Intl. Conf. on
  Commun. (ICC)}.\hskip 1em plus 0.5em minus 0.4em\relax IEEE, 2020.

\bibitem{9500281}
K.~Keykhosravi, M.~F. Keskin, G.~Seco-Granados, and H.~Wymeersch, ``{SISO}
  {RIS}-enabled joint {3D} downlink localization and synchronization,'' in
  \emph{Proc., IEEE Intl. Conf. on Commun. (ICC)}, 2021.

\bibitem{9513781}
E.~Čišija, A.~M. Ahmed, A.~Sezgin, and H.~Wymeersch, ``{RIS}-aided {mmWave}
  {MIMO} radar system for adaptive multi-target localization,'' in \emph{Proc.,
  IEEE Statistical Signal Processing Workshop (SSP)}, 2021, pp. 196--200.

\bibitem{9726785}
G.~C. Alexandropoulos, I.~Vinieratou, and H.~Wymeersch, ``Localization via
  multiple reconfigurable intelligent surfaces equipped with single receive
  {RF} chains,'' \emph{IEEE Wireless Commun. Letters}, vol.~11, no.~5, pp.
  1072--1076, 2022.

\bibitem{keykhosravi2022ris}
K.~Keykhosravi, G.~Seco-Granados, G.~C. Alexandropoulos, and H.~Wymeersch,
  ``{RIS}-enabled self-localization: Leveraging controllable reflections with
  zero access points,'' \emph{arXiv:2202.11159}, 2022.

\bibitem{rahal2021ris}
M.~Rahal, B.~Denis, K.~Keykhosravi, B.~Uguen, and H.~Wymeersch, ``{RIS}-enabled
  localization continuity under near-field conditions,'' in \emph{Proc., IEEE
  Intl. Conf. on Acoustics, Speech, and Sig. Proc. (ICASSP)}.\hskip 1em plus
  0.5em minus 0.4em\relax IEEE, 2021, pp. 436--440.

\bibitem{9500663}
Z.~Abu-Shaban, K.~Keykhosravi, M.~F. Keskin, G.~C. Alexandropoulos,
  G.~Seco-Granados, and H.~Wymeersch, ``Near-field localization with a
  reconfigurable intelligent surface acting as lens,'' in \emph{Proc., IEEE
  Intl. Conf. on Commun. (ICC)}, 2021.

\bibitem{9625826}
D.~Dardari, N.~Decarli, A.~Guerra, and F.~Guidi, ``{LOS/NLOS} near-field
  localization with a large reconfigurable intelligent surface,'' \emph{IEEE
  Trans. on Wireless Commun.}, to appear.

\bibitem{9508872}
A.~Elzanaty, A.~Guerra, F.~Guidi, and M.-S. Alouini, ``Reconfigurable
  intelligent surfaces for localization: Position and orientation error
  bounds,'' \emph{IEEE Trans. on Signal Processing}, vol.~69, pp. 5386--5402,
  Aug. 2021.

\bibitem{garcia2017direct}
N.~Garcia, H.~Wymeersch, E.~G. Larsson, A.~M. Haimovich, and M.~Coulon,
  ``Direct localization for massive {MIMO},'' \emph{IEEE Trans. on Signal
  Processing}, vol.~65, no.~10, pp. 2475--2487, May 2017.

\bibitem{8240645}
A.~Shahmansoori, G.~E. Garcia, G.~Destino, G.~Seco-Granados, and H.~Wymeersch,
  ``Position and orientation estimation through millimeter-wave {MIMO in {5G}}
  systems,'' \emph{IEEE Trans. on Wireless Commun.}, vol.~17, no.~3, pp.
  1822--1835, Mar. 2018.

\bibitem{8515231}
R.~Mendrzik, H.~Wymeersch, G.~Bauch, and Z.~Abu-Shaban, ``Harnessing {NLOS}
  components for position and orientation estimation in {5G} millimeter wave
  {MIMO},'' \emph{IEEE Trans. on Wireless Commun.}, vol.~18, no.~1, pp.
  93--107, Jan. 2019.

\bibitem{8356190}
Z.~Abu-Shaban, X.~Zhou, T.~Abhayapala, G.~Seco-Granados, and H.~Wymeersch,
  ``Error bounds for uplink and downlink {3D} localization in {5G} millimeter
  wave systems,'' \emph{IEEE Trans. on Wireless Commun.}, vol.~17, no.~8, pp.
  4939--4954, Aug. 2018.

\bibitem{guerra2018single}
A.~Guerra, F.~Guidi, and D.~Dardari, ``Single-anchor localization and
  orientation performance limits using massive arrays: {MIMO vs.
  beamforming},'' \emph{IEEE Trans. on Wireless Commun.}, vol.~17, no.~8, pp.
  5241--5255, Aug. 2018.

\bibitem{fascista2021downlink}
A.~Fascista, A.~Coluccia, H.~Wymeersch, and G.~Seco-Granados, ``Downlink
  single-snapshot localization and mapping with a single-antenna receiver,''
  \emph{IEEE Trans. on Wireless Commun.}, vol.~20, no.~7, pp. 4672--4684, July
  2021.

\bibitem{8755880}
------, ``Millimeter-wave downlink positioning with a single-antenna
  receiver,'' \emph{IEEE Trans. on Wireless Commun.}, vol.~18, no.~9, pp.
  4479--4490, Sep. 2019.

\bibitem{li2019massive}
X.~Li, E.~Leitinger, M.~Oskarsson, K.~{\AA}str{\"o}m, and F.~Tufvesson,
  ``Massive {MIMO}-based localization and mapping exploiting phase information
  of multipath components,'' \emph{IEEE Trans. on Wireless Commun.}, vol.~18,
  no.~9, pp. 4254--4267, Sep. 2019.

\bibitem{9082200}
F.~Ghaseminajm, Z.~Abu-Shaban, S.~S. Ikki, H.~Wymeersch, and C.~R. Benson,
  ``Localization error bounds for {5G mmWave} systems under {I/Q} imbalance,''
  \emph{IEEE Trans. on Veh. Technology}, vol.~69, no.~7, pp. 7971--7975, July
  2020.

\bibitem{kay1993fundamentals}
S.~M. Kay, \emph{{Fundamentals of Statistical Signal Processing: Estimation
  Theory}}.\hskip 1em plus 0.5em minus 0.4em\relax Prentice-Hall, Inc., 1993.

\bibitem{DBLP:books/wi/Trees02}
H.~L.~V. Trees, \emph{Optimum Array Processing: Part {IV} of Detection,
  Estimation, and Modulation Theory}.\hskip 1em plus 0.5em minus 0.4em\relax
  Wiley, 2002.

\bibitem{5571900}
Y.~Shen and M.~Z. Win, ``Fundamental limits of wideband localization — part
  {I}: A general framework,'' \emph{IEEE Trans. on Info. Theory}, vol.~56,
  no.~10, pp. 4956--4980, Oct. 2010.

\bibitem{5571889}
Y.~Shen, H.~Wymeersch, and M.~Z. Win, ``Fundamental limits of wideband
  localization — part {II}: Cooperative networks,'' \emph{IEEE Trans. on
  Info. Theory}, vol.~56, no.~10, pp. 4981--5000, Oct. 2010.

\bibitem{5762798}
M.~Z. Win, A.~Conti, S.~Mazuelas, Y.~Shen, W.~M. Gifford, D.~Dardari, and
  M.~Chiani, ``Network localization and navigation via cooperation,''
  \emph{IEEE Commun. Magazine}, vol.~49, no.~5, pp. 56--62, May 2011.

\bibitem{8421288}
M.~Z. Win, Y.~Shen, and W.~Dai, ``A theoretical foundation of network
  localization and navigation,'' \emph{Proceedings of the IEEE}, vol. 106,
  no.~7, pp. 1136--1165, July 2018.

\bibitem{6655520}
F.~Montorsi, S.~Mazuelas, G.~M. Vitetta, and M.~Win, ``On the impact of a
  priori information on localization accuracy and complexity,'' in \emph{Proc.,
  IEEE Intl. Conf. on Commun. (ICC)}, 2013, pp. 5792--5797.

\bibitem{8264804}
S.~Mazuelas, A.~Conti, J.~C. Allen, and M.~Z. Win, ``Soft range information for
  network localization,'' \emph{IEEE Trans. on Signal Processing}, vol.~66,
  no.~12, pp. 3155--3168, June 2018.

\bibitem{8827486}
A.~Conti, S.~Mazuelas, S.~Bartoletti, W.~C. Lindsey, and M.~Z. Win, ``Soft
  information for localization-of-things,'' \emph{Proceedings of the IEEE},
  vol. 107, no.~11, pp. 2240--2264, Nov. 2019.

\bibitem{zekavat2011handbook}
R.~Zekavat and R.~M. Buehrer, \emph{{Handbook of Position Location: Theory,
  Practice and Advances}}.\hskip 1em plus 0.5em minus 0.4em\relax John Wiley \&
  Sons, 2011, vol.~27.

\bibitem{9044729}
Y.~Li, S.~Ma, G.~Yang, and K.-K. Wong, ``Robust localization for mixed
  {LOS/NLOS} environments with anchor uncertainties,'' \emph{IEEE Trans. on
  Commun.}, vol.~68, no.~7, pp. 4507--4521, June 2020.

\bibitem{6862045}
Z.~Ma and K.~C. Ho, ``A study on the effects of sensor position error and the
  placement of calibration emitter for source localization,'' \emph{IEEE Trans.
  on Wireless Commun.}, vol.~13, no.~10, pp. 5440--5452, Oct. 2014.

\bibitem{7247270}
S.~Yousefi, R.~Monir~Vaghefi, X.-W. Chang, B.~Champagne, and R.~M. Buehrer,
  ``Sensor localization in {NLOS} environments with anchor uncertainty and
  unknown clock parameters,'' in \emph{Proc., IEEE Intl. Conf. on Commun.
  Workshops (ICC Workshops)}, 2015, pp. 742--747.

\bibitem{6684554}
G.~Naddafzadeh-Shirazi, M.~B. Shenouda, and L.~Lampe, ``Second order cone
  programming for sensor network localization with anchor position
  uncertainty,'' \emph{IEEE Trans. on Wireless Commun.}, vol.~13, no.~2, pp.
  749--763, Feb. 2014.

\bibitem{9186070}
K.~Gu, Y.~Wang, and Y.~Shen, ``Cooperative detection by multi-agent networks in
  the presence of position uncertainty,'' \emph{IEEE Trans. on Signal
  Processing}, vol.~68, pp. 5411--5426, 2020.

\bibitem{9064586}
Z.~Abu-Shaban, H.~Wymeersch, T.~Abhayapala, and G.~Seco-Granados,
  ``Single-anchor two-way localization bounds for {5G} mmwave systems,''
  \emph{IEEE Trans. on Veh. Technology}, vol.~69, no.~6, pp. 6388--6400, Apr.
  2020.

\bibitem{1097833}
J.-Y. Lee and R.~Scholtz, ``Ranging in a dense multipath environment using an
  {UWB} radio link,'' \emph{IEEE Journal on Sel. Areas in Commun.}, vol.~20,
  no.~9, pp. 1677--1683, Dec. 2002.

\bibitem{sohrabi2017hybrid}
F.~Sohrabi and W.~Yu, ``Hybrid analog and digital beamforming for mmwave {OFDM}
  large-scale antenna arrays,'' \emph{IEEE Journal on Sel. Areas in Commun.},
  vol.~35, no.~7, pp. 1432--1443, Apr. 2017.

\bibitem{horn2012matrix}
R.~A. Horn and C.~R. Johnson, \emph{{Matrix Analysis}}.\hskip 1em plus 0.5em
  minus 0.4em\relax Cambridge University Press, 2012.

\bibitem{ellingson2021path}
S.~W. Ellingson, ``Path loss in reconfigurable intelligent surface-enabled
  channels,'' in \emph{Proc. IEEE 32nd Annu. Int. Symp. Pers., Indoor, Mobile
  Radio Commun. (PIMRC)}, Oct. 2021.

\end{thebibliography}
\end{document}